\documentclass[twoside,leqno]{article}

\usepackage[letterpaper]{geometry}

\usepackage{siamproceedings}

 \usepackage[T1]{fontenc}
\usepackage{amsfonts}
\usepackage{graphicx}

 \usepackage{amsopn}

\usepackage{algorithm}
\usepackage[english]{babel}

\usepackage{thm-restate}
\usepackage{mathtools}
\usepackage{newtxtext, newtxmath}
\usepackage{enumitem}
\usepackage{titling}
\usepackage{epsfig}
\usepackage[utf8]{inputenc}
\usepackage[english]{babel}
\usepackage{CJKutf8}
\usepackage{algpseudocode}
\usepackage{tikz}
\usetikzlibrary{calc}
\usetikzlibrary{decorations.pathreplacing}
\usetikzlibrary{decorations.pathmorphing}
\usetikzlibrary{shapes.geometric}
\usetikzlibrary{fit}
\usetikzlibrary{arrows.meta}
\usetikzlibrary{arrows}

\renewcommand{\emph}[1]{\textit{#1}}
\renewcommand{\epsilon}{\varepsilon}

\newsiamthm{claim}{Claim}
\newsiamthm{problem}{Problem}
\newsiamremark{remark}{Remark}

\newtheorem*{dreamtheorem}{Dream Theorem}

\DeclareMathOperator{\lca}{lca}

\newcommand{\cov}{\mathrm{cov}}
\newcommand{\covr}{\overrightarrow{\mathrm{cov}}}
\newcommand{\covw}{\overleftarrow{\mathrm{cov}}}

\newcommand{\viw}{\mathrm{viwidth}}
\newcommand{\viwu}{\mathrm{viwidth}_{up}}
\newcommand{\viwd}{\mathrm{viwidth}_{down}}
\newcommand{\viwr}{\overrightarrow{\mathrm{viwidth}}}
\newcommand{\viww}{\overleftarrow{\mathrm{viwidth}}}

\begin{document}
\title{A Better-Than-2 Approximation for the Directed Tree Augmentation Problem}
\author{
  Meike Neuwohner\thanks{Department of Mathematics, London School of Economics and Political Science} \and
  Olha Silina\thanks{Department of Mathematics, Carnegie Mellon University} \and
  Michael Zlatin\thanks{Department of Computer Science, Pomona College}
}

\date{}

\fancyfoot[C]{\thepage}

\maketitle
\thispagestyle{fancy}

\begin{abstract}

We introduce and study a directed analogue of the weighted Tree Augmentation Problem (WTAP). In the weighted Directed Tree Augmentation Problem (WDTAP), we are given an oriented tree $T = (V,A)$ and a set of directed links $L \subseteq V \times V$ with positive costs. The goal is to select a minimum cost set of links which enters each fundamental dicut of $T$ (cuts with one leaving and no entering tree arc). 
WDTAP captures the problem of covering a cross-free set family with directed links. It can also be used to solve weighted multi $2$-TAP, in which we must cover the edges of an undirected tree at least twice.
WDTAP can be approximated to within a factor of $2$ using standard techniques. We provide an improved $(1.75+ \varepsilon)$-approximation algorithm for WDTAP in the case where the links have bounded costs, a setting that has received significant attention for WTAP. To obtain this result, we discover a class of instances, called ``willows'', for which the natural set covering LP is an integral formulation.  
We further introduce the notion of ``visibly $k$-wide'' instances which can be solved exactly using dynamic programming. Finally, we show how to leverage these tractable cases to obtain an improved approximation ratio via an elaborate structural analysis of the tree. 
\end{abstract}

\section{Introduction}
The design of networks that are resilient to connection failures constitutes a fundamental task in combinatorial optimization. An important class of network design problems are network \textit{augmentation} problems, in which we are given a graph and a set of additional edges, called \emph{links}, which we seek to add so that the resulting graph achieves the desired connectivity properties. One of the most well-studied network augmentation problems is the Tree Augmentation Problem (TAP). Given an undirected tree $T=(V,E)$ and a set of links $L \subseteq {V \choose 2}$, TAP asks for a minimum cardinality subset of the links whose addition renders $T$ $2$-edge-connected. A solution to TAP must include a link crossing each fundamental cut induced by the edges of $T$. The fundamental cuts form a cross-free family\footnote{A set family $\mathcal{C} \subseteq 2^V$ is \textit{cross-free} if for every $A,B\in\mathcal{C}$ with $A\cap B\ne \emptyset$ and $A\cup B\ne V$, we have $A\subseteq B$ or $B\subseteq A$.} and can be represented by a laminar set family. 

TAP can also be interpreted as a set covering problem on the edges of $T$, where a link $\ell=\{u,v\}$ covers every edge on the unique $u$-$v$-path in $T$. As TAP is known to NP-hard and APX-hard~\cite{DBLP:journals/siamcomp/KortsarzKL04}, there has been a long line of research on approximation algorithms for TAP~\cite{frederickson1981approximation,DBLP:journals/dam/Nagamochi03,DBLP:journals/algorithmica/CheriyanG18,DBLP:journals/algorithmica/CheriyanG18a,DBLP:journals/talg/EvenFKN09,DBLP:conf/stoc/CecchettoTZ21, DBLP:journals/talg/KortsarzN16}, starting with a $2$-approximation by Frederickson and J{\'a}J{\'a}~\cite{frederickson1981approximation} and culminating in the best known approximation ratio of $1.393$~\cite{DBLP:conf/stoc/CecchettoTZ21}.
There has further been a lot of research on the \emph{weighted} Tree Augmentation Problem (WTAP), where every link is equipped with a positive cost, and the task is to minimize the total cost of the selected link set. Until a few years ago, no better approximation ratio than $2$ was known. Recently, this approximation barrier has been breached, resulting in the best known approximation ratio of $1.5+\epsilon$~\cite{traub2022better,traub2022local}. Prior to this, several works have considered the bounded cost ratio case, where the ratio between the maximum and the minimum link cost can be bounded by a constant~\cite{adjiashvili2018beating,DBLP:conf/soda/Fiorini0KS18,grandoni2018improved}. In this setting, the best known approximation factor is $1.458$~\cite{grandoni2018improved}.

\paragraph{The Directed Tree Augmentation Problem} In this paper, we introduce a directed variant of the tree augmentation problem in which both the links and the underlying tree consist of directed arcs. Given an oriented tree $T=(V,A)$ and an arc $a=(u,v)$, we define the \textit{fundamental dicut associated with $a$} to be the vertex set $U$ of the weakly connected component of $T-a$ containing $u$. We say that a link $\ell=(x,y)$ \textit{covers} the dicut $U$ if $y\in U$, but $x\notin U$. In the Directed Tree Augmentation Problem (DTAP), we are given an oriented tree $T = (V,A)$ and a set of directed links $L \subseteq V \times V$, and the goal is to cover all fundamental dicuts of $T$ using a minimum cardinality subset of $L$. Note that the fundamental dicuts of $T$ form a cross-free set family. In fact, using a result of Edmonds and Giles on tree-representations of cross-free families~\cite{EDMONDS1977185}, DTAP captures the problem of covering an arbitrary cross-free family with directed links. DTAP can also be seen as a covering problem on the arcs of an oriented tree: a tree arc is \textit{covered} by a directed link $(u,v)$ if the unique path from $v$ to $u$ in $T$ contains the tree arc in the \textit{forward} direction. See \Cref{fig:DTAPexample} for an example.

%
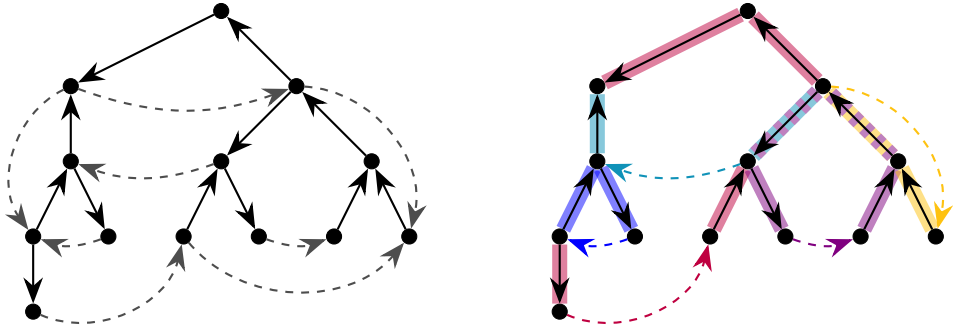
\begin{figure}[h]
\centering
\begin{tikzpicture}[mynode/.style={draw, fill, circle, minimum size = 2mm, inner sep = 0pt}, myarc/.style={thick, arrows = {-Stealth[scale=1.5]}},mylink/.style={thick, arrows = {-Stealth[scale=1.5]}, black!70!white, dashed}, covered/.style={line width = 2mm, draw opacity = 0.5}]
\node[mynode] (r) at (0,0){};
\node[mynode] (v1) at (-2,-1){};
\node[mynode] (v2) at (1,-1){};
\node[mynode] (v3) at (-2,-2){};
\node[mynode] (v4) at (0,-2){};
\node[mynode] (v5) at (2,-2){};
\node[mynode] (v6) at (-1.5,-3){};
\node[mynode] (v7) at (-2.5,-3){};
\node[mynode] (v8) at (-0.5,-3){};
\node[mynode] (v9) at (0.5,-3){};
\node[mynode] (v10) at (1.5,-3){};
\node[mynode] (v11) at (2.5,-3){};
\node[mynode] (v12) at (-2.5,-4){};

\draw[myarc] (r) to (v1);
\draw[myarc] (v2) to (r);
\draw[myarc] (v3) to (v1);
\draw[myarc] (v2) to (v4);
\draw[myarc] (v5) to (v2);
\draw[myarc] (v3) to (v6);
\draw[myarc] (v7) to (v3);
\draw[myarc] (v8) to (v4);
\draw[myarc] (v4) to (v9);
\draw[myarc] (v10) to (v5);
\draw[myarc] (v11) to (v5);
\draw[myarc] (v7) to (v12);

\draw[mylink] (v1) to[bend right = 20] (v2);
\draw[mylink] (v2) to [bend left= 50] (v11);
\draw[mylink] (v4) to[bend left = 20] (v3);
\draw[mylink] (v1) to [bend right= 50] (v7);
\draw[mylink] (v6) to [bend left= 20] (v7);
\draw[mylink] (v12) to [bend right= 50] (v8);
\draw[mylink] (v8) to [bend right= 50] (v11);
\draw[mylink] (v9) to [bend right= 20] (v10);

\begin{scope}[shift={(7,0)}]
\node[mynode] (r) at (0,0){};
\node[mynode] (v1) at (-2,-1){};
\node[mynode] (v2) at (1,-1){};
\node[mynode] (v3) at (-2,-2){};
\node[mynode] (v4) at (0,-2){};
\node[mynode] (v5) at (2,-2){};
\node[mynode] (v6) at (-1.5,-3){};
\node[mynode] (v7) at (-2.5,-3){};
\node[mynode] (v8) at (-0.5,-3){};
\node[mynode] (v9) at (0.5,-3){};
\node[mynode] (v10) at (1.5,-3){};
\node[mynode] (v11) at (2.5,-3){};
\node[mynode] (v12) at (-2.5,-4){};

\draw[covered,purple] (r) to (v1);
\draw[covered, purple] (v2) to (r);
\draw[covered, cyan!70!black] (v3) to (v1);
\draw[covered, cyan!70!black] (v2) to (v4);
\draw[line width = 2mm, dashed, violet!50!white] (v2) to (v4);
\draw[covered, yellow!50!orange] (v5) to (v2);
\draw[line width = 2mm, dashed, violet!50!white] (v5) to (v2);
\draw[covered, blue] (v3) to (v6);
\draw[covered, blue] (v7) to (v3);
\draw[covered, purple] (v8) to (v4);
\draw[covered, violet] (v4) to (v9);
\draw[covered,violet] (v10) to (v5);
\draw[covered, yellow!50!orange] (v11) to (v5);
\draw[covered, purple] (v7) to (v12);
\draw[myarc] (r) to (v1);
\draw[myarc] (v2) to (r);
\draw[myarc] (v3) to (v1);
\draw[myarc] (v2) to (v4);
\draw[myarc] (v5) to (v2);
\draw[myarc] (v3) to (v6);
\draw[myarc] (v7) to (v3);
\draw[myarc] (v8) to (v4);
\draw[myarc] (v4) to (v9);
\draw[myarc] (v10) to (v5);
\draw[myarc] (v11) to (v5);
\draw[myarc] (v7) to (v12);

\draw[mylink, yellow!50!orange] (v2) to [bend left= 50] (v11);
\draw[mylink, cyan!70!black] (v4) to[bend left = 20] (v3);
\draw[mylink, blue] (v6) to [bend left= 20] (v7);
\draw[mylink, purple] (v12) to [bend right= 50] (v8);
\draw[mylink, violet] (v9) to [bend right= 20] (v10);
\end{scope}
\end{tikzpicture}
\caption{A DTAP instance is shown on the left, with links drawn as dashed lines. A feasible solution is shown on the right. Colors indicate the tree arcs that are covered by each link.}
\label{fig:DTAPexample}
\end{figure}

 Interestingly, DTAP also captures some other natural set covering problems on the edges of a tree. Suppose $T = (V,E)$ is an undirected tree and a set of links is given. The (multi) 2-TAP problem is the problem of selecting a smallest multi-set of links (meaning that we are allowed to select the same link twice) so that each tree edge $e \in E$ is covered at least twice by the links we select. The multi 2-TAP problem reduces to DTAP, see \Cref{subsec:multi-2-tap}.

DTAP can be shown to be NP-hard as well as APX-hard, using similar reductions as in~\cite{frederickson1981approximation,DBLP:journals/siamcomp/KortsarzKL04} for Strong Connectivity Augmentation on oriented trees. For completeness, we give a hardness proof in \Cref{subsec:hardness}. 

We further define the \emph{weighted} Directed Tree Augmentation Problem (WDTAP), in which we are given a positive cost for each link, and the task is to minimize the total cost of the selected link set.
Like many problems in network design, WDTAP admits a straightforward $2$-approximation. To see this, note that WDTAP can be solved in polynomial time when the underlying tree $T$ is an arborescence. Indeed, in this case the constraint matrix of the natural integer programming formulation is a network matrix and hence totally unimodular. This tractable case can be leveraged to obtain a $2$-approximation in general. First, choose an arbitrary root vertex $r$. This partitions the arcs of the tree into \emph{up-arcs} pointing towards the root, and \emph{down-arcs} pointing away from the root. Contracting the up-arcs and down-arcs, respectively, yields two instances of WDTAP in which the oriented tree is a rooted arboresence. Thus, we can cover the up-arcs and down-arcs separately, paying at most the cost of an optimum solution each. The union of these two solutions yields a $2$-approximation.  
 
The main result of this paper is a better-than-2 approximation for WDTAP in the case where the \textit{cost ratio} of the instance, the ratio between the maximum and the minimum cost of a link, is bounded. 
\begin{restatable}{theorem}{mainresult}\label{theorem:main_result}
Let $\Delta\ge 1$ and let $\varepsilon > 0$. There exists a polynomial-time $(1.75+\varepsilon)$-approximation algorithm for WDTAP, restricted to instances with cost ratio at most $\Delta$.
\end{restatable}

\paragraph{Further related work}

Many network design problems exhibit a natural approximation barrier of $2$. This is in part due to a fundamental result of Jain~\cite{DBLP:conf/focs/Jain98} who gave a unified iterative rounding 2-approximation algorithm for the Survivable Network Design Problem, which captures, e.g., the weighted Tree and Connectivity Augmentation Problem and their Steiner variants~\cite{DBLP:conf/soda/0001ZZ23,DBLP:conf/esa/HathcockZ24}, the Steiner Forest Problem, and the (weighted) $k$-Edge-Connected Spanning Subgraph Problem. Jain's algorithm represented the best known approximation ratio for these problems for many decades, and only in recent years we have seen several breakthroughs. As mentioned, the best ratio for WTAP, and in fact, also the weighted Connectivity Augmentation Problem (WCAP) is now $1.5+\varepsilon$~\cite{traub2022local}, and recent exciting progress on the Steiner forest has shown that a better-than-2 approximation is possible there as well~\cite{ahmadi2025breakinglongstandingbarrier2varepsilon}. If the weights are uniform, both TAP and CAP can be approximated to within a factor of 1.393~\cite{DBLP:conf/stoc/CecchettoTZ21} while the 2-Edge-Connected Spanning Subgraph Problem admits a slightly better than 1.25-approximation~\cite{bosch20255,hommelsheim2025better}. The weighted $2$-Edge-Connected Spanning Subgraph Problem remains at a factor of $2$, even in the bounded cost setting. 

Directed network design problems are often more challenging to approximate. The Directed Steiner Tree Problem is already set cover hard, see e.g.~\cite{DBLP:conf/soda/CharikarCCDGGL98}. In terms of augmentation problems, the Strong Connectivity Augmentation Problem (SCAP) is natural: we are given a weakly connected digraph which we seek to make strongly connected by adding directed links of cheapest cost. 
SCAP has seen little progress since Frederickson and J\'aj\'a proved that it admits a 2-approximation in the 1980s~\cite{frederickson1981approximation}. It has been shown to be fixed parameter tractable with respect to the solution size~\cite{DBLP:conf/soda/KlinkbyMS21}. Polyhedral results due to Schrijver~\cite{schrijver1982min} show that the linear program of finding a minimum cost strong augmentation of a given digraph $D$ is integral when $D$ is source-sink connected, or when the available arcs are exactly the reverse arcs of those in $D$. There are no known better-than-2 approximations for strongly connecting an oriented tree, even in the unweighted case.

    

\paragraph{Organization of the paper}
The remainder of this paper is organized as follows. In \Cref{sec:prelims}, we formally define WDTAP and introduce basic notation which will be used throughout the paper. In \Cref{sec:pastwork}, we briefly describe past works for WTAP and how these approaches fail in our setting. In \Cref{sec:ourcontribution}, we give an overview of our techniques to prove \Cref{theorem:main_result}. In \Cref{sec:willows}, we define the class of ``willows'' for which the standard LP relaxation turns out to be integral. In \Cref{sec:DPvisiblewidth}, we characterize instances that can be solved in polynomial time via a standard dynamic programming approach. \Cref{sec:partialseparation,sec:proof_weak_dream,sec:components_and_cores,sec:best_of_three} present our $(1.75+\epsilon)$-approximation for WDTAP with bounded cost ratio and its analysis.  

\section{Preliminaries}\label{sec:prelims}
An instance of WDTAP consists of a directed tree $T = (V,A)$, and directed links $L \subseteq V \times V$ with positive costs $c: L \to \mathbb{R}_{> 0}$. For $\ell=(u,v)\in L$, we denote the unique $u$-$v$-path in $T$ by $P_\ell$. The \emph{(directed) coverage} $\covr(\ell)$ consists of all backward arcs on $P_\ell$. To be consistent with the literature on (undirected) tree augmentation, we use $\cov(\ell)$ to denote the set of all arcs on $P_\ell$. However, we point out that in the context of WDTAP, a link $\ell$ only covers the arcs in $\covr(\ell)$, as opposed to all arcs in $\cov(\ell)$. Finally, we write $\covw(\ell)\coloneqq \cov(\ell)\setminus\covr(\ell)$ to denote the set of arcs on $P_\ell$ that $\ell$ ``covers in the wrong direction''. A set of links $F \subseteq L$ is \emph{feasible} if every tree arc in $A$ is covered by some link in $F$, i.e., $A\subseteq \bigcup_{\ell\in F}\covr(\ell)$. WDTAP asks for a feasible link set of minimum cost $c(F):=\sum_{\ell \in F} c(\ell)$.
For convenience, given an instance of $(T=(V,A), L, c)$ of WDTAP, we will fix a vertex $r\in V$ and call it the \emph{root}. We will call the tuple $(T,L,c,r)$ a \emph{rooted instance} of WDTAP. For $v\in V$, we write $T_v=(U_v,A_v)$ to denote the \emph{subtree rooted at $v$}. We call arcs that are pointing towards/away from the root \emph{up-arcs} and \emph{down-arcs}, respectively, and we write $A = A_{up} \dot \cup A_{down}$ to denote the partition into up- and down-arcs. For a link $\ell = (u,v)$, we define the \emph{apex} of $\ell$ to be $\mathrm{apex}(\ell)\coloneqq \lca(u,v)$ (where $\lca(u,v)$ denotes the least common ancestor of $u$ and $v$, the vertex on $P_\ell$ closest to the root). Note that a link $\ell = (u,v)$ \emph{covers} the up-arcs along the $v$-$\mathrm{apex}(\ell)$-path in $T$, and the down-arcs along the $u$-$\mathrm{apex}(\ell)$-path in $T$.
We call a link of the form $\ell=(u,v)$ with $v=\mathrm{apex}(\ell)$ ($u=\mathrm{apex}(\ell)$) an up-link (down-link). 
A \emph{shadow} of a link $\ell=(u,v)$ is a link of the form $\ell' = (u',v')$, where $u'$ and $v'$ appear in this order on $P_\ell$.
We may assume without loss of generality that the WDTAP instances we are working with are \emph{shadow-complete}: this means that for every $\ell\in L$, $L$ contains every possible shadow $\ell'$ of $\ell$ and moreover, $c(\ell')\le c(\ell)$. 
Note that unlike the undirected setting where a link $\ell$ covers all edges on the tree path connecting its endpoints, and, in particular, every shadow has a strictly smaller coverage, this is no longer true in the directed case. For this reason, we define the \emph{generic shadow} $s(\ell)$ of a link $\ell$ as the minimal shadow of $\ell$ with $\covr(s(\ell))=\covr(\ell)$ and let $\overline{P}_\ell\coloneqq P_{s(\ell)}$. Given a feasible solution to a WDTAP instance, we may always replace each link by its generic shadow without increasing costs or destroying feasibility.

The \emph{cost ratio} of a WDTAP instance is defined as $\Delta=\frac{\max_{\ell\in L} c(\ell)}{ \min_{\ell\in L} c(\ell)}$.

The \emph{arc-link-coverage matrix} of an instance of WDTAP is the matrix $M\in\{0,1\}^{A\times L}$ with $M_{a,\ell}=1$ if and only if $a\in\covr(\ell)$. For subsets $B\subseteq A$ and $L'\subseteq L$, we denote by $M[B,L']$ the submatrix of $M$ with rows indexed by $B$ and columns indexed by $L'$. Obtaining an optimum solution to WDTAP is equivalent to finding an optimum \emph{integral} solution to the linear program \eqref{eq:WDTAP_LP}.
\begin{equation}\min\left\{\sum_{\ell\in L} c(\ell)\cdot x_\ell \colon M\cdot x\ge \mathbf{1},x\ge 0\right\}.\label{eq:WDTAP_LP}\end{equation}

\section{Comparison with previous work}\label{sec:pastwork}
\subsection{A decomposition-based approach for WTAP with bounded cost ratio $\dots$\label{section:decomposition_WTAP}}
 In order to provide intuition for our algorithm, it is instructive to first describe some of the ideas that are used in \cite{adjiashvili2018beating,DBLP:conf/soda/Fiorini0KS18,grandoni2018improved} to obtain better-than-$2$-approximations for WTAP with bounded cost ratio\footnote{In doing so, we will mostly follow the description in \cite{grandoni2018improved}, but provide a slightly modified perspective on certain arguments to make them align better with the remainder of this paper.}. The overall approach pursued in these works consists of two main steps:


\begin{enumerate}
 	\item Design an $\alpha$-approximation algorithm for a certain class of well-structured instances.

    \item Decompose a general instance into subinstances from this class such that $\alpha$-approximate solutions to the subinstances can be combined to an $(\alpha+\epsilon)$-approximate solution to the original instance.
 \end{enumerate}


\cite{grandoni2018improved} consider the class of so-called \emph{$k$-wide instances}, where $k\in\mathbb{N}$ is a constant. A (rooted) instance of WTAP is called $k$-wide if every subtree of a child of the root contains at most $k$ leaves. A $1.5$-approximation for $k$-wide instances can be obtained by trading off two different algorithms: after splitting all cross-links\footnote{In the undirected setting, a link is called an \emph{up-link} if its apex coincides with one of its endpoints. A \emph{cross-link} is a link whose apex is the root that is not an up-link. All links that are neither up- nor cross-links are called \emph{in-links}.} at their apex, a $k$-wide instance decomposes into a union of independent instances with at most $k$ leaves each. These can be solved exactly using dynamic programming~\cite{grandoni2018improved}. On the other hand, after splitting every in-link into two up-links, the natural LP relaxation becomes integral after adding so-called \emph{odd cut constraints}~\cite{DBLP:conf/soda/Fiorini0KS18}.

The decomposition into $k$-wide instances is guided by a solution $x$ to a linear programming relaxation, which is used to estimate the cost of \emph{splitting links} to cut off subtrees as independent subinstances. More precisely, when saying that we obtain the LP solution $x'$ from the LP solution $x$ by splitting a set of links $L'$ at a vertex $v$, we mean that $x'$ arises from $x$ by, for every $\ell=\{u,w\}\in L'$ such that $v$ is an inner vertex of $P_\ell$, setting $x'(\ell)=0$ and increasing the $x'$-value on each of the two shadows $\{u,v\}$ and $\{v,w\}$ by $x(\ell)$.
The constant cost ratio ensures that up to a constant factor, costs are proportional to $x$-values; if the total $x$-value only increases by an $\epsilon$-fraction, then the cost will also only increase by an $\mathcal{O}(\epsilon)$-fraction.
The first step of the decomposition procedure is to contract edges $e\in E(T)$ that are \emph{heavily covered}~\cite{adjiashvili2018beating}, i.e., for which $x(\{\ell\colon e\in\cov(\ell)\})\ge \zeta_\epsilon\coloneqq \frac{2}{\epsilon}$. These edges can be covered at cost $\epsilon\cdot c(x)$~\cite{adjiashvili2018beating}.
In the second step of the decomposition procedure, the tree is traversed from bottom to top and subtrees hanging off certain inner vertices are split off. 
More precisely, an edge $e=\{v,w\}$, where $w$ is closer to the root, is called \emph{$\epsilon$-light}\footnote{This term was introduced in \cite{grandoni2018improved}. Their definition (slightly) differs from the one presented here because they do not work with a rooted tree.} if $x(\{\ell\colon e\in\cov(\ell)\})\le \epsilon\cdot x(\{\ell\colon \cov(\ell)\cap E(T_v)\ne \emptyset\})$, i.e., if the total coverage of $e$ amounts to at most an $\epsilon$-fraction of the total coverage of $T_v$.
When disattaching the subtree $T_v$, every link $\ell=\{u,w\}$ that leaves $T_v$ is split at $v$ into two shadows, one whose coverage is contained in $T_v$ and one whose coverage is contained in $T-T_v$. As every split link covers $e$, the splitting increases the total $x$-value by at most $x(\{\ell\colon e\in\cov(\ell)\})$. These costs can be charged to the total $x$-value within the subtree $T_v$ that is subsequently cut off, ensuring that the iterated splitting only increases the total cost by an $\mathcal{O}(\epsilon)$-fraction. 
At the end of the splitting, every subinstance is $k_\epsilon\coloneqq 2\cdot \epsilon^{-1}\cdot \zeta_\epsilon$-wide. To see this, let $w$ be the root of a subinstance and let $T_v$ be a subtree hanging off the root. If there are more than $k_\epsilon$ leaves in the subtree, then $x(\{\ell\colon \cov(\ell)\cap E(T_v)\ne \emptyset\})> \frac{k_\epsilon}{2}$ because every link can cover the edges incident to at most two leaves. On the other hand, $x(\{\ell\colon \{v,w\}\in\cov(\ell)\})< \zeta_\epsilon$ because all heavily covered edges were contracted, implying that $\{v,w\}$ is $\epsilon$-light. But this means that $T_v$ would have been cut off, a contradiction.

\subsection{$\dots$ and why it doesn't work for DTAP}\label{section:decomposition_does_not_work}
Given the similarities between WTAP and WDTAP, it appears tempting to transfer the ideas described in~\Cref{section:decomposition_WTAP} from the undirected to the directed setting. In this section, we point out the issues with this approach, before discussing how to resolve them in the following sections. Recall that in order to decompose a WTAP instance (with bounded cost ratio) into $k$-wide subinstances, while only increasing the total cost by an $\mathcal{O}(\epsilon)$-fraction, we had to ensure two properties: 
first of all, we argued that subtrees with many leaves attract a large LP value. More precisely, we observed that whenever a subtree $T_v$ hanging off the edge $e=\{v,w\}$ contains more than $k_\epsilon$ leaves, then $x(\{\ell\colon \cov(\ell)\cap E(T_v)\ne \emptyset\})$, the total LP value on links covering edges within the subtree, is large. Second of all, we had to make sure that $x(\{\ell\colon e\in\cov(\ell)\})$, the total LP value of links covering the edge $e$, is comparably small. Note that the links covering $e$ are precisely the links one has to duplicate to split off the subtree $T_v$ as an independent instance. Hence, combining both properties ensures that the total cost increase incurred by splitting links only amounts to an $\mathcal{O}(\epsilon)$-fraction of the optimum cost.

The first property can also be easily ensured in the directed setting. Let $a=(v,w)$ be an up-arc (down-arcs can be handled analogously) and assume that $T_v$ contains more than $k$ leaves. Again, each link $\ell$ can cover the arcs incident to at most two leaves; more precisely, $\covr(\ell)$ contains at most one up- and at most one down-arc incident to a leaf.
Hence, any solution $x$ to the natural LP relaxation \eqref{eq:WDTAP_LP} will satisfy $x(\{\ell\colon \covr(\ell)\cap A(T_v)\ne \emptyset\})\ge \frac{k}{2}$. A problem arises, however, with the second condition. In the undirected setting, every link $\ell$ with $e\in E(P_\ell)$ covers the edge $e$. This property can be used to ensure that $x(\{\ell\colon e\in \cov(\ell)\})$ is not too large by covering all heavily covered edges at an $\mathcal{O}(\epsilon)$-fraction of the LP cost and then contracting them. In contrast, in the directed setting, this reasoning can only be used to guarantee that no arc is \emph{heavily covered in the right direction}, i.e., that $x(\{\ell\colon a\in\covr(\ell)\})$ is not too large. However, we cannot control whether an arc is \emph{heavily covered in the wrong direction}, i.e., whether $x(\{\ell\colon a\in\covw(\ell)\})$ is large. But in order to fully cut off $T_v$, we need to be able to bound
$x(\{\ell\colon a\in \cov(\ell)\})=x(\{\ell\colon a\in\covr(\ell)\})+x(\{\ell\colon a\in\covw(\ell)\})$. \cref{figure:no_splitting_possible} illustrates an example where we cannot cut off a subtree, even though it contains more than $k$ leaves.

The reader might wonder why we are focusing on the techniques used in \cite{adjiashvili2018beating,DBLP:conf/soda/Fiorini0KS18,grandoni2018improved} that yield better-than-$2$-approximations for WTAP with bounded cost ratio, and do not consider the more recent works \cite{traub2022better,traub2022local} that give better-than-$2$-approximations for (general) WTAP. The reason for this is that the techniques in \cite{traub2022better,traub2022local} crucially rely on certain structural properties of undirected solutions, such as the fact that the coverage of a link is connected. As such, it is unclear how to apply them in the directed setting.
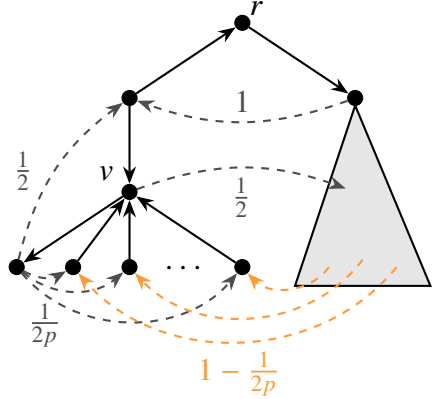
\begin{figure} 
	\centering
	\begin{tikzpicture}[scale = 1, mynode/.style={draw, fill, circle, minimum size = 2mm, inner sep = 0pt}, myarc/.style={thick, arrows = {-Stealth[scale=1]}},mylink/.style={thick, arrows = {-Stealth[scale=1]}, black!70!white, dashed}, every node/.style={font = \large}]
		
		\node[mynode] (r)   at ( 0   ,  0)  {};
		\node           at ( 0.2 ,  0.2) {$r$};
		
		\node[mynode] (v1) at (-1.5,-1){};
		\node           at ( -1.8 ,  -2) {$v$};
		\node[mynode] (v2) at (1.5,-1){};
		\node (w3) at (1.5,-2.25){};
		
		\node[mynode] (v3) at (-1.5,-2.25){};
		
		\node[mynode] (v4) at (-3,-3.25){};
		\node[mynode] (v5) at (-2.25,-3.25){};
		\node[mynode] (v6) at (-1.5,-3.25){};
		\node (v7) at (-0.75,-3.25){$\dots$};
		\node[mynode] (v8) at (0,-3.25){};
		
		\draw[myarc] (v1) to (r);
		\draw[myarc] (r) to (v2);
		\draw[myarc] (v1) to (v3);
		\draw[myarc] (v3) to (v4);
		\draw[myarc] (v5) to (v3);
		\draw[myarc] (v6) to (v3);
		\draw[myarc] (v8) to (v3);
		
		\coordinate (v2T) at ($(v2)+(0,-0.1)$);   
		\coordinate (v2L) at ($(v2)+(-0.8,-2.5)$);   
		\coordinate (v2R) at ($(v2)+( 1,-2.5)$);   
		
		\draw[thick, fill=gray!20]       
		(v2T)  -- (v2L) -- (v2R) -- cycle;
		
		\coordinate (subtree1) at ($(v2)+(-0.35,-2.25)$); 
		\coordinate (subtree2) at ($(v2)+(0.1,-2.15)$);
		\coordinate (subtree3) at ($(v2)+(0.55,-2.25)$);
		
		\draw[mylink] (v2) to [bend left = 20] node [midway, above]{$1$} (v1);
		\draw[mylink] (v3) to [bend left = 20] node [midway, below]{$\frac{1}{2}$}(w3);
		\draw[mylink] (v4) to [bend left = 20] node [midway, left=5pt]{$\frac{1}{2}$} (v1);
		\draw[mylink] (v4) to [bend right = 30] node [midway, below=6pt]{$\frac{1}{2p}$} (v5);
		\draw[mylink] (v4) to [bend right = 40] (v6);
		\draw[mylink] (v4) to [bend right = 50] (v8);
		\draw[mylink,orange!80!white] (subtree3) to [bend right = -50] node [midway, below]{$1-\frac{1}{2p}$} (v5);
		\draw[mylink,orange!80!white] (subtree2) to [bend right = -50] (v6);
		\draw[mylink,orange!80!white] (subtree1) to [bend right = -50] (v8);
		
	\end{tikzpicture}
	\caption{An instance of WDTAP, with arcs indicated by solid arrows and links shown as dashed arrows. A solution to \eqref{eq:WDTAP_LP} is indicated next to the links. All links have cost $1$. The subtree hanging off the vertex $v$ contains a large number $p$ of leaves that are connected to $v$ via up-arcs. In the given LP solution, these up-arcs are (partially) covered by the orange links, each of which covers the down-arc entering $v$ in the wrong direction. Splitting all of the orange links at $v$ is too expensive.}\label{figure:no_splitting_possible}
\end{figure}
\section{Our contribution}\label{sec:ourcontribution}
The main result of this paper is a polynomial-time better-than-$2$-approximation for WDTAP with bounded cost ratio.
\mainresult*
Our approach towards \cref{theorem:main_result} is inspired by the decomposition strategy pursued in~\cite{adjiashvili2018beating,DBLP:conf/soda/Fiorini0KS18,grandoni2018improved} to obtain better-than-$2$-approximations for WTAP with bounded cost ratio. However, as we discussed in the previous section, WDTAP exhibits fundamental differences to its undirected analogue, which renders a simple adaptation of prior techniques infeasible. Instead of completely decomposing the instance, our approach relies on \emph{partially decomposing} the instance by splitting certain links in such a way that the total cost only increases by an $\mathcal{O}(\epsilon)$-fraction.
Our main technical contributions can be summarized as follows:
\begin{itemize}
    \item We introduce the concept of \emph{bounded visible width}, which allows us to characterize when an instance of WDTAP can be solved in polynomial time using a standard dynamic programming approach.
    \item We define a new class of instances called \emph{willows} and show that for these instances, the natural LP relaxation \eqref{eq:WDTAP_LP} is integral. 
    \item We discuss how to carefully split certain links to achieve strong structural properties, while only incurring an arbitrary small increase in the total solution cost.
    Then, we explain how splitting certain subsets of the links results in instances of bounded visible width and willows, respectively, both of which we can solve exactly. Trading off three different solutions results in the final approximation guarantee.
\end{itemize}

\subsection{New notions of partial decomposition: visible width and willows}
In this section, we introduce the concepts of visible width and up- and down-independence, which allow us to characterize the structural properties that we gain by splitting links, and to leverage these to solve certain types of instances exactly.
\paragraph{\textbf{Visible width}}
The notion of constant visible width characterizes WDTAP instances that can be solved exactly using a natural dynamic programming approach, similar to the one used for $k$-wide instances~\cite{grandoni2018improved}. The basic idea is to traverse the tree from the leaves to the root and to iteratively construct a cheapest solution covering the subtree $T_v$, given a fixed ``interface'' to the remaining instance. More precisely, for $v\in V$, let $L_v$ be the set of links $\ell$ such that $v$ is an inner vertex of $\overline{P}_\ell$.
We further partition $L_v$ into the set $L^{cross}_v$ of \emph{$v$-cross-links} having $v$ as their apex, the set $L^\downarrow_v\coloneqq \{\ell=(u,w)\in L_v\colon u\notin U_v \wedge w\in U_v\setminus \{v\}\}$ of links \emph{pointing into} and the set $L^\uparrow_v\coloneqq \{\ell=(u,w)\in L_v\colon u\in U_v\setminus \{v\} \wedge w\notin U_v\}$ of links \emph{pointing out of} $T_v$. Note that the links in $L^{cross}_v$ are precisely those that connect the solutions in different subtrees of $T_v$ hanging off $v$, while the links in $L^\uparrow_v\cup L^\downarrow_v$ are precisely those creating interactions between $T_v$ and $T-T_v$.
In order to obtain an optimum WDTAP solution, it suffices to, for $v\in V$ and $F\subseteq L^\uparrow_v\cup L_v^\downarrow$, compute a cheapest solution $S(v,F)\subseteq L$ for $(T_v,L,c)$ with the property that $S(v,F)\cap (L^\uparrow_v\cup L^\downarrow_v)=F$. Then $S(r,\emptyset)$ constitutes an optimum solution to $(T,L,c)$. Of course, the problem with this approach is that in general, there are exponentially many choices for the set $F$ and moreover, we need to be able to control the number of $v$-cross-links used in a solution in order to merge solutions for the children of $v$ into a solution for $v$ efficiently. Hence, our goal is to characterize instances for which we can guarantee the existence of an optimum solution $S$ such that $|S\cap L_v|$ can be bounded by a constant for every $v\in V$; such a solution can be found in polynomial time using the above-mentioned DP approach.
To this end, consider a \emph{shadow-minimal} optimum solution $S^*$, i.e., an optimum solution in which no link can be replaced by a proper shadow whilst maintaining feasibility. It is not hard to see that every $\ell\in S^*$ covers the first and the last arc of $P_\ell$ and moreover, it is the unique link in $S^*$ that does so. This implies that the lowest up-arcs in $T_v$ covered by links in $(L^\downarrow_v\cup L^{cross}_v)\cap S^*$ are pairwise distinct and form an ancestor-free set of up-arcs, i.e., for two arcs $a\ne a'$ in this set, $a'$ does not appear on the path connecting the top vertex of $a$ to the root. Similarly, the lowest down-arcs in $T_v$ covered by links $(L^\uparrow_v\cup L^{cross}_v)\cap S^*$ form an ancestor-free set of down-arcs.
Motivated by these observations, we say that a vertex $v$ can \emph{see} an up-arc (a down-arc) $a\in A_v$ if there exists a link $\ell\in L^\downarrow_v\cup L^{cross}_v$ ($\ell\in L^\uparrow_v\cup L^{cross}_v$) with $a\in \covr(\ell)$. Equivalently, $v$ can see an arc $a\in A_v$ if there is a link $\ell$ with $a\in\covr(\ell)$ (i.e., $\ell$ covers $a$) such that $v$ is an \emph{inner vertex} of $\overline{P}_\ell$. 
We define the \emph{visible up-width (visible down-width)} at $v$ to be the maximum size of an ancestor-free set of up-arcs (down-arcs) that $v$ can see. The \emph{visible width} of an instance is the maximum over the visible up- and down-widths at the vertices. The previous considerations imply that WDTAP instances with constant visible width can be solved exactly in polynomial time via dynamic programming. 

\usetikzlibrary{calc}
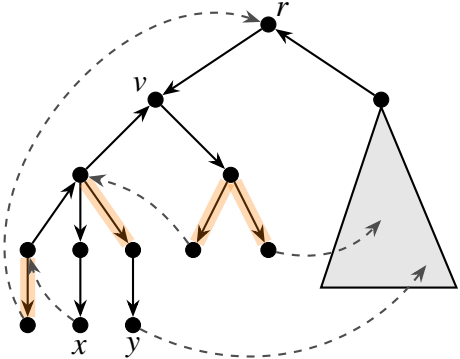
\begin{figure}[h]
\centering
\begin{tikzpicture}[scale = 1, mynode/.style={draw, fill, circle, minimum size = 2mm, inner sep = 0pt}, myarc/.style={thick, arrows = {-Stealth[scale=1]}},mylink/.style={thick, arrows = {-Stealth[scale=1]}, black!70!white, dashed}, every node/.style={font = \large},covered/.style={line width = 2mm, draw opacity = 0.3}]

\node[mynode] (r)   at ( 0   ,  0)  {};
\node           at ( 0.2 ,  0.2) {$r$};

\node[mynode] (v1)  at (-1.5 , -1)  {};
\node           at ( -1.7 ,  -.8) {$v$};

\node[mynode] (v2)  at ( 1.5 , -1)  {};

\node[mynode] (v3)  at (-2.5 , -2)  {};
\node[mynode] (v4)  at (-0.5 , -2)  {};

\node[mynode] (v6)  at (-2.5 , -3)  {};
\node[mynode] (v7)  at (-3.2 , -3)  {};
\node[mynode] (v8)  at (-1.8 , -3)  {};
\node[mynode] (v9)  at (-1   , -3)  {};
\node[mynode] (v11) at ( 0   , -3)  {};

\node[mynode] (v12) at (-3.2 , -4)  {};
\node[mynode] (v13) at (-2.5 , -4)  {};
\node[mynode] (v14) at (-1.8 , -4)  {};














\node           at ( -2.5 ,  -4.3) {$x$};

\node           at ( -1.8 ,  -4.3) {$y$};

\draw[myarc] (r) to (v1);
\draw[myarc] (v2) to (r);
\draw[myarc] (v3) to (v1);
\draw[myarc] (v3) to (v6);
\draw[myarc] (v7) to (v3);
\draw[myarc] (v1) to (v4);

\draw[myarc] (v3) to (v8);

\draw[myarc] (v4) to (v9);

\draw[myarc] (v4) to (v11);

\draw[myarc] (v7) to (v12);
\draw[myarc] (v6) to (v13);
\draw[myarc] (v8) to (v14);

\draw[covered, orange] (v12) to (v7);
\draw[covered, orange] (v4) to (v11);
\draw[covered, orange] (v3) to (v8);
\draw[covered, orange] (v4) to (v9);

\coordinate (v2T) at ($(v2)+(0,-0.1)$);   
\coordinate (v2L) at ($(v2)+(-0.8,-2.5)$);   
\coordinate (v2R) at ($(v2)+( 1,-2.5)$);   

\draw[thick, fill=gray!20]       
      (v2T)  -- (v2L) -- (v2R) -- cycle;

\coordinate (subtree1) at ($(v2)+(0.6,-2.2)$);   

\coordinate (subtree2) at ($(v2)+(-0.3,-2.2)$);   

\coordinate (subtree3) at ($(v2)+(0,-1.6)$);   

\draw[mylink] (v12) to [bend left = 70] (r);
\draw[mylink] (v14) to [bend right = 40] (subtree1);
\draw[mylink] (v11) to [bend right = 30] (subtree3);
\draw[mylink] (v9) to [bend right = 20] (v3);
\draw[mylink] (v13) to [bend left = 20] (v7);


\end{tikzpicture}
\caption{The visible width at $v$ is at least 4, as certified by the four orange-shaded arcs in its subtree. This is an ancestor-free set of down-arcs which are all visible to $v$. However, the black down-arc incident to $x$, as well as the down-arc above, are not visible to $v$ since they are not covered by a link $\ell$ for which $v$ is an inner vertex of $\overline{P}_\ell$. The down-arc incident to $y$ is visible to $v$, but does not form an ancestor-free set with the shaded down-arc above.}
\label{fig:visible-wide}
\end{figure}

Note that the visible width of an instance depends both on the rooted tree as well as the set of links that we are allowed to select. During the course of our algorithm, we will maintain a solution $x$ to \eqref{eq:WDTAP_LP} and we will always compute the visible width with respect to the \emph{support of $x$} (and its shadows). In particular, modifying $x$ by splitting links can block a vertex from seeing certain arcs in its subtree and decrease the visible width.
\paragraph{\textbf{Willows}} For WTAP, \cite{DBLP:conf/soda/Fiorini0KS18} have shown that if the instance only contains cross-links and up-links, then the constraint matrix of the natural LP relaxation is a binet matrix, which they use to argue that adding odd cut constraints suffices to guarantee integrality. For WDTAP, it is not hard to see that if the instance only contains cross-links and up- and down-links, then the constraint matrix of \eqref{eq:WDTAP_LP} is totally unimodular. We generalize this result by introducing \emph{willows}, a class of instances that may contain cross-links with respect to multiple ``local roots''.

Let $(T,L,c,r)$ be a rooted WDTAP instance. We call a vertex $v\in V(T)$ \emph{up-independent (down-independent)} if $L_v^\downarrow=\emptyset$ ($L_v^\uparrow=\emptyset$). If $v$ is up-independent (down-independent), then the problem of covering the up-arcs (down-arcs) in $T_v$ is ``independent from'' the problem of covering arcs outside $T_v$ in the sense that no link can cover both. We call $(T,L,c,r)$ a \emph{willow} if there exists a set $W\subseteq V(T)$ such that every vertex in $W$ is up- or down-independent, and every link in $L$ is either an up-link, a down-link, or a \emph{$W$-cross-link}, meaning that its apex is contained in $W$. Note that the root $r$ is always both up- and down-independent.
\begin{theorem}\label{theorem:willow_integral}
Let $(T,L,c,r)$ be a willow. Then \eqref{eq:WDTAP_LP} is integral.
\end{theorem}

\begin{figure}[h]
\centering
\begin{tikzpicture}[scale = 1, mynode/.style={draw, fill, circle, minimum size = 2mm, inner sep = 0pt}, myarc/.style={thick, arrows = {-Stealth[scale=1]}},mylink/.style={thick, arrows = {-Stealth[scale=1]}, black!70!white, dashed}, every node/.style={font = \large}]

\tikzset{
  halo/.style={                        
    append after command={             
      \pgfextra{
        \draw[violet!70!black, ultra thick, opacity = 0.25]   
              (\tikzlastnode.center) circle[radius=3.5mm];
        \fill[red!20, opacity=.25]          
              (\tikzlastnode.center) circle[radius=3.5mm];
      }
    }
  }
}

\node[mynode, halo] (r)  at ( 0  ,  0) {};
\node           at ( 0.2,  0.5) {$r$};

\node[mynode] (v1p)  at (-1  , -0.5) {};
\node[mynode] (v1)  at (-2  , -1) {};
\node[mynode] (v2)  at ( 1  , -1) {};

\node[mynode, halo] (v3)  at (-2  , -2) {};
\node           at ( -2.5,  -1.9) {$u$};

\node[mynode] (v4)  at ( 0.5, -2) {};
\node[mynode, halo] (v5)  at ( 2  , -2) {};
\node           at ( 2.5,  -1.9) {$v$};

\node[mynode] (v6)  at (-1  , -3) {};
\node[mynode] (v7)  at (-3  , -3) {};
\node[mynode] (v10) at ( 2  , -3) {};
\node[mynode] (v11) at ( 3  , -3) {};

\node[mynode] (v12) at (-3  , -4) {};
\node[mynode] (v13) at (-2  , -3) {};
\node[mynode] (v14) at ( 1  , -3) {};
\node[mynode] (v15) at ( 0.5, -4) {};
\node[mynode] (v16) at ( 1.5, -4) {};
















\draw[myarc] (v1p) to (r);
\draw[myarc] (v1) to (v1p);
\draw[myarc] (r) to (v2);
\draw[myarc] (v1) to (v3);
\draw[myarc] (v4) to (v2);
\draw[myarc] (v2) to (v5);
\draw[myarc] (v6) to (v3);
\draw[myarc] (v7) to (v3);
\draw[myarc] (v5) to (v10);
\draw[myarc] (v5) to (v11);
\draw[myarc] (v12) to (v7);
\draw[myarc] (v3) to (v13);
\draw[myarc] (v14) to (v5);
\draw[myarc] (v14) to (v15);
\draw[myarc] (v14) to (v16);

\draw[mylink] (v1) to[bend right = 70] (v12);

\draw[mylink] (v11) to[bend right = 40] (r);
\draw[mylink] (v3) to [bend left= -20] (v4);
\draw[mylink] (v13) to [bend left= 20] (v7);
\draw[mylink] (v15) to [bend left= 20] (v6);
\draw[mylink] (v10) to [bend left= 20] (v14);
\draw[mylink] (v16) to [bend left= 70] (v13);
\draw[mylink] (v13) to [bend left= -20] (v6);

\end{tikzpicture}
\caption{A willow  (choosing $W = \{r, u, v\}$). Notice that $u$ is down-independent, $v$ is up-independent, and the root $r$ is both. All links are either up-links, down-links, or have their apex in $W$.}
\label{fig:willow}
\end{figure}
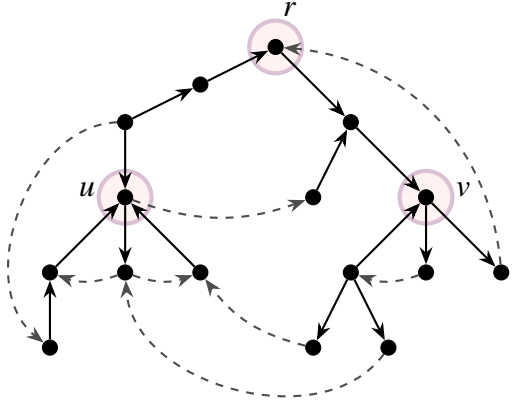

\subsection{Our approach}\label{section:algorithm}
\paragraph{Blue-sky version}
To provide some intuition how the notions of visible width and willows can be leveraged towards a better-than-$2$-approximation for WDTAP with bounded cost ratio, assume for a moment that we could prove the following ``dream theorem''.
\begin{dreamtheorem}
Let $\epsilon>0$ and $x$ a solution to \eqref{eq:WDTAP_LP}. We can, in polynomial time, compute a solution $x^*$ to \eqref{eq:WDTAP_LP} of cost $c(x^*)\le (1+\epsilon)\cdot c(x)$ that arises from $x$ by splitting links, and a set $W\subseteq V$ such that:
\begin{enumerate}[label=(\roman*)]
	\item $W$ consists of up- and down-independent vertices with respect to $\mathrm{supp}(x^*)=\{\ell\in L\colon x^*(\ell)>0\}$.
	\item Let $L'$ arise from $\mathrm{supp}(x^*)$ by splitting every $W$-cross-link at its apex. Then $(T,L')$  has visible width at most $k(\epsilon,\Delta)$ (where $k(\epsilon,\Delta)$ is some constant depending on $\epsilon$ and the cost ratio $\Delta$ of the instance).
\end{enumerate}
\end{dreamtheorem}
Using the dream theorem, we could obtain a $(1.5+\mathcal{O}(\epsilon))$-approximation for WDTAP with bounded cost ratio as follows: first apply the dream theorem to compute $x^*$ and $W$ subject to (i)-(ii). Let $L^*_{cross}$ denote the set of $W$-cross-links in $\mathrm{supp}(x^*)$. By splitting all links in $\mathrm{supp}(x^*)\setminus L^*_{cross}$ that are neither up- nor down-links at their apex, we obtain a willow. Hence, we can compute a solution of cost at most $c(x^*)+\sum_{\ell\in L^*_{cross}} c(\ell)\cdot x^*(\ell)$ by \cref{theorem:willow_integral}. By splitting all links in $L^*_{cross}$ at their apex, we obtain an instance of visible width at most $k(\epsilon,\Delta)$ by (ii), which we can solve optimally. If we could argue that this solution costs at most $c(x^*)+\sum_{\ell\in L\setminus L^*_{cross}} c(\ell)\cdot x^*(\ell)$, then taking the best of the two solutions gives a solution of cost at most $1.5\cdot c(x^*)\le 1.5\cdot (1+\epsilon)\cdot c(x)$. There is the slight issue that the optimum solution found by the DP might be more expensive than the LP ``suggests''. To remedy this, we can embed our algorithm into the partial separation framework from \cite{adjiashvili2018beating} that has also been used in \cite{grandoni2018improved} to obtain a solution with the appropriate cost relative to the LP (see \Cref{sec:partialseparation} for the details). In each step, our \emph{partial separation oracle} will either find a $(1.75+\mathcal{O}(\epsilon))$-approximate solution, or a \emph{violated visibly $k$-wide modification inequality}. Visibly $k$-wide modification inequalities are valid constraints for the integer hull of \eqref{eq:WDTAP_LP} that, loosely speaking, enforce that our LP solution is not ``too cheap'' on ``subinstances'' of visible width at most $k$. See \cref{def:visibly_k_wide_modification} for  a formal definition.

\paragraph{Coming down to earth}
The problem with the blue-sky approach is that the dream theorem is not true, essentially due to the issue with heavy coverage in the ``wrong direction'' outlined in \Cref{section:decomposition_does_not_work}. However, we can prove a \emph{weaker version} of the dream theorem that tells us that heavy coverage in the wrong direction is, in fact, \emph{the only issue} that we have to handle.

Given $\epsilon > 0$, we fix constants $\zeta_1\ll\zeta_2\ll k$ (see \Cref{sec:proof_weak_dream} for the precise values). $\zeta_1$ and $\zeta_2$ will be thresholds for considering an arc to be heavily covered in the ``right'' and ``wrong direction'', respectively. $k$ will be the bound on the visible width of instances that we aim for. 
In the following, we describe our partial separation oracle, that, given a solution $x$ to \eqref{eq:WDTAP_LP}, either finds a violated visibly $k$-wide modification inequality, or a solution of cost at most $(1+\epsilon)^2\cdot 1.75 \cdot c(x)$. As a first step, similar to the outline in \Cref{section:decomposition_WTAP}, we will contract every arc $a$ that is \emph{$\zeta_1$-covered}, i.e., satisfies $x(\{\ell\colon a\in\covr(\ell)\})\ge \zeta_1$. For $\zeta_1$ chosen large enough, these can be covered at a total cost of $\epsilon\cdot c(x)$. Hence, we may assume in the following that there are no $\zeta_1$-covered arcs.
We further call an arc $a$ \emph{$\zeta_2$-heavy} if $x(\{\ell\colon a\in \covw(\ell)\})\ge \zeta_2$. Finally, for $v\in V\setminus \{r\}$, let $a_v$ be the first arc on the $v$-$r$-path in $T$. We are now ready to state our weaker version of the dream theorem:
\begin{restatable}{theorem}{weakeneddreamtheorem}\label{theorem:weak_dream}
Let $\epsilon>0$ and $x$ a solution to \eqref{eq:WDTAP_LP}. We can, in polynomial time, compute a solution $x^*$ to \eqref{eq:WDTAP_LP} of cost $c(x^*)\le (1+\epsilon)\cdot c(x)$ that arises from $x$ by splitting links, and a set $W\subseteq V$ such that:
\begin{enumerate}[label=(\roman*)]
	\item $W$ consists of up- and down-independent vertices with respect to $\mathrm{supp}(x^*)=\{\ell\in L\colon x^*(\ell)>0\}$.
	\item Let $L'$ arise from $\mathrm{supp}(x^*)$ by splitting every $W$-cross-link at its apex.
	\begin{enumerate}
		\item The visible width at $r$, as well as at every $v\in V\setminus\{r\}$ for which $a_v$ is not $\zeta_2$-heavy, is at most $k$.
		\item For every $v\in V\setminus \{r\}$ such that $a_v$ is a $\zeta_2$-heavy up-arc (down-arc), the visible up-width (down-width) at $v$ is at most $k$.
	\end{enumerate}
\end{enumerate}	
\end{restatable}

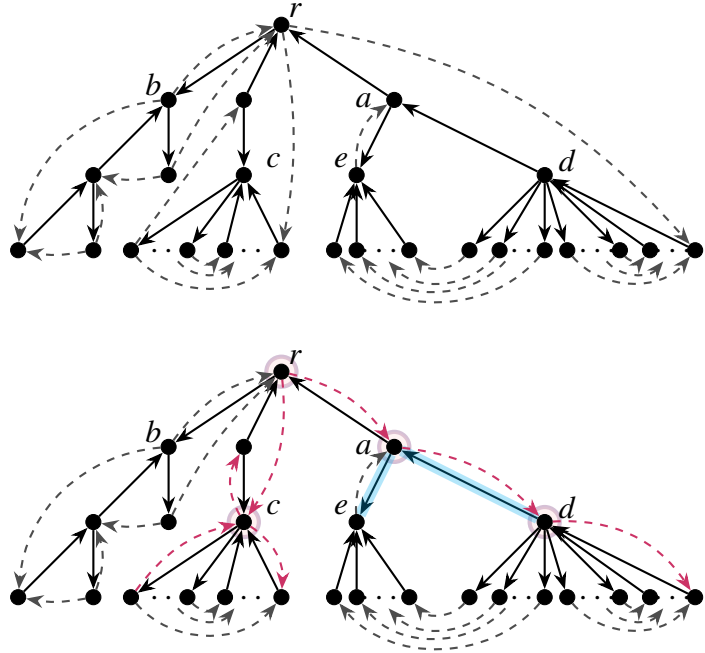
\begin{figure}[h]
\centering

\begin{tikzpicture}[scale = 1, mynode/.style={draw, fill, circle, minimum size = 2mm, inner sep = 0pt}, myarc/.style={thick, arrows = {-Stealth[scale=1]}},mylink/.style={thick, arrows = {-Stealth[scale=1]}, black!70!white, dashed}, every node/.style={font = \large},covered/.style={line width = 2mm, draw opacity = 0.3}]

\node[mynode] (r)   at (0.5, 3)  {};
\node           at (0.7, 3.2) {$r$};

\node[mynode] (b)  at (-1, 2)  {};
\node           at (-1.2, 2.2) {$b$};

\node[mynode] (v2)  at (0, 2)  {};
\node[mynode] (v3)  at (2, 2)  {};
\node           at (1.6, 2) {$a$};
\node[mynode] (v4)  at (-2 , 1)  {};
\node[mynode] (v5)  at (-1 , 1)  {};
\node[mynode] (v6)  at (-3 , 0)  {};
\node[mynode] (c)  at (0 , 1)  {};
\node           at ( 0.4 ,  1.2) {$c$};
\node[mynode] (e)  at (1.5 , 1)  {};
\node           at ( 1.3 ,  1.2) {$e$};
\node[mynode] (d)  at (4 , 1)  {};
\node           at (4.3, 1.2) {$d$};
\node[mynode] (v7) at (-2, 0) {};
\node[mynode] (v8) at (-1.5, 0) {};
\node (v9) at (-1, 0) {$\ldots$};
\node[mynode] (v10) at (-0.75, 0) {};
\node[mynode] (v11) at (-0.25, 0) {};
\node (v12) at (0, 0) {$\ldots$};
\node[mynode] (v13) at (0.5, 0) {};
\node[mynode] (a1) at (1.2, 0) {};
\node[mynode] (a2) at (1.5, 0) {};
\node (a3) at (1.8, 0) {$\ldots$};
\node[mynode] (a4) at (2.2, 0) {};
\node[mynode] (d1) at (3, 0) {};
\node[mynode] (d2) at (3.4, 0) {};
\node (d3) at (3.6, 0) {$\ldots$};
\node[mynode] (d4) at (4, 0) {};
\node[mynode] (d5) at (4.3, 0) {};
\node (d6) at (4.7, 0) {$\ldots$};
\node[mynode] (d7) at (5, 0) {};
\node[mynode] (d8) at (5.4, 0) {};
\node (d9) at (5.7, 0) {$\ldots$};
\node[mynode] (d10) at (6, 0) {};

\draw[myarc] (r) to (b);
\draw[myarc] (v2) to (r);
\draw[myarc] (v3) to (r);
\draw[myarc] (b) to (v5);
\draw[myarc] (v4) to (b);
\draw[myarc] (v4) to (v7);
\draw[myarc] (v6) to (v4);
\draw[myarc] (v2) to (c);
\draw[myarc] (c) to (v8);
\draw[myarc] (c) to (v10);
\draw[myarc] (v11) to (c);
\draw[myarc] (v13) to (c);
\draw[myarc] (v3) to (e);
\draw[myarc] (d) to (v3);
\draw[myarc] (a1) to (e);
\draw[myarc] (a2) to (e);
\draw[myarc] (a4) to (e);
\draw[myarc] (d) to (d1);
\draw[myarc] (d) to (d2);
\draw[myarc] (d) to (d4);
\draw[myarc] (d10) to (d);
\draw[myarc] (d8) to (d);
\draw[myarc] (d) to (d5);
\draw[myarc] (d) to (d7);

\draw[mylink] (b) to[bend right = 40] (v6);
\draw[mylink] (v5) to[bend left = 10] (v4);
\draw[mylink] (v7) to[bend right = 20] (v4);
\draw[mylink] (v7) to[bend left = 10] (v6);
\draw[mylink] (v5) to[bend left = 10] (r);
\draw[mylink] (b) to[bend left = 20] (r);

\draw[mylink] (v8) to[bend left = 0] (v2);
\draw[mylink] (r) to[bend left = 10] (v13);
\draw[mylink] (v10) to[bend right = 50] (v11);
\draw[mylink] (v9) to[bend right = 50] (v12);
\draw[mylink] (v8) to[bend right = 50] (v13);

\draw[mylink] (e) to[bend left = 30] (v3);

\draw[mylink] (d1) to[bend left = 50] (a4);
\draw[mylink] (d2) to[bend left = 50] (a3);
\draw[mylink] (d3) to[bend left = 50] (a2);
\draw[mylink] (d4) to[bend left = 50] (a1);

\draw[mylink] (d5) to[bend right = 50] (d10);
\draw[mylink] (d6) to[bend right = 50] (d9);
\draw[mylink] (d7) to[bend right = 50] (d8);

\draw[mylink] (r) to[bend left = 20] (d10);

\end{tikzpicture}

\vspace{0.8em}

\begin{tikzpicture}[scale = 1, mynode/.style={draw, fill, circle, minimum size = 2mm, inner sep = 0pt}, myarc/.style={thick, arrows = {-Stealth[scale=1]}},mylink/.style={thick, arrows = {-Stealth[scale=1]}, black!70!white, dashed}, every node/.style={font = \large},covered/.style={line width = 1.5mm, draw opacity = 0.3}]
\tikzset{
  halo/.style={
    append after command={
      \pgfextra{
        \draw[violet!70!black, ultra thick, opacity = 0.25]
              (\tikzlastnode.center) circle[radius=2mm];
        \fill[red!20, opacity=.25]
              (\tikzlastnode.center) circle[radius=2mm];
      }
    }
  }
}
\node[mynode, halo] (r)   at (0.5, 3)  {};
\node           at (0.7, 3.2) {$r$};

\node[mynode] (b)  at (-1, 2)  {};
\node           at (-1.2, 2.2) {$b$};

\node[mynode] (v2)  at (0, 2)  {};
\node[mynode, halo] (v3)  at (2, 2)  {};
\node           at (1.6, 2) {$a$};
\node[mynode] (v4)  at (-2 , 1)  {};
\node[mynode] (v5)  at (-1 , 1)  {};
\node[mynode] (v6)  at (-3 , 0)  {};
\node[mynode, halo] (c)  at (0 , 1)  {};
\node           at ( 0.4 ,  1.2) {$c$};
\node[mynode] (e)  at (1.5 , 1)  {};
\node           at ( 1.3 ,  1.2) {$e$};
\node[mynode, halo] (d)  at (4 , 1)  {};
\node           at (4.3, 1.2) {$d$};
\node[mynode] (v7) at (-2, 0) {};
\node[mynode] (v8) at (-1.5, 0) {};
\node (v9) at (-1, 0) {$\ldots$};
\node[mynode] (v10) at (-0.75, 0) {};
\node[mynode] (v11) at (-0.25, 0) {};
\node (v12) at (0, 0) {$\ldots$};
\node[mynode] (v13) at (0.5, 0) {};
\node[mynode] (a1) at (1.2, 0) {};
\node[mynode] (a2) at (1.5, 0) {};
\node (a3) at (1.8, 0) {$\ldots$};
\node[mynode] (a4) at (2.2, 0) {};
\node[mynode] (d1) at (3, 0) {};
\node[mynode] (d2) at (3.4, 0) {};
\node (d3) at (3.6, 0) {$\ldots$};
\node[mynode] (d4) at (4, 0) {};
\node[mynode] (d5) at (4.3, 0) {};
\node (d6) at (4.7, 0) {$\ldots$};
\node[mynode] (d7) at (5, 0) {};
\node[mynode] (d8) at (5.4, 0) {};
\node (d9) at (5.7, 0) {$\ldots$};
\node[mynode] (d10) at (6, 0) {};

\draw[myarc] (r) to (b);
\draw[myarc] (v2) to (r);
\draw[myarc] (v3) to (r);
\draw[myarc] (b) to (v5);
\draw[myarc] (v4) to (b);
\draw[myarc] (v4) to (v7);
\draw[myarc] (v6) to (v4);
\draw[myarc] (v2) to (c);
\draw[myarc] (c) to (v8);
\draw[myarc] (c) to (v10);
\draw[myarc] (v11) to (c);
\draw[myarc] (v13) to (c);
\draw[myarc] (v3) to (e);
\draw[covered, cyan] (v3) to (e);
\draw[myarc] (d) to (v3);
\draw[covered, cyan] (d) to (v3);
\draw[myarc] (a1) to (e);
\draw[myarc] (a2) to (e);
\draw[myarc] (a4) to (e);
\draw[myarc] (d) to (d1);
\draw[myarc] (d) to (d2);
\draw[myarc] (d) to (d4);
\draw[myarc] (d10) to (d);
\draw[myarc] (d8) to (d);
\draw[myarc] (d) to (d5);
\draw[myarc] (d) to (d7);

\draw[mylink] (b) to[bend right = 40] (v6);
\draw[mylink] (v5) to[bend left = 10] (v4);
\draw[mylink] (v7) to[bend right = 20] (v4);
\draw[mylink] (v7) to[bend left = 10] (v6);
\draw[mylink] (v5) to[bend left = 10] (r);
\draw[mylink] (b) to[bend left = 20] (r);

\draw[mylink, purple!80!white] (v8) to[bend left = 20] (c);
\draw[mylink, purple!80!white] (c) to[bend left = 30] (v2);
\draw[mylink, purple!80!white] (r) to[bend left = 20] (c);
\draw[mylink, purple!80!white] (c) to[bend left = 30] (v13);
\draw[mylink] (v10) to[bend right = 50] (v11);
\draw[mylink] (v9) to[bend right = 50] (v12);
\draw[mylink] (v8) to[bend right = 50] (v13);

\draw[mylink] (e) to[bend left = 30] (v3);

\draw[mylink] (d1) to[bend left = 50] (a4);
\draw[mylink] (d2) to[bend left = 50] (a3);
\draw[mylink] (d3) to[bend left = 50] (a2);
\draw[mylink] (d4) to[bend left = 50] (a1);

\draw[mylink] (d5) to[bend right = 50] (d10);
\draw[mylink] (d6) to[bend right = 50] (d9);
\draw[mylink] (d7) to[bend right = 50] (d8);

\draw[mylink, purple!80!white] (r) to[bend left = 20] (v3);
\draw[mylink, purple!80!white] (v3) to[bend left = 20] (d);

\draw[mylink, purple!80!white] (d) to[bend left = 30] (d10);

\end{tikzpicture}

\caption{Illustration of \cref{theorem:weak_dream}. (Top) A DTAP instance and a solution $x$ to \eqref{eq:WDTAP_LP} whose support is shown as dashed links. (Bottom) The resulting solution $x^*$ with $W=\{a,c,d,r\}$, where arcs $a_e$ and $a_d$ are $\zeta_2$-heavy (light blue). The links contained in the support of $x^*$, but not in the support of $x$, are shown in purple. All vertices except $a,c, d, e$ have visible width at most $k=3$. After splitting $W$-cross-links, $a$ and $c$ will have visible width $0$, $d$ will have visible up-width $0$ and $e$ will have visible down-width $0$.}
\label{fig:dtap-stacked}
\end{figure}

 In the following, we will sketch how to prove the weakened dream theorem using ideas similar to the decomposition approach used for WTAP. The heavy lifting happens in \Cref{section:phase_two_overview}, where we explain how to handle heavy arcs and leverage the weakened dream theorem to obtain the desired approximation guarantee.
\paragraph{Proving \cref{theorem:weak_dream}}
Our strategy to prove \cref{theorem:weak_dream} is similar to the decomposition approach for WTAP with bounded cost ratio described in \Cref{section:decomposition_WTAP}. However, instead of completely cutting off certain subtrees, we will only split certain links covering light arcs. This requires a more intricate scheme to bound the cost of the splitting and avoid ``overcharging''. We fix a constant $\gamma\sim\epsilon$. For $v\in V\setminus \{r\}$, we say that the arc $a_v$ is \emph{$\gamma$-up-light} if $x(L^\downarrow_v)\le \gamma\cdot x(\{\ell\in L\colon \text{ $\ell$ covers an up-arc in $T_v$ that $v$ can see}\}\setminus L^\downarrow_v)$, where visibility is defined with respect to the support of $x$. Analogously, we define the notion of \emph{$\gamma$-down-light} arcs by replacing $L^\downarrow_v$ by $L^\uparrow_v$ and ``up-arc'' by ``down-arc''. We obtain an LP solution $x^*$ and a vertex set $W$ as stated in \cref{theorem:weak_dream} as follows. We traverse $V\setminus \{r\}$ in order of decreasing distance to the root (i.e., from bottom to top). For $v\in V\setminus \{r\}$, if $a_v$ is $\gamma$-up-light, we split every link in $L^\downarrow_v$ at $v$ and add $v$ to $W$. Similarly, if $a_v$ is $\gamma$-down-light, we split every link in $L^\uparrow_v$ at $v$ and add $v$ to $W$. Finally, add $r$ to $W$.
By construction, every vertex in $W$ is up- or down-independent.
To bound the total cost increase by $\mathcal{O}(\epsilon)\cdot c(x)$, we observe that if $\ell\in L\setminus L^\downarrow_v$ covers an up-arc in $T_v$, then $\ell$ only covers arcs within $T_v$; otherwise, we had $\ell\in L^\downarrow_v$. When splitting all links in $L^\downarrow_v$ at $v$, every up-arc in $T_v$ becomes invisible to every vertex outside $T_v$. Hence, $x(\ell)$ will be ``charged against'' for splitting at a $\gamma$-up-light arc at most once, and the same reasoning applies for $\gamma$-down-light arcs. Hence, $\gamma\sim\epsilon$ yields the desired cost bound.
Finally, we sketch how to derive (ii). We will only explain how to bound the visible up-width at every vertex; the visible-down width can be handled analogously. As $r\in W$, the visible (up)-width at $r$ is $0$ after splitting all $W$-cross-links at their apex. Next, let $v\in V\setminus \{r\}$ and assume that the visible up-width at $v$ is greater than $k$. This means that there is an ancestor-free set $A'$ of at least $k+1$ up-arcs in $T_v$ that are all visible from $v$. Given that no link can cover two up-arcs from an ancestor-free set simultaneously, this implies $x(\{\ell\in L\colon \text{ $\ell$ covers an up-arc in $T_v$ that $v$ can see}\})> k$ holds at the end of the splitting procedure and by the order in which vertices are considered, it also holds when we look at $v$. If $x(L^\downarrow_v)<\zeta_2$, then by our choice of constants, $a_v$ is $\gamma$-up-light. Hence, every link in $L^\downarrow_v$ is split at $v$, $v\in W$ and after splitting all $W$-cross-links at their apex, the visible up-width at $v$ is $0<k$, a contradiction. So we must have $x(L^\downarrow_v)\ge \zeta_2>\zeta_1$. Note that if $a_v$ is an up-arc, then $a_v\in\covr(\ell)$ for every $\ell\in L^\downarrow_v$; otherwise $a_v\in \covw(\ell)$ for every $\ell\in L^\downarrow_v$. As we contracted all $\zeta_1$-covered arcs, $a_v$ must be a down-arc that is $\zeta_2$-heavy.

\subsection{Components and cores: handling coverage in the wrong direction\label{section:phase_two_overview}}
\paragraph{An instructive special case}
To explain how we handle $\zeta_2$-heavy arcs, it is helpful to first consider the slightly artificial, but instructive special case in which for every $\zeta_2$-heavy arc $a_v$, the parent arc (if exists) is oppositely oriented. More precisely, we assume that if $a_v=(v,w)$ is a $\zeta_2$-heavy up-arc and $w\ne r$, then $a_w$ is a down-arc, and if $a_v=(w,v)$ is a $\zeta_2$-heavy down-arc and $w\ne r$,then $a_w$ is an up-arc. In this situation, we can again obtain a $(1.5+\mathcal{O}(\epsilon))$-approximation using the following result:
\begin{theorem}\label{theorem:heavy_link_handling_special_case}
We can, in polynomial time, compute a solution $x^{**}$ to \eqref{eq:WDTAP_LP} of cost $c(x^{**})\le (1+\epsilon)\cdot c(x^*)$ that arises from $x^*$ by splitting links, and $X\subseteq V$ such that:
\begin{enumerate}[label=(\roman*)]
\item $X$ consists of up- and down-independent vertices with respect to $\mathrm{supp}(x^{**})=\{\ell\in L\colon x^{**}(\ell)>0\}$.
\item Let $L'$ arise from $\mathrm{supp}(x^{**})$ by splitting every $X$-cross-link at its apex. For every $v\in V\setminus \{r\}$ such that $a_v$ is a $\zeta_2$-heavy up-arc (down-arc), the visible down-width (up-width) of $v$ w.r.t.\ $L'$ is $0$.
\end{enumerate}
\end{theorem}
Before proving \cref{theorem:heavy_link_handling_special_case}, let us first discuss how to leverage it to obtain the desired approximation guarantee. As splitting links can only reduce the visible width, \cref{theorem:weak_dream}~(ii) and \cref{theorem:heavy_link_handling_special_case}~(ii) tell us that after splitting every $(W\cup X)$-cross-link in $\mathrm{supp}(x^{**})$, we obtain an instance of visible width at most $k$. On the other hand, as splitting links cannot destroy up- or down-independence, \cref{theorem:weak_dream}~(i) and \cref{theorem:heavy_link_handling_special_case}~(i) tell us that after splitting every link in $\mathrm{supp}(x^{**})$ that is not a $(W\cup X)$-cross-link at its apex, we obtain a willow. Hence, we may proceed as in \Cref{section:algorithm} to obtain a solution of cost $(1.5+\mathcal{O}(\epsilon))\cdot c(x)$.
\paragraph{Proving \cref{theorem:heavy_link_handling_special_case}}
We obtain $x^{**}$ as follows: for $v\in V\setminus\{r\}$ such that $a_v=(v,w)$ is a $\zeta_2$-heavy up-arc, we split every link in $L^\downarrow_v$ at $v$ and every link in $L^\uparrow_w$ at $w$. We add $v$ and $w$ to $X$. Similarly, for $v\in V\setminus\{r\}$ such that $a_v=(w,v)$ is a $\zeta_2$-heavy down-arc, we split every link in $L^\uparrow_v$ at $v$ and every link in $L^\downarrow_w$ at $w$. Again, we add $v$ and $w$ to $X$. Property (i) is clear by construction. For property (ii), let w.l.o.g.\ $v\in V\setminus\{r\}$ such that $a_v=(v,w)$ is a $\zeta_2$-heavy up-arc. Then $L'$ doesn't contain any $v$-cross-link. Moreover, every link in $\ell\in L^\uparrow_v\cap L'$ has to end at $w$ because all $w$-cross-links and all links in $L^\uparrow_w$ were split at $w$. In particular, as the up-link $a_v$ is not covered by $\ell$, $v$ is not an inner vertex of $\overline{P}_\ell$. But this implies that $v$ cannot see any down-arc in $T_v$. It remains to bound the cost of the splitting. Again, let $v\in V\setminus\{r\}$ such that $a_v=(v,w)$ is a $\zeta_2$-heavy up-arc. We know that $x(L^\downarrow_v)\le x(\{\ell\colon a_v\in\covr(\ell)\}<\zeta_1$ because there are no $\zeta_1$-covered arcs. Similarly, $x(L^\uparrow_w)< \zeta_1$ because if $w=r$, then $L^\uparrow_w=\emptyset$, and otherwise, $a_w$ is a down-arc. On the other hand, $x(\{\ell\colon a_v\in \covw(\ell)\})\ge \zeta_2$, and $\{\ell\colon a_v\in \covw(\ell)\}\subseteq \{\ell\colon a_w\in\covr(\ell)\}\cup \{\ell\colon\mathrm{apex}(\ell)=w\}$, if $w\ne r$, and $\{\ell\colon a_v\in \covw(\ell)\}\subseteq \{\ell\colon\mathrm{apex}(\ell)=w\}$ otherwise. This implies that $x(\{\ell\colon\mathrm{apex}(\ell)=w\})\ge \zeta_2-\zeta_1$ because $a_w$, if exists, is not $\zeta_1$-covered. Using $\zeta_1\ll \zeta_2$, we can charge the splitting of the links in $L^\downarrow_v$ and $L^\uparrow_w$ against the total costs of the links with apex $w$. 
\paragraph{The general case} To handle the general case, we consider connected \emph{components} of the (oriented) forests $(V,A_{up})$ and $(V,A_{down})$, and define the \emph{core of a component} to consist of all of the paths connecting $\zeta_2$-heavy arcs in the component to its root. We denote the set of vertices and arcs that are contained in a core $C$ by $V_C$ and $A_C$, respectively. Note that the special case we considered corresponds to the situation in which every core has depth $1$. While the proof of \cref{theorem:heavy_link_handling_special_case} extends to the case where every core has constant depth, this approach is too costly in general. Instead, we perform a more involved trade-off between three different solutions, obtaining a solution of cost at most $(1.75+\mathcal{O}(\epsilon))\cdot c(x)$, or a violated visibly $k$-wide modification inequality. To this end, we define $L_{cross}$ to be the collection of all $W\cup V_C$-cross-links, where $W$ is the vertex set from \cref{theorem:first_phase_splitting}. We define $\overrightarrow{L}\coloneqq \{\ell\colon \covr(\ell)\cap A_C\ne \emptyset\}$ and $\overleftarrow{L}\coloneqq \{\ell\colon \covw(\ell)\cap A_C\ne \emptyset\}$ to be the sets of links covering a core arc in the right or in the wrong direction, respectively. Via carefully designed splitting operations, that only increase the total costs by an $\mathcal{O}(\epsilon)$-fraction, we can ensure useful structural properties, including $\overrightarrow{L}\cap \overleftarrow{L}=\emptyset$. Moreover, we can establish the following three statements.
\begin{enumerate}
\item Splitting all links in $\overleftarrow{L}\cup L_{cross}$ yields an instance of constant visible width.
\item Splitting all links in $\overleftarrow{L}$ and all links in $L\setminus L_{cross}$ at their apex yields a willow.
\item Splitting all links in $L\setminus\overleftarrow{L}$ at their apex and every link in $\overrightarrow{L}$ once more yields an instance corresponding to the disjoint union of willows.
\end{enumerate}
Taking an appropriate weighted average of these three solutions yields an approximation guarantee of $1.75 + \mathcal{O}(\varepsilon)$.

\section{Total unimodularity for willows}\label{sec:willows}
Let $(T = (V,A),L,c,r)$ be a rooted WDTAP instance, and recall the linear programming relaxation given in (\ref{eq:WDTAP_LP}). This is not an integral formulation in general; see \Cref{subsection:integralitygap}. In this section, we derive sufficient conditions for the incidence matrix $M$ to be totally unimodular, yielding an integral formulation and allowing us to solve the corresponding WDTAP instance in polynomial time. 

We begin by formally defining the notions of up- and down-independence and willows introduced in \Cref{sec:ourcontribution}.
 \begin{definition} \label{def:up-down-independent}
We say that $v\in V$ is \emph{up-independent} with respect to  $L'\subseteq L$ if for every $\ell\in L'$, we have $\covr(\ell)\cap A_v\cap A_{up}=\emptyset$ or $\covr(\ell)\subseteq A_v$. We say that $v$ is \emph{down-independent} with respect to $L'$ if for every $\ell\in L'$, we have $\covr(\ell)\cap A_v\cap A_{down}=\emptyset$ or $\covr(\ell)\subseteq A_v$. 
\end{definition}

We recap the following definitions from \Cref{sec:prelims} and \Cref{sec:ourcontribution}.
We call a link $\ell=(u,v)$ an \emph{up-link} if $v=\mathrm{apex}(\ell)$ and a \emph{down-link} if $u=\mathrm{apex}(\ell)$. For a set of vertices $W$, we call $\ell$ a \emph{$W$-cross-link} if $\mathrm{apex}(\ell)\in W$ and $\ell$ is neither an up- nor a down-link. 
\begin{definition}\label{def:willow}
  We call a rooted WDTAP instance $(T,L,c,r)$ a \emph{willow} if there is a vertex set $W\subseteq V(T)$ such that 
  \begin{itemize}
      \item every vertex in $W$ is up- or down-independent with respect to $L$ and
      \item every link in $L$ is an up-link, a down-link or a $W$-cross-link.
  \end{itemize}
\end{definition}
\begin{theorem}[unimodularity theorem]\label{theorem:unimodularity}
Let $(T,L,c,r)$ be a willow and let $M$ be its arc-link-coverage matrix. Then $M$ is totally unimodular. In particular, an optimum integral solution to \eqref{eq:WDTAP_LP} can be found in polynomial time.
\end{theorem}
For the proof, it is convenient to introduce the following additional notation: For an arc $a$, we call the endpoint of $a$ that is closer to the root the \emph{apex} of $a$ and denote it by $\mathrm{apex}(a)$. Given two vertices $u$ and $v$ of $T$, we write $P_{uv}$ to denote the $u$-$v$-path in $T$.
\begin{proof}[Proof of \cref{theorem:unimodularity}]
Let $T=(V,A)$ and let $W$ be as in \cref{def:willow}.
We may assume $r\in W$ because $r$ is both up- and down-independent.
To establish total unimodularity of $M$, we use the criterion by Ghouila-Houri~\cite{ghouila1962caracterisation}. It states that a matrix $A\in\{-1,0,1\}^{m\times n}$ is totally unimodular if and only if for every subset $R\subseteq \{1,\dots,m\}$ of the rows, there exists a signing $\sigma\colon R\rightarrow\{-1,+1\}$ such that for every column $j\in\{1,\dots,n\}$, $\sum_{i\in R} \sigma(i)\cdot A_{ij} \in \{-1,0,1\}$, where $A_{ij}$ denotes the entry of $A$ in row $i$ and column $j$.

Applying this to our setting where rows correspond to arcs and columns correspond to links, we need to prove that for every $B\subseteq A$, there exists a signing $\sigma\colon B\rightarrow\{-1,+1\}$ such that 
\begin{equation}
\text{for every $\ell\in L$ }\sum_{a\in\covr(\ell)\cap B} \sigma(a)\in\{-1,0,1\}.\label{eq:ghouila_houri_signing}
\end{equation}

For two vertices $v$ and $w$, we define $\mathrm{dist}_{up}(v,w)$ and $\mathrm{dist}_{down}(v,w)$ to be the number of up- and down-arcs from $B$ on the $v$-$w$-path in $T$, respectively.
To construct the signing, we define starting signs $\varphi_{up},\varphi_{down}\colon W\rightarrow\{-1,+1\}$ in order of increasing distance (in all of $T$) to the root $r$.
\begin{itemize}
    \item We set $\varphi_{up}(r)=+1$ and $\varphi_{down}(r)=-1$.
    \item Let $u\in W$ be up-independent and let $v\in W\setminus\{u\}$ be the next vertex after $u$ on the $u$-$r$-path in $T$. We define $\varphi_{down}(u)\coloneqq\varphi_{down}(v)\cdot (-1)^{\mathrm{dist}_{down}(u,v)}$ and $\varphi_{up}(u)\coloneqq -\varphi_{down}(u).$
    \item Let $u\in W$ be down-independent (but not up-independent) and let $v\in W\setminus\{u\}$ be the next vertex after $u$ on the $u$-$r$-path in $T$. We define $\varphi_{up}(u)\coloneqq\varphi_{up}(v)\cdot (-1)^{\mathrm{dist}_{up}(u,v)}$ and $\varphi_{down}(u)\coloneqq -\varphi_{up}(u).$
\end{itemize}
For an arc $a$, let $\mu(a)$ be the first vertex (i.e., then one closest to $\mathrm{apex}(a)$) from $W$ on the $\mathrm{apex}(a)$-$r$-path in $T$.
\begin{itemize}
    \item For an up-arc $a\in B$, we set $\sigma(a)=\varphi_{up}(\mu(a))\cdot (-1)^{\mathrm{dist}_{up}(\mathrm{apex}(a),\mu(a))}$.
    \item For a down-arc $a\in B$, we set $\sigma(a)=\varphi_{down}(\mu(a))\cdot (-1)^{\mathrm{dist}_{down}(\mathrm{apex}(a),\mu(a))}$.
\end{itemize}
\cref{fig:willow_signing} shows an example of this signing for $B=A$.
\begin{claim} \label{claim:sign_up_arc}
Let $a\in A_{up}\cap B$ and let $u\in W\cap V(P_{\mathrm{apex}(a)r})$. Assume that no vertex in $W\cap V(P_{\mathrm{apex}(a)u})\setminus \{u\}$ is up-independent. Then $\sigma(a)=\varphi_{up}(u)\cdot (-1)^{\mathrm{dist}_{up}(\mathrm{apex}(a),u)}.$
\end{claim}
\begin{proof}[Proof of claim]
Let $W\cap V(P_{\mathrm{apex}(a)u})=(\mu(a)=u_s,\dots,u_0=u)$ with $u_s,\dots,u_0$ appearing in this order when traversing $P_{\mathrm{apex}(a)u}$ from $\mathrm{apex}(a)$ to $u$. Using $\varphi_{up}(u_i)=\varphi_{up}(u_{i-1})\cdot (-1)^{\mathrm{dist}_{up}(u_i,u_{i-1})}$ for $i=1,\dots,s$, we obtain
\begin{align*}\sigma(a)&=\varphi_{up}(u_s)\cdot (-1)^{\mathrm{dist}_{up}(\mathrm{apex}(a),u_s)}=\varphi_{up}(u)\cdot(-1)^{\mathrm{dist}_{up}(\mathrm{apex}(a),u_s)+\sum_{i=1}^s \mathrm{dist}_{up}(u_{i},u_{i-1})}\\
&=\varphi_{up}(u)\cdot (-1)^{\mathrm{dist}_{up}(\mathrm{apex}(a),u)}.\end{align*}
\end{proof}
Analogously, we obtain the following claim.
\begin{claim}\label{claim:sign_down_arc}
	Let $a\in A_{down}\cap B$ and let $u\in W\cap V(P_{\mathrm{apex}(a)r})$. Assume that no vertex in $W\cap V(P_{\mathrm{apex}(a)u})\setminus \{u\}$ is down-independent. Then $\sigma(a)=\varphi_{down}(u)\cdot (-1)^{\mathrm{dist}_{down}(\mathrm{apex}(a),u)}.$
\end{claim}
\begin{claim}\label{claim:up_alternate}
Let $\ell=(u,v)\in L$ and let $a$ and $a'$ be two up-arcs in $B$ that appear consecutively on $P_{\mathrm{apex}(\ell)v}$. Then $\sigma(a')=-\sigma(a)$.
\end{claim}
\begin{proof}[Proof of claim]
Let $a=(x,y)$ and $a'=(x',y')$ and assume w.l.o.g.\ that $a$ appears before $a'$ on $P_{\mathrm{apex}(\ell)v}$ (traversing it from $\mathrm{apex}(\ell)$ to $v$), i.e., $a$ is above $a'$.
No vertex $v\in V(P_{xy'})$ is up-independent because $a'\in \covr(\ell)\cap A_v\cap A_{up}$ and $a\in\covr(\ell)\setminus A_v$. As $\mu(a)$ is the first vertex from $W$ on $P_{yr}$, we can apply \cref{claim:sign_up_arc} to conclude that 
\[\sigma(a')=\varphi_{up}(\mu(a))\cdot(-1)^{\mathrm{dist}_{up}(y',\mu(a))}=(-1)\cdot \varphi_{up}(\mu(a))\cdot(-1)^{\mathrm{dist}_{up}(y,\mu(a))}=-\sigma(a)\] because $a$ and $a'$ are consecutive up-arcs from $B$ on $P_{\mathrm{apex}(\ell)v}$, i.e., $\mathrm{dist}_{up}(y',\mu(a))=\mathrm{dist}_{up}(y,\mu(a))+1$.
\end{proof}
Analogously, we obtain the following claim:
\begin{claim}\label{claim:down_alternate}
	Let $\ell=(u,v)\in L$ and let $a$ and $a'$ be two down-arcs in $B$ that appear consecutively on $P_{u\mathrm{apex}(\ell)}$. Then $\sigma(a')=-\sigma(a)$.
\end{claim}
Now, we are ready to show that our signing satisfies \eqref{eq:ghouila_houri_signing}. Let $\ell\in L$. If $\covr(\ell)\cap B$ consists of only up- or only down-arcs, this follows from \cref{claim:up_alternate} or \cref{claim:down_alternate}, respectively.
Finally, assume that $\ell$ is a $W$-cross-link such that $\covr(\ell)\cap B$ contains at least one up- and one down-arc. Let $\mathrm{apex}(\ell)=u$ and let $a=(x,y)$ and $a'=(x',y')$ be the up- and the down-arc in $\covr(\ell)\cap B$ closest to $u$. No vertex $v\in V(P_{yu})\setminus\{u\}$ is up-independent because $a\in\covr(\ell)\cap A_{up}\cap A_v$ and $a'\in\covr(\ell)\setminus A_v$. No vertex $v\in V(P_{x'u})\setminus\{u\}$ is down-independent because $a'\in\covr(\ell)\cap A_{down}\cap A_v$ and $a\in\covr(\ell)\setminus A_v$. By \cref{claim:sign_up_arc} and \cref{claim:sign_down_arc}, using that $a$ and $a'$ are the up-/down-arc in $\covr(\ell)\cap B$ closest to $u$, we get $\sigma(a)=\varphi_{up}(u)=-\varphi_{down}(u)=-\sigma(a')$. \cref{claim:up_alternate} and \cref{claim:down_alternate} allow us to conclude that the signs of the arcs in $\covr(\ell)\cap B$ alternate along $P_\ell$, implying $\sum_{a\in\covr(\ell)\cap B} \sigma(a)\in\{-1,0,1\}$ as desired.

The fact that $M$ is totally unimodular implies that all vertex solutions to the linear program \eqref{eq:WDTAP_LP} are integral. We can find an optimum vertex solution in polynomial time, giving an optimum solution to $(T,L,c,r)$.
\end{proof}

\begin{figure}[h]
\centering
\begin{tikzpicture}[scale = 1, mynode/.style={draw, fill, circle, minimum size = 2mm, inner sep = 0pt}, myarc/.style={thick, arrows = {-Stealth[scale=1]}},mylink/.style={thick, arrows = {-Stealth[scale=1]}, black!70!white, dashed}, every node/.style={font = \large}]

\tikzset{
  halo/.style={                        
    append after command={             
      \pgfextra{
        \draw[violet!70!black, ultra thick, opacity = 0.25]   
              (\tikzlastnode.center) circle[radius=3.5mm];
        \fill[red!20, opacity=.25]          
              (\tikzlastnode.center) circle[radius=3.5mm];
      }
    }
  }
}

\node[mynode, halo] (r)  at ( 0  ,  0) {};
\node           at ( 0.2,  0.5) {$r$};

\node[mynode] (v1p)  at (-1  , -0.5) {};
\node[mynode] (v1)  at (-2  , -1) {};
\node[mynode] (v2)  at ( 1  , -1) {};

\node[mynode, halo] (v3)  at (-2  , -2) {};
\node           at ( -2.5,  -1.9) {$u$};

\node[mynode] (v4)  at ( 0.5, -2) {};
\node[mynode, halo] (v5)  at ( 2  , -2) {};
\node           at ( 2.5,  -1.9) {$v$};

\node[mynode] (v6)  at (-1  , -3) {};
\node[mynode] (v7)  at (-3  , -3) {};
\node[mynode] (v10) at ( 2  , -3) {};
\node[mynode] (v11) at ( 3  , -3) {};

\node[mynode] (v12) at (-3  , -4) {};
\node[mynode] (v13) at (-2  , -3) {};
\node[mynode] (v14) at ( 1  , -3) {};
\node[mynode] (v15) at ( 0.5, -4) {};
\node[mynode] (v16) at ( 1.5, -4) {};
















\draw[myarc] (v1p) to node[midway, above] {$\color{purple} +$} (r);
\draw[myarc] (v1) to node[midway, above] {$\color{blue!70!black} -$} (v1p);
\draw[myarc] (r) to node[midway, right] {$\color{blue!70!black} -$} (v2);
\draw[myarc] (v1) to node[midway, right] {$\color{blue!70!black} -$} (v3);
\draw[myarc] (v4) to  node[midway, left] {$\color{purple} +$}(v2);
\draw[myarc] (v2) to node[midway, right] {$\color{purple} +$} (v5);
\draw[myarc] (v6) to node[midway, right] {$\color{purple} +$} (v3);
\draw[myarc] (v7) to node[midway, left] {$\color{purple} +$} (v3);
\draw[myarc] (v5) to node[midway, right] {$\color{blue!70!black} -$} (v10);
\draw[myarc] (v5) to node[midway, right] {$\color{blue!70!black} -$} (v11);
\draw[myarc] (v12) to node[midway, right] {$\color{blue!70!black} -$} (v7);
\draw[myarc] (v3) to node[midway, left]{$\color{blue!70!black} -$} (v13);
\draw[myarc] (v14) to node[midway, left] {$\color{purple} +$} (v5);
\draw[myarc] (v14) to node[midway, left] {$\color{blue!70!black} -$} (v15);
\draw[myarc] (v14) to node[midway, right] {$\color{blue!70!black} -$} (v16);

\draw[mylink] (v1) to[bend right = 70] (v12);

\draw[mylink] (v11) to[bend right = 40] (r);
\draw[mylink] (v3) to [bend left= -20] (v4);
\draw[mylink] (v13) to [bend left= 20] (v7);
\draw[mylink] (v15) to [bend left= 20] (v6);
\draw[mylink] (v10) to [bend left= 20] (v14);
\draw[mylink] (v16) to [bend left= 70] (v13);
\draw[mylink] (v13) to [bend left= -20] (v6);

\end{tikzpicture}
\caption{The signing constructed in the proof of \cref{theorem:unimodularity} for the willow from \cref{fig:willow} and $B=A$ (the set of all arcs).}
\label{fig:willow_signing}
\end{figure}
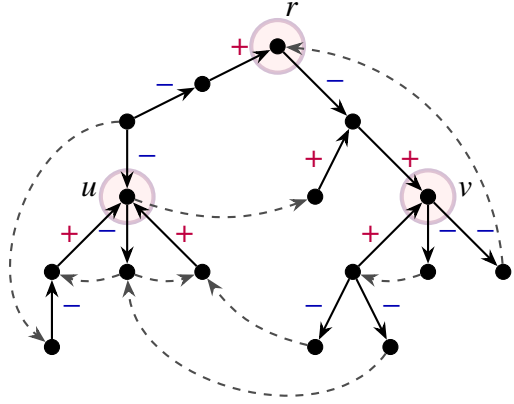
\section{Dynamic program for instances of constant visible width}\label{sec:DPvisiblewidth}
In this section, we define the notion of the \textit{visible width} of a WDTAP instance. We then show that WDTAP instances with constant visible width can be solved exactly using a dynamic program. 
Let $(T=(V,A),L,c)$ be an instance of WDTAP and let $T$ be rooted at $r\in V$. We introduce some further common terminology that we will use in the following. Given $v\in V\setminus \{r\}$, we call the endpoint of $a_v$ other than $v$ the \emph{parent} of $v$. (Recall that $a_v$ is the first arc on the $v$-$r$-path in $T$.) We call a vertex $w$ that has $v$ as its parent a \emph{child} of $v$. For $v\in V$, we say that a vertex $w$ is an \emph{ancestor} of $v$ if $w$ lies on the $v$-$r$-path in $T$, and we say that $v$ is a \emph{descendant} of $w$. If in addition, $w\ne v$, we call $w$ a \emph{strict ancestor} of $v$ and $v$ a \emph{strict descendant} of $w$. Note that for $v\in V$, the set of descendants of $v$ is $U_v$, the vertex set of $T_v$.
Finally, we call an arc $a'$ an \emph{ancestor} of another arc $a$ if $a'$ appears on the $\mathrm{apex}(a)$-$r$ path in $T$. In order to formally introduce the concept of visible width, we need the notion of an \emph{ancestor-free} arc set.
\begin{definition}
We call an arc set $F\subseteq A$ \emph{ancestor-free} if there are no arcs $a,a'\in F$ such that $a'$ appears in the $\mathrm{apex}(a)$-$r$ path in $T$.
\end{definition}

We now define the notion of which arcs in the subtree of $v$ are visible to $v$. 

\begin{definition}
We say that an arc $a\in A_v$ is \emph{visible} to a vertex $v\in V$ (with respect to a set of links $L'$) if there exists a link $\ell\in L'$ such that $a\in \covr(\ell)$ and $v\in\mathrm{in}(\overline{P}_\ell)$, where $\mathrm{in}(\overline{P}_\ell)$ denotes the set of inner vertices of $\overline{P}_\ell$.
We denote by $A^{vis}_v(L')$ the set of arcs that are visible from $v$ with respect to \ $L'$. 
\end{definition}
\begin{definition}
For a vertex $v\in V$, we define the \emph{visible up-width}, denoted by $\viwu(v)$ (\emph{visible down-width}, denoted by $\viwd(v)$) at $v$ to be the maximum size of an ancestor-free set of up-arcs (down-arcs) that are visible for $v$ (with respect to $L$). We define the \emph{visible width} at $v$ as 
\[\viw(v)\coloneqq \max\{\viwu(v),\viwd(v)\}.\]
We define the \emph{visible width} of the instance to be $\max_{v\in V} \viw(v)$.
\end{definition}
\begin{definition}
We call a link set $L'\subseteq L$ \emph{$k$-thin} if for every $v\in V$,
$|\{\ell\in L'\colon v\in\mathrm{in}(P_\ell)\}|\le k$.
\end{definition}
We remark that our definition of thinness slightly differs from the one introduced in \cite{traub2022better} (in the context of WTAP) in that we do not count links ending in a vertex $v$.
\begin{lemma}\label{lem:kthin}
Assume that $(T=(V,A),L,c)$ has visible width at most $k$. Let $F\subseteq L$ be a shadow-minimal (meaning that no link can be replaced by a strict shadow without destroying feasibility) solution to the instance. Then $F$ is $2k$-thin.
\end{lemma}
\begin{proof}
As $F$ is shadow-minimal, we have $\ell=s(\ell)$ and $P_\ell=\overline{P}_\ell$ for every $\ell\in F$.
Let $v\in V$ and let $F'\coloneqq \{\ell\in F\colon v\in \mathrm{in}(P_\ell)\}$. We need to show that $|F'|\le 2k$. For each $\ell\in F'$, let $w_\ell\in U_v\setminus\{v\}$ be an endpoint of $\ell$ (this endpoint is unique unless $\ell$ is a $v$-cross-link, in which case we may select either endpoint). Let $a_\ell\in A(P_\ell)$ be the arc incident to $w_\ell$. Note that $a_\ell\in A_v$. Let $F'_{up}\coloneqq \{\ell\in F'\colon a_\ell\in A_{up}\}$ and let $F'_{down}\coloneqq \{\ell\in F'\colon a_\ell\in A_{down}\}$. We show that $|F'_{up}|\le k$ and $|F'_{down}|\le k$, which implies the desired statement. We only show $|F'_{up}|\le k$, $|F'_{down}|\le k$ can be derived analogously.
We observe that by shadow-minimality of $F$, we must have $a_\ell\in\covr(\ell)$ for every $\ell\in F'$. In particular, $\ell$ witnesses that $a_\ell$ is visible for $v$. In fact, shadow-minimality allows us to derive an even stronger statement: we must have $a_\ell\in\covr(\ell)\setminus \bigcup_{\ell'\in F\setminus\{\ell\}} \covr(\ell')$. In particular, the arcs $(a_\ell)_{\ell\in F'_{up}}$ are pairwise distinct. We further claim that they form an ancestor-free arc set. As $\viw(v)\le k$, this implies $|F'_{up}|\le k$. Assume towards a contradiction that there were two links $\ell=(u,x),\ell'=(u',x')\in F'_{up}$ such that $a_{\ell}$ appears on the path $P_{y'r}$ in $T$ from the head $y'$ of $a_{\ell'}\eqqcolon (x',y')$ to the root $r$. As $a_{\ell}\in A_v$,  $a_\ell$ appears on the $y'$-$v$-subpath $P_{y'v}$ of $P_{y'r}$. As $y'$ is the parent of the head $x'$ of $\ell'$ and $v\in\mathrm{in}(P_{\ell'})$, $P_{y'v}$ is a subpath of $P_{\ell'}$ and as $v$ is an ancestor of $y'$, $\ell'$ covers every up-arc on that path, including $a_\ell$. But this contradicts the fact that $a_\ell\in\covr(\ell)\setminus \bigcup_{\ell''\in F\setminus\{\ell\}} \covr(\ell'')$.
\end{proof}
We remark that there always exists a shadow-minimal optimum solution because we can iteratively replace links in an optimum solution by strict shadows without increasing the cost until the solution is shadow-minimal.
\begin{lemma}\label{lem:DPkthin}
Let $N\in\mathbb{N}$ be a constant. Given a rooted WDTAP instance $(T,L,c,r)$, we can, in polynomial time, find a cheapest $N$-thin solution, or decide that the instance is infeasible.
\end{lemma}
\begin{proof}
Let $(T,L,c,r)$ be a rooted instance of WDTAP. Recall that for $v\in V$, $T_v=(U_v,A_v)$ is the subtree rooted at $v$. We define the following three links sets for $v\in V$:
\begin{itemize}
	\item $L_v$ is the set of links that have at least one endpoint in $U_v\setminus \{v\}$. Note that any link that covers an arc $a\in A_v$ must be contained in $L_v$.
	\item $L_v^{out}$ consists of all links with one endpoint in $U_v\setminus\{v\}$ and the other endpoint in $V\setminus U_v$. Note that for every $\ell\in L_v^{out}$, $v\in\mathrm{in}(P_\ell)$.
	\item $L_v^{cross}$ consists of all links $\ell$ with both endpoints in $U_v\setminus\{v\}$ and $\mathrm{apex}(\ell)=v$. Note that for every $\ell\in L_v^{cross}$, $v\in\mathrm{in}(P_\ell)$.
\end{itemize}
We further point out that if $\ell\in L$ and $v\in\mathrm{in}(P_\ell)$, then $\ell\in L_v^{out}\cup L_v^{cross}$.

Let $v\in V$ and $Y\subseteq L_v^{out}$. We call a link set $F\subseteq L_v$ \emph{feasible} for $(v,Y)$ if $F$ is $N$-thin, $F\cap L_v^{out}=Y$ and $F$ covers every arc in $A_v$. We define $c(v,Y)$ to be the minimum cost of a feasible link set for $(v,Y)$, or $\infty$, if no such link set exists.

We will use dynamic programming to, for every $v\in V$ and $Y\subseteq L_v^{out}$ with $|Y|\le N$, compute $c(v,Y)$, as well as a feasible link set $F^*(v,Y)$ for $(v,Y)$ with $c(F^*(v,Y))=c(v,Y)$, or $F^*(v,Y)=\emptyset$, if $c(v,Y)=\infty$. 
We remark that if the instance admits a feasible solution, then by shadow-completeness, we, for every arc $a=(u,w)$, have a link $\ell=(w,u)$ just covering $a$, and we can always use them to complete $Y$ to an $N$-thin solution.
We note that $L_r=L$ (assuming that we do not have links of the form $(v,v)$ that do not cover any arc) and $L_r^{out}=\emptyset$, so $F^*(r,\emptyset)$ yields a cheapest $N$-thin solution to the instance, or $c(r,\emptyset)=\infty$ and $F^*(r,\emptyset)=\emptyset$ if no such solution exists.

As $N$ is a constant, there is only a polynomial number of pairs $(v,Y)$ that we consider. We traverse the pairs in order of non-increasing distance of $v$ to the root, which ensures that when considering a pair $(v,Y)$, all pairs $(v',Y')$ with $v'\in U_v\setminus \{v\}$ have already been processed. Hence, it remains to show how to, in polynomial time, compute $c(v,Y)$ and $F^*(v,Y)$, assuming that we have already computed $c(v',Y')$ and $F^*(v',Y')$ for all pairs $(v',Y')$ with $v'\in U_v\setminus \{v\}$.

If $v$ is a leaf of $T$, then $L_v=L_v^{out}=\emptyset$ and $A_v=\emptyset$ and we have $c(v,\emptyset)=0$ and $F^*(v,\emptyset)=\emptyset$.
Next, assume that $v$ is not a leaf of $T$. Let $Y\subseteq L_v^{out}$ such that $|Y|\le N$. For $Z\subseteq L_v^{cross}$ with $|Y|+|Z|\le N$, we say that a link set $F\subseteq L_v$ is \emph{feasible} for $(v,Y,Z)$ if it is feasible for $(v,Y)$ and $F\cap L_v^{cross}=Z$. We denote the minimum cost of a link set that is feasible for $(v,Y,Z)$ by $c(v,Y,Z)$ and let $c(v,Y,Z)=\infty$ if no such link set exists. In addition to the values $c(v,Y,Z)$, we will compute link sets $F^*(v,Y,Z)$ such that $F^*(v,Y,Z)$ is a feasible link set for $(v,Y,Z)$ with $c(F^*(v,Y,Z))=c(v,Y,Z)$, if exists, and $F^*(v,Y,Z)=\emptyset$ if $c(v,Y,Z)=\infty$.
We have \[c(v,Y)=\min\{c(v,Y,Z)\colon Z\subseteq L_v^{cross},|Y|+|Z|\le N\}\] because if $F\subseteq L_v$ is feasible for $(v,Y,Z)$, then it is also feasible for $(v,Y)$, and conversely, if $F$ is feasible for $(v,Y)$, then $F$ is $N$-thin, so $N\ge |F\cap L_v^{out}|+|F\cap L_v^{cross}|= |Y|+|F\cap L_v^{cross}|$, and hence, $F$ is feasible for $(v,Y,F\cap L_v^{cross})$. Moreover, if $Z$ attains the above minimum, then we can set $F^*(v,Y)=F^*(v,Y,Z)$. As there is only a polynomial number of sets $Z\subseteq L_v^{cross}$ with $|Y|+|Z|\le N$, it suffices to show how to, for a fixed choice of $Z$, compute $c(v,Y,Z)$ and $F^*(v,Y,Z)$.

Let $v_1,\dots,v_k$ be the children of $v$ in $T$ (recall that $v$ is not a leaf). For $i\in\{0,\dots,k\}$, let $L_v^i$ be the set of links with at least one endpoint in $U_i\coloneqq \bigcup_{j=1}^i U_{v_i}$, i.e., $L_v^0=\emptyset$ and $L_v^k=L_v$. We call a link set $F\subseteq L_v^i$ \emph{feasible} for $(v,Y,Z,i)$ if $F$ is $N$-thin, $F\cap L_v^{out}=Y\cap L_v^i$, $F\cap L_v^{cross}=Z\cap L_v^i$, and $F$ covers every arc in $A_i\coloneqq \bigcup_{j=1}^i A_{v_j}\cup \{a_{v_j}\}$. (Recall that $a_{v_j}$ is the arc connecting $v_j$ to its parent $v$.) We define $c(v,Y,Z,i)$ to be the minimum cost of a feasible link set for $(v,Y,Z,i)$, or $\infty$, if no such link set exists. We will compute the values $c(v,Y,Z,i)$ for $i=0,\dots,k$, and, whenever $c(v,Y,Z,i)\ne \infty$, we will compute a link set $F^*(v,Y,Z,i)$ attaining $c(v,Y,Z,i)$; otherwise, we will set $F^*(v,Y,Z,i)=\emptyset$. 

Note that $c(v,Y,Z)=c(v,Y,Z,k)$ and that $F^*(v,Y,Z,k)$ is a feasible choice for $F^*(v,Y,Z)$. Hence, it remains to explain how to determine the values $c(v,Y,Z,i)$ and $F^*(v,Y,Z,i)$ in polynomial time.

For $i\in \{1,\dots,k\}$, let $\ell^*_i$ be the link with endpoints $v_i$ and $v$ that covers $a_{v_i}$, i.e., $\ell^*_i=(v,v_i)$ if $a_{v_i}=(v_i,v)$ and vice versa. Note that $\ell^*_i\in L$ by shadow-completeness and because there exists a link in $L$ covering $a_{v_i}$; otherwise, the instance is infeasible and we can return this information. Moreover, let $\mathcal{Y}_i$ be the collection of all sets $Y'\subseteq L_{v_i}^{out}$ such that $|Y'|\le N$ and $Y'\cap (L_v^{out}\cup L_v^{cross})=(Y\cup Z)\cap L_{v_i}^{out}$.
\begin{claim}
We have $c(v,Y,Z,0)=0$. For $i\in\{1,\dots,k\}$,
\begin{align*}c(v,Y,Z,i)&=c(v,Y,Z,i-1)+c(Y\cap (L_v^i\setminus L_v^{i-1}))+c(Z\cap (L_v^i\setminus L_v^{i-1}))\\
	&+\min\{c(v_i,Y')-c(Y'\cap (Y\cup Z))+\chi[\text{$a_{v_i}$ not covered by $Y'\cup Y\cup Z$}]\cdot c(\ell^*_i)\colon Y'\in\mathcal{Y}_i\},\end{align*}
where $\chi[\text{$a_{v_i}$ is not covered by $Y'\cup Y\cup Z$}]$ is $1$ if $a_{v_i}$ is not covered by $Y'\cup Y\cup Z$, and $0$ otherwise.
\end{claim}
\begin{proof}[Proof of claim]
As $L_v^0=\emptyset$, $c(v,Y,Z,0)=0$. Next, let $i\in\{1,\dots,k\}$. We first prove that every set $Y'\in\mathcal{Y}_i$ for which the right hand side is finite yields a valid upper bound on $c(v,Y,Z,i)$. To this end, assume that $c(v,Y,Z,i-1)<\infty$ and let $F'\coloneqq F^*(v,Y,Z,i-1)$ attain this value. Let further $Y'\in \mathcal{Y}_i$ such that $c(v_i,Y')<\infty$ and let $F_i\coloneqq F^*(v_i,Y')$. Let $F\coloneqq F'\cup F_i\cup (Y\cup Z)\cap L_v^i$, if $a_{v_i}$ is covered by $Y'\cup Y\cup Z$, and let $F\coloneqq F'\cup F_i\cup (Y\cup Z)\cap L_v^i\cup \{\ell^*_i\}$ otherwise.
Then \begin{equation}F\cap L_v^{i-1}=F' \text{ and } F\cap L_{v_i}=F_i \text{ and } F\cap (L_v^{out}\cup L_v^{cross})=(Y\cup Z)\cap L_v^i\label{eq:intersection_with_subinstances}\end{equation} by construction and by definition of $\mathcal{Y}_i$. Then $F$ covers every arc in $A_i$ because $F'$ covers every arc in $A_{i-1}$, $F_i$ covers every arc in $A_{v_i}$, and we also made sure that $a_{v_i}$ is covered. We further have $F\subseteq L_v^i$ by construction. To see that $F$ is $N$-thin, we note that for every vertex $v'\in U_{i-1}$, there are at most $N$ links in $F$ with $v'\in \mathrm{in}(P_\ell)$ because $F'$ is $N$-thin, by \eqref{eq:intersection_with_subinstances} and because if $v'\in \mathrm{in}(P_\ell)$ for $\ell\in F$, then $\ell\in L_v^{i-1}$. Similarly, for every vertex $v'\in U_{v_i}$, there are at most $N$ links in $F$ with $v'\in\mathrm{in}(P_\ell)$ by \eqref{eq:intersection_with_subinstances} and because $F_i$ is $N$-thin. There are at most $N$ links in $F$ with $v\in \mathrm{in}(P_\ell)$ by \eqref{eq:intersection_with_subinstances} and because $|Y|+|Z|\le N$. For $w\in U_v\setminus (U_i\cup \{v\})$, there are at most $N$ links in $F$ with $w\in \mathrm{in}(P_\ell)$ because $F\subseteq L^i_v$ and $|Z|\le N$. Finally, for $w\in V\setminus U_v$, there are at most $N$ links in $F$ with $w\in\mathrm{in}(P_\ell)$ because $|Y|\le N$.

It remains to show that the expression on the right-hand side that we evaluate yields an upper bound on $c(F)$. If we include the link $\ell^*_i$, then we add its cost. The cost of every link in $L_v^{i-1}\cap F=F'$ is added (via the term $c(v,Y,Z,i-1)$). 
The cost of every link in $F_i\setminus (Y\cup Z)$ is added via the term $c(v_i,Y')-c(Y'\cap (Y\cup Z))$ because
\[F_i\cap (Y\cup Z)=(F_i\cap L_{v_i}^{out})\cap (Y\cup Z)=Y'\cap (Y\cup Z).\]
Finally, the cost of every link in $(Y\cup Z)\cap (L_v^i\setminus L_v^{i-1})$ is added.

Next, we show that if $c(v,Y,Z,i)$ is finite, there exists a set $Y'\in\mathcal{Y}_i$ for which the value of the right-hand side is at most $c(v,Y,Z,i)$. To this end, let $F\subseteq L_v^i$ be feasible for $(v,Y,Z,i)$ with $c(F)=c(v,Y,Z,i)$. Define $F'\coloneqq F\cap L_v^{i-1}$ and $F_i\coloneqq F\cap L_{v_i}$, and let $Y'\coloneqq F\cap L_{v_i}^{out}$. Then $F'$ is feasible for $(v,Y,Z,i-1)$ and $F_i$ is feasible for $(v,Y')$. Moreover, $|Y'|\le N$ (as $F$ is $N$-thin) and 
\begin{align*}Y'\cap (L_v^{out}\cup L_v^{cross})&=F\cap L_{v_i}^{out}\cap (L_v^{out}\cup L_v^{cross})=L_{v_i}^{out}\cap (F\cap (L_v^{out}\cup L_v^{cross}))\\&=L_{v_i}^{out}\cap (Y\cup Z)\cap L_v^i=L_{v_i}^{out}\cap (Y\cup Z),\end{align*}
so $Y'\in\mathcal{Y}_i$.
It remains to show that the cost term that we get on the right-hand side when choosing $Y'$ is at most $c(F)$. To this end, we have \[c(v,Y,Z,i-1)\le c(F') \text{ and }c(v_i,Y')-c(Y'\cap (Y\cup Z))\le c(F_i)-c(F_i\cap (Y\cup Z))=c(F_i\setminus (Y\cup Z))\] because $Y'\cap (Y\cup Z)=(F\cap L_{v_i}^{out})\cap (Y\cup Z)=F_i\cap (Y\cup Z)$ and because $F'$ and $F_i$ are feasible for $(v,Y,Z,i-1)$ and $(v_i,Y')$, respectively. We note that the subsets $F'$, $F_i\setminus (Y\cup Z)$, $Y\cap (L^i_v\setminus L^{i-1}_v)$ and $Z\cap (L^i_v\setminus L^{i-1}_v)$ are pairwise distinct because $F'\cap F_i\subseteq Z$. Finally, we observe that none of the previous subsets can contain the link $\ell^*_i$ and that if $a_{v_i}$ is not covered by $Y'\cup Y\cup Z$, then $a_{v_i}$ can only be covered by $\ell^*_i$, so $\ell^*_i\in F$. This is because every link covering $a_{v_i}$ must have one endpoint in $U_{v_i}$ and its other endpoint in $V\setminus U_{v_i}$ and unless the endpoints are $v$ and $v_i$ (i.e., $\ell=\ell^*_i$), we have $\ell\in Y'\cup Y\cup Z$. 
\end{proof}
Using the claim, we can compute all of the values $c(v,Y,Z,i)$ in polynomial time. In the proof of the claim, we have further seen how to compute $F^*(v,Y,Z,i)$ attaining $c(v,Y,Z,i)$ in polynomial time. This concludes the proof.
\end{proof}

Combining the results of~\Cref{lem:kthin,lem:DPkthin}, we conclude that WDTAP instances of constant visible width can be solved exactly in polynomial time.

\begin{corollary}\label{cor:solve_DTAP_bounded_viwidth}
    If $(T, L, c,r)$ is a rooted instance of WDTAP with visible width at most $k$, then we can, in polynomial time, find an optimal solution, or decide that the instance is infeasible.
\end{corollary}
\begin{proof}
    For a feasible visibly $k$-wide instance, there exists an optimal solution that is at most $2k$-thin. Hence, we can run the dynamic programming algorithm above with $N = 2k$ to find the optimal solution for this instance, or decide that it is infeasible. 
\end{proof}

 \section{The partial separation framework}\label{sec:partialseparation}

This section describes the high level framework of our algorithm, which is to implement a partial separation oracle for a certain LP formulation we call the visibly $k$-wide modification LP. We then show how this partial separation oracle implies an algorithm for the WDTAP problem. 

 \subsection{Splitting links}
First, we formalize the ``link splitting" operation, which will be used throughout the paper and in particular will allow us to define the visibly $k$-wide modification LP. 

Fix a (rooted) WDTAP instance $(T,L,c,r)$.
\begin{definition}
 A \emph{splitting} of the link set $L$ is a function $\sigma\colon L\rightarrow 2^L$ mapping $\ell\in L$ to a set of shadows $\ell_1,\dots,\ell_t$ of $\ell$ such that (the arc sets of) $P_{\ell_1},\dots,P_{\ell_t}$ form a partition of (the arc set of) $P_\ell$. The \emph{support} of the splitting is $\mathrm{supp}(\sigma)\coloneqq \{\ell\in L\colon \exists \ell'\in L\colon \ell\in\sigma(\ell')\}.$   \end{definition}
Next, we define how to apply a splitting to a solution to \eqref{eq:WDTAP_LP} to generate a new feasible solution of \eqref{eq:WDTAP_LP}.
\begin{definition}
Let $x$ be a feasible solution to \eqref{eq:WDTAP_LP} and let $\sigma$ be a splitting of $L$. We let the solution $x'=\mathrm{split}(x,\sigma)$ to \eqref{eq:WDTAP_LP} that we \emph{obtain from $x$ by applying $\sigma$} be defined by $x'_{\ell'}\coloneqq \sum_{\ell\in L\colon \ell'\in\sigma(\ell)} x_{\ell}$.
\end{definition}
 The following proposition shows that splitting links can only reduce the visible up- or down-width of any vertex.
To state it, we introduce the following notation.
 \begin{definition}
 	Let $x$ be a solution to \eqref{eq:WDTAP_LP}. The \emph{support} $\mathrm{supp}(x)$ of $x$ consists of all links $\ell$ with $x_\ell > 0$.
 \end{definition}
 
\begin{proposition}\label{lemma:splitting_cannot_increase_width}
	Let $x$ be a solution to \eqref{eq:WDTAP_LP}, let $\sigma$ be a splitting of $L$ and let $x'\coloneqq \mathrm{split}(x,\sigma)$. Then for every vertex $v$, the visible up-width (visible down-width) of $v$ with respect to $\mathrm{supp}(x')$ is at most the visible up-width (visible down-width) of $v$ with respect to $\mathrm{supp}(x)$.
\end{proposition}

\begin{proof}
	It suffices to show that every arc that is visible for a vertex $v$ with respect to $\mathrm{supp}(x')$ is also visible for that vertex with respect to $\mathrm{supp}(x)$.
	Let $a$ be an arc that is visible for $v$ with respect to $\mathrm{supp}(x')$. Then there is $\ell'\in \mathrm{supp}(x')$ such that $a\in\covr(\ell')$ and $v\in\mathrm{in}(\overline{P}_{\ell'})$. As $\ell'\in \mathrm{supp}(x')$, there is $\ell\in\mathrm{supp}(x)$ such that $\ell'\in\sigma(\ell)$. Then $\ell'$ is a shadow of $\ell$, so $a\in\covr(\ell)$ and $v\in\mathrm{in}(\overline{P}_{\ell})$. Hence, $a$ is also visible for $v$ with respect to $\mathrm{supp}(x)$. 
\end{proof}

The following proposition shows that the coverage of all tree arcs is preserved by the splitting operation.

\begin{proposition}\label{lemma:splitting_same_coverage}
Let $x$ be a feasible solution to \eqref{eq:WDTAP_LP}, let $\sigma$ be a splitting of $L$ and let $x'\coloneqq \mathrm{split}(x,\sigma)$. For $a\in A$, we have
\begin{itemize}
    \item $x'(\{\ell\in L\colon a\in\covr(\ell)\})=x(\{\ell\in L\colon a\in\covr(\ell)\})$,
    \item $x'(\{\ell\in L\colon a\in\covw(\ell)\})=x(\{\ell\in L\colon a\in\covw(\ell)\})$ and 
    \item $x'(\{\ell\in L\colon a\in\cov(\ell)\})=x(\{\ell\in L\colon a\in\cov(\ell)\})$.
\end{itemize}
\end{proposition}
\begin{proof}
Let $a\in A$. We only show the first equality, the other ones can be derived analogously. For any link $\ell \in L$ with $a\in \covr(\ell)$, there exists a unique link $\ell_a \in \sigma(\ell)$ with $a\in\covr(\ell_a)$. Conversely, if $a\in\covr(\ell')$ and $\ell'\in\sigma(\ell)$, then $a\in\covr(\ell)$ (and $\ell'=\ell_a$). This implies
\begin{align*}x'(\{\ell\in L\colon a\in\covr(\ell)\})&=\sum_{\ell'\in L\colon a\in\covr(\ell')}\sum_{\ell\in L\colon \ell'\in \sigma(\ell)} x_\ell=\sum_{\ell\in L\colon a\in\covr(\ell)}\underbrace{|\{\ell'\in\sigma(\ell)\colon a\in\covr(\ell')\}|}_{=1} \cdot x_\ell\\
&= x(\{\ell\in L\colon a\in\covr(\ell)\}).\end{align*}
\end{proof}

The following proposition simply counts the additional cost incurred by splitting.

\begin{proposition} \label{lemma:splitting_LP_solution}
Let $x$ be a feasible solution to \eqref{eq:WDTAP_LP}, let $\sigma$ be a splitting of $L$ and let $x'\coloneqq \mathrm{split}(x,\sigma)$. Then $x'$ is a feasible solution to \eqref{eq:WDTAP_LP} of cost 
\[c(x')=\sum_{\ell\in L} \left(\sum_{\ell'\in \sigma(\ell)} c(\ell')\right)\cdot x_\ell.\]
\end{proposition}
\begin{proof}
Feasibility of $x'$ follows from \cref{lemma:splitting_same_coverage}. For the cost, we calculate
\[\sum_{\ell\in L} c(\ell)\cdot x'_{\ell}=\sum_{\ell\in L} \sum_{\ell'\in \sigma(\ell)} c(\ell')\cdot x_\ell.\]
\end{proof}

We will often apply splittings sequentially, which is captured by the following definition.
\begin{definition}
  Let $\sigma$ and $\sigma'$ be two splitting of $L$. Then \emph{concatenation} $\sigma'\circ\sigma$ of the two splittings is defined via $(\sigma'\circ\sigma)(\ell)=\bigcup_{\ell'\in\sigma(\ell)}\sigma'(\ell')$. 
\end{definition}
Note that the concatenation of two splittings of $L$ is again a splitting of $L$. We further observe the following.
\begin{proposition}
 Let $x$ be a feasible solution to \eqref{eq:WDTAP_LP} and let $\sigma$ and $\sigma'$ be two splitting of $L$.
 Then $\mathrm{split}(x,\sigma'\circ\sigma)=\mathrm{split}(\mathrm{split}(x,\sigma),\sigma')$.
\end{proposition}
\begin{proof}
Let $x'\coloneqq \mathrm{split}(x,\sigma)$ and $x''\coloneqq \mathrm{split}(\mathrm{split}(x,\sigma),\sigma')$.
For $\ell''\in L$, we have
\[x''_{\ell''}=\sum_{\ell'\in L\colon \ell''\in\sigma'(\ell')} x'_{\ell'}=\sum_{\ell'\in L\colon \ell''\in\sigma'(\ell')}\sum_{\ell\in L\colon \ell'\in\sigma(\ell)} x_\ell=\sum_{\ell\in L\colon \ell''\in (\sigma'\circ\sigma)(\ell)} x_\ell.\]
For the last equality, we used that for $\ell\in L$, $\sigma(\ell)$ consists of shadows of $\ell$ with pairwise disjoint undirected coverages. In particular, we can have $\ell''\in\sigma'(\ell')$ for at most one $\ell'\in\sigma(\ell)$ because $\ell''$ has to be a shadow of $\ell'$.
\end{proof}
The following proposition helps us to bound the cost increase incurred by splittings.
\begin{proposition}\label{lemma:cost_increase_concatenate_splittings}
	Let $\Delta > 1$ and assume that $c\colon L\rightarrow [1,\Delta]$. Let $x$ be a solution to \eqref{eq:WDTAP_LP} and let $\sigma$ be a splitting of $L$. Let $x'\coloneqq \mathrm{split}(x,\sigma)$. 
	Then \[ c(x')\le c(x)+\sum_{\ell\in L} (|\sigma(\ell)|-1)\cdot c(\ell)\cdot x_\ell\le c(x)+\Delta\cdot \sum_{\ell\in L} (|\sigma(\ell)|-1)\cdot x_\ell.\]
\end{proposition}
\begin{proof}
We have \begin{align*}c(x')&=\sum_{\ell\in L}\sum_{\ell'\in\sigma(\ell)} c(\ell')\cdot x_\ell \le \sum_{\ell\in L}|\sigma(\ell)|\cdot c(\ell)\cdot x_\ell\\
	&= c(x)+\sum_{\ell\in L} (|\sigma(\ell)|-1)\cdot c(\ell)\cdot x_\ell\le c(x)+\Delta\cdot \sum_{\ell\in L} (|\sigma(\ell)|-1)\cdot x_\ell,\end{align*}
where the first inequality follows from the fact that $c(\ell')\le c(\ell)$ whenever $\ell'$ is a shadow of $\ell$, and the second inequality follows from $|\sigma(\ell)|\ge 1$ and $c(\ell)\le \Delta$.
\end{proof}
\subsection{The visibly $k$-wide modification LP}
Using splittings, we will introduce a new type of valid inequality for the integer hull of \eqref{eq:WDTAP_LP}. To define it, we need to consider subinstances that arise by contracting certain arcs.
Given an arc set $A^*$, we denote by $T/A^*$ the tree that arises from $T$ by contracting the arcs in $A^*$. For a link set $L^*$, $L^*/A^*$ denotes the link set arising from this contraction.
 For a function $f: A \to B$, and $C \subseteq A$, we use the notation $f \upharpoonright_ C$ to denote the restriction of $f$ to the domain $C$.
\begin{lemma}\label{lemma:k_wide_modification}
 Let $\sigma$ be any splitting of the link set and let $A'\subseteq A$. Then
 \begin{equation}\sum_{\ell\in L} \left(\sum_{\ell'\in\sigma(\ell)}c(\ell')\right)\cdot  x_\ell\ge c(OPT(T/A',\mathrm{supp}(\sigma)/A', c\upharpoonright_{\mathrm{supp}(\sigma)}))\label{eq:k_wide_modification}\end{equation} is a valid constraint for the integer hull of (\ref{eq:WDTAP_LP}), where $OPT(T/A',\mathrm{supp}(\sigma)/A', c\upharpoonright_{\mathrm{supp}(\sigma)})$ denotes an optimum solution to the WDTAP instance $(T/A',\mathrm{supp}(\sigma)/A', c\upharpoonright_{\mathrm{supp}(\sigma)})$.
 \end{lemma}

 \begin{proof}
Let $x$ be an integral solution to \eqref{eq:WDTAP_LP} and let $x'\coloneqq \mathrm{split}(x,\sigma)$. Then $x'$ is an integral solution to \eqref{eq:WDTAP_LP} with $\mathrm{supp}(x')\subseteq \mathrm{supp}(\sigma)$. In particular, $\mathrm{supp}(x')/A'$ is a feasible solution to $(T/A',\mathrm{supp}(\sigma)/A',c\upharpoonright_{\mathrm{supp}(\sigma)})$ of cost at most $c(x')= \sum_{\ell\in L} \left(\sum_{\ell'\in\sigma(\ell)}c(\ell')\right)\cdot  x_\ell$ by \cref{lemma:splitting_LP_solution}.
 \end{proof}
 \begin{definition}\label{def:visibly_k_wide_modification}
  A \emph{visibly $k$-wide modification} is a pair $(\sigma,A')$, where $\sigma$ is a splitting of the link set and $A'\subseteq A$, such that $(T/A',\mathrm{supp}(\sigma)/A',r)$ has visible width at most $k$. We call the corresponding contraint \eqref{eq:k_wide_modification} a \emph{visibly $k$-wide modification inequality}.
 \end{definition}
 Our approach will be to observe certain solutions to the linear program (\ref{eq:WDTAP_LP}) and to obtain an integral solution of relatively low cost, or to find a visibly $k$-wide modification inequality violated by the current solution to add to the constraints of (\ref{eq:WDTAP_LP}).
 
 \subsection{Proof of \cref{theorem:main_result}}
 The main technical theorem of this paper guarantees the existence of a partial separation oracle for the visibly $k$-wide-modification LP. This theorem is stated below, and in this subsection we will show how to use it to prove \cref{theorem:main_result}.
 \begin{restatable}{theorem}{theoremoracle}\label{theorem:partial_separation_oracle}
     Let $\bar{\epsilon},\Delta>0$. We can compute a constant $k(\bar{\epsilon},\Delta)$ with the following property: Given a rooted instance $(\bar{T},\bar{L},\bar{c},\bar{r})$ of WDTAP with cost ratio at most $\Delta$ and a feasible solution $\bar{x}$ to \eqref{eq:WDTAP_LP}, we can, in polynomial time, either find a solution $S\subseteq \bar{L}$ with $\bar{c}(S)\le (1.75+\bar{\epsilon})\cdot \bar{c}(\bar{x})$, or find a visibly $k(\bar{\epsilon},\Delta)$-wide modification inequality that is violated by $\bar{x}$.
 \end{restatable}
 Assuming \cref{theorem:partial_separation_oracle}, we are now ready to prove \cref{theorem:main_result}, which we restate for convenience.
 \mainresult*
 \begin{proof}
 	Let $\Delta \ge 1$ and let $\epsilon >0$. We may assume that the constants $\epsilon$ and $\Delta$ are rational numbers because we can replace them with rational constants $\epsilon'$ and $\Delta'$ with $0<\epsilon'<\epsilon$ and $\Delta <\Delta'$ otherwise.
 	 Fix a rooted WDTAP instance $(T,L,c,r)$ with cost ratio at most $\Delta$. We can check in polynomial time if $(T,L,c,r)$ is feasible by checking if each tree arc is covered by at least one link in $L$. Hence, we will assume that $(T,L,c,r)$ is feasible in the following. By re-scaling the costs, we may assume $c\colon L\rightarrow[1,\Delta]$. Then the cost of an optimal solution $OPT$ satisfies $1 \leq c(OPT) \leq \Delta |L| \leq \Delta n^2$, where $n$ is the number of vertices of $T$. Let $\bar{\epsilon}\coloneqq \min\{1,\frac{\epsilon}{10}\}$, let $k\coloneqq k(\bar{\epsilon},\Delta)$ and let $M\coloneqq \lceil \log_{1+\bar{\epsilon}} n^2\Delta\rceil$. 
  We will use binary search on the interval $[1,(1+\bar{\epsilon})^M]$. Note that the runtime of the algorithm will depend on $\epsilon$ and $\Delta$.

In the following, we will describe a subroutine that, given a rational number $c^*\in[1,(1+\bar{\epsilon})^M]$, in polynomial time (in the encoding lengths of $(T,L,c,r)$, $\Delta$, $\epsilon$ and $c^*$) either returns a solution $F$ to $(T,L,c,r)$ with $c(F)\le (1.75+\bar{\epsilon})\cdot(1+\bar{\epsilon})\cdot c^*$, or decides that $c^*<c(OPT)$. The subroutine is defined as follows.
    Given $c^*$, we apply the ellipsoid method to (try to) find a feasible point $x$ in the polyhedron $P$ given by the constraints in \eqref{eq:WDTAP_LP}, all visibly $k$-wide modification inequalities, and $c(x)\le (1+\bar{\epsilon})\cdot c^*$. Note that the encoding length of every constraint is polynomially bounded in the encoding lengths of $(T,L,c,r)$, $\epsilon$ and $c^*$. Moreover, $P\subseteq [0,(1+\bar{\epsilon})\cdot c^*]^L$ since $c(\ell)\ge 1$ for every $\ell\in L$. If $c^*\ge c(OPT)$, then we further have $OPT+[0,\frac{\bar{\epsilon}}{\Delta\cdot |L|}\cdot c^*]^L\subseteq P$, where we interpret $OPT$ as a vector in $\{0,1\}^L$. Finally, we can separate all constraints in \eqref{eq:WDTAP_LP}, as well as the constraint $c(x)\le (1+\bar{\epsilon})\cdot c^*$, in polynomial time. To separate the visibly $k$-wide modification inequalities, we will use \Cref{theorem:partial_separation_oracle}.
    
      More precisely, in each iteration of the ellipsoid method, given $y\in\mathbb{Q}^L$, we do the following: If $y$ violates any of the constraints of \eqref{eq:WDTAP_LP} or $c(y)> (1+\bar{\epsilon})\cdot c^*$, we return the corrresponding violated constraint. Otherwise, we apply \Cref{theorem:partial_separation_oracle} to either find a WDTAP solution $F$ with $c(F)\le (1.75+\bar{\epsilon})\cdot c(y)\le (1.75+\bar{\epsilon})\cdot(1+\bar{\epsilon})\cdot c^*$, or a violated visibly $k$-wide modification inequality. In the first case, we return $F$ and stop; in the second case, we continue the ellipsoid method. After a polynomial number of iterations, we have either found a WDTAP solution $F$ with $c(F)\le (1.75+\bar{\epsilon})\cdot(1+\bar{\epsilon})\cdot c^*$, or the volume of the ellipsoid is small enough, allowing us to deduce that $c^*<c(OPT)$.

    Throughout the binary search, we maintain an interval $[(1+\bar{\epsilon})^a, (1+\bar{\epsilon})^b]$ such that $c(OPT) \ge (1+\bar{\epsilon})^a$ and we have a WDTAP solution $F$ with $c(F)\le (1.75+\bar{\epsilon})\cdot (1+\bar{\epsilon})^{b+1}$. Continually checking the point $c^* = (1+\bar{\epsilon})^{\lfloor\frac{a+b}{2}\rfloor}$, we obtain an interval of the form $[(1+\bar{\epsilon})^t, (1+\bar{\epsilon})^{t+1}]$. In this case, we are guaranteed an integral solution $F$ of cost at most \begin{align*}c(F)&\leq (1.75+\bar{\epsilon})\cdot (1+\bar{\epsilon})^{t+2}\le (1.75+\bar{\epsilon})\cdot (1+\bar{\epsilon})^{2}\cdot c(OPT)\\
    	&\le (1.75+4.5\bar{\epsilon}+3.75\bar{\epsilon}^2+\bar{\epsilon}^3)\cdot c(OPT)\le (1.75+\epsilon)\cdot c(OPT).\end{align*}
\end{proof}
We remark that the partial separation framework has already been used in \cite{adjiashvili2018beating,grandoni2018improved}.
The remainder of the main part of this paper is dedicated to proving \cref{theorem:partial_separation_oracle}.
\section{Proving the weakened dream theorem}\label{sec:proof_weak_dream}
In this section, we prove the weakened dream theorem (\cref{theorem:weak_dream}). To this end, we will first introduce some notation that allows us to state \cref{theorem:first_phase_splitting}, a slightly more formal version of \cref{theorem:weak_dream}. 
\begin{definition}
Let $(T,L,c,r)$ be a rooted instance of WDTAP, let $x$ be a solution to \eqref{eq:WDTAP_LP} and let $\alpha\ge 0$. We call an arc $a$ \emph{$\alpha$-covered} if $x(\{\ell\colon a\in\covr(\ell)\})\ge \alpha$ and $\alpha$-heavy if $x(\{\ell\colon a\in\covw(\ell)\})\ge \alpha$. We call a link $\ell$ \emph{$\alpha$-heavily involved} if there is an $\alpha$-heavy arc $a$ with $a\in\covw(\ell)$.
\end{definition}

Let $\epsilon \in (0,1)$ and let $\Delta > 0$. We define the following constants:
\begin{itemize}
	\item $\gamma\coloneqq \frac{\epsilon}{2\Delta}$ is used to define whether an arc $a$ is lightly covered, allowing us to cheaply split links $\ell$ with $a\in\cov(\ell)$.
	\item $\zeta_1\coloneqq \frac{2}{\epsilon}$ is our threshold for an arc to be ``heavily covered in the right direction'', allowing us to contract it.
	\item $\zeta_2\coloneqq \frac{6\cdot \zeta_1\cdot \Delta}{\epsilon\cdot (1-\epsilon)}$ is our our threshold for an arc to be ``heavily covered in the wrong direction''.
	\item $k\coloneqq (1+\gamma^{-1})\cdot \zeta_2$ is our bound on the visible width of certain instances that we will target.
\end{itemize}
Observe that our choice of constants satisfies the following inequalities. In fact, all but \eqref{eq:constants_3} are actually equalities, but we do not need that fact.
\begin{align}
2\cdot\Delta\cdot \gamma &\leq \epsilon \label{eq:constants_1}\\
\frac{2}{\epsilon} &\leq \zeta_1 \label{eq:constants_5}\\
\zeta_1 & < \epsilon\cdot\zeta_2 \label{eq:constants_3}\\
3\cdot \zeta_1\cdot \Delta&\le \epsilon\cdot\frac{1}{2}\cdot(1-\epsilon)\cdot\zeta_2\label{eq:constants_4}\\
(1+\gamma^{-1})\cdot \zeta_2 &\le k \label{eq:constants_2}
\end{align}
\begin{lemma}\label{lemma:solution_for_zeta_1_covered_arcs}
Let $(T,L,c,r)$ be a rooted instance of WDTAP, let $x$ be a solution to \eqref{eq:WDTAP_LP} and let $F'$ be a feasible solution to the instance we obtain after contracting all $\zeta_1$-covered arcs. Then we can, in polynomial time, compute a solution $F$ of cost $c(F)\le c(F')+\epsilon\cdot c(x)$ to the original instance.
\end{lemma}
\begin{proof}
    We show how to, in polynomial time, compute a link set $F''$ of cost $c(F'')\le \epsilon\cdot c(x)$ that covers all $\zeta_1$-covered arcs. Let $(\bar{T}=(\bar{V},\bar{A}),\bar{L},\bar{c},\bar{r})$ arise from $(T,L,c,r)$ by contracting all arcs that are not $\zeta_1$-covered. Then $x$ corresponds to a solution $\bar{x}$ of cost $\bar{c}(\bar{x})=c(x)$ to \eqref{eq:WDTAP_LP} for $(\bar{T},\bar{L},\bar{c})$ with the property that $\bar{x}(\ell\in \bar{L}\colon a\in\covr(\ell)\ge \zeta_1$ for every $a\in \bar{A}$. In particular, $x'\coloneqq \frac{1}{\zeta_1}\cdot \bar{x}$ is a feasible solution to \eqref{eq:WDTAP_LP} for $(\bar{T},\bar{L},\bar{c})$ as well.
    Obtain $x''$ from $x'$ by splitting every link in $\mathrm{supp}(x')$ that is not an up- or down-link already at its apex. Then \[\bar{c}(x'')\le 2\cdot \bar{c}(x')=\frac{2}{\zeta_1}\cdot c(x)\le \epsilon\cdot c(x)\] by \eqref{eq:constants_5}. Note that $(\bar{T},\mathrm{supp}(x''),\bar{c},\bar{r})$ is a willow (choosing $U=\emptyset$), so we can, in polynomial time, compute an optimum solution $\bar{F}$ to $(\bar{T},\mathrm{supp}(x''),\bar{c})$ of cost at most $\bar{c}(x'')\le \epsilon\cdot c(x)$ by \cref{theorem:unimodularity}. The uncontracted links corresponding to $\bar{F}$ yield the desired link set $F''$. Setting $F=F'\cup F''$ concludes the proof.
\end{proof}

In the following, it will be convenient to make the following assumption.
\begin{equation}
\text{There are no $\zeta_1$-covered arcs.} \label{eq:no_covered_arcs}
\end{equation}
\cref{lemma:solution_for_zeta_1_covered_arcs} essentially tells us that we can assume \eqref{eq:no_covered_arcs} at the cost of a cost increase by $\epsilon\cdot c(x)$.

Before stating \cref{theorem:first_phase_splitting}, the more formal version of \cref{theorem:weak_dream}, we need to introduce the following notation:
\begin{itemize}
    \item For a vertex $v\in V\setminus\{r\}$, we let $a_v$ be the arc connecting $v$ to its parent.
    \item For a vertex $v\in V\setminus \{r\}$ such that $a_v$ is an up-arc, we define $\viwr(v)\coloneqq \viwu(v)$, $\viww(v)\coloneqq \viwd(v)$, $\overrightarrow{A}_v\coloneqq A_v\cap A_{up}$ and $\overleftarrow{A}_v\coloneqq A_v\cap A_{down}$.
    \item For a vertex $v$ such that $a_v$ is a down-arc, we let $\viwr(v)\coloneqq\viwd(v)$, $\viww(v)\coloneqq \viwu(v)$, $\overrightarrow{A}_v\coloneqq A_v\cap A_{down}$ and $\overleftarrow{A}_v\coloneqq A_v\cap A_{up}$.
\end{itemize}
\begin{theorem}\label{theorem:first_phase_splitting}
Let $(T,L,c,r)$ be an instance of WDTAP with cost ratio at most $\Delta$ and let $x$ be a solution to \eqref{eq:WDTAP_LP} satisfying \eqref{eq:no_covered_arcs}.
We can, in polynomial time, compute a splitting $\sigma^*$ of $L$ and a vertex set $W^*\subseteq V$ with the following properties:
\begin{enumerate}[label=(\roman*)]
    \item \label{item:first_phase_cost_bound} Let $x^*\coloneqq \mathrm{split}(x,\sigma^*)$. We have $c(x^*)\le (1+\epsilon)\cdot c(x)$.
    \item \label{item:first_phase_W} $W^*$ consists of up- and down-independent vertices with respect to $\mathrm{supp}(x^*)$.
    \item \label{item:first_phase_viwidth} Let $L'$ arise from $\mathrm{supp}(x^*)$ by splitting every $W^*$-cross-link at its apex. With respect to $L'$, we have the following:
    \begin{enumerate}
        \item $\viwr(v)\le k$ for every $v\in V$.
        \item $\viww(r)\le k$ and we have $\viww(v)\le k$ for every $v\in V\setminus \{r\}$ such that $a_v$ is not $\zeta_2$-heavy (with respect to $x^*$).
    \end{enumerate}
\end{enumerate}
\end{theorem}
Note that when saying that $L'$ arises from $\mathrm{supp}(x^*)$ by splitting every $W^*$-cross-link at its apex, we mean the following: There exists a splitting $\sigma$ of $L$ such that $L'=\mathrm{supp}(\mathrm{split}(x^*,\sigma))$ and such that for every $W^*$-cross-link $\ell$, $\sigma(\ell)$ consists of up- and down-links only, i.e., $\ell$ has been split at its apex (and potentially at further vertices).

The rest of this section is dedicated to proving \cref{theorem:first_phase_splitting}.
Fix a rooted WDTAP instance $(T=(V,A),L,c,r)$ with cost ratio at most $\Delta$. By rescaling the costs, we may assume without loss of generality that $c\colon L\rightarrow[1,\Delta]$.
To establish \cref{theorem:first_phase_splitting}, we will traverse the tree $T$ bottom-up, starting from the leaves and working our way up towards the root. Whenever we encounter a vertex $v$ of high visible width, we will try to split links with one endpoint in $T_v$ and one endpoint outside $T_v$, rendering $v$ up- or down-independent. We introduce the following notation, which slightly differs from the one used in \Cref{sec:ourcontribution}.
\begin{definition}\label{def:L_uparrow_downarrow}
Let $v\in V\setminus \{r\}$. We say that a link $\ell=(u,w)$ \emph{points into $T_v$} if $w\in U_v\setminus \{v\}$ and $u\notin U_v$. We say that $\ell$ \emph{points out of $T_v$} if $u\in U_v\setminus\{v\}$ and $w\notin U_v$.
We denote the set of links pointing into/ out of $T_v$ by $L^\downarrow_v$ and $L^\uparrow_v$, respectively.
\end{definition} 
\begin{proposition}\label{lemma:up-down-independent}
If $L^\downarrow_v=\emptyset$, then $v$ is up-independent. If $L^\uparrow_v=\emptyset$, then $v$ is down-independent.
\end{proposition}
\begin{proof}
We only prove the first statement, the second one can be derived analogously. Assume $L^\downarrow_v=\emptyset$ and let $\ell=(y,z)\in L$. Assume $\covr(\ell)\cap A_v\cap A_{up}\ne \emptyset$. Then $z\in U_v\setminus \{v\}$. As $\ell\notin L^\uparrow_v=\emptyset$, $z\in U_v$. Hence, $\covr(\ell)\subseteq \cov(\ell)\subseteq A_v$.
\end{proof}
We are now ready to define when an arc $a_v$ is ``light'' with respect to a solution to \eqref{eq:WDTAP_LP}, allowing us to split all links covering it in the right or in the wrong direction, respectively, without increasing the cost of the LP solution by too much. As outlined in \Cref{sec:ourcontribution}, we will charge the cost of the splitting to the coverage of visible arcs in the subtree hanging off $v$. In doing so, it will be convenient to measure the coverage of these arcs with respect to the \emph{original LP solution $x$}, whilst defining visibility with respect to the support $L'$ of the \emph{split LP solution $x'$}.
\begin{definition}\label{def:up_down_light}
Let $\gamma\in (0,1)$, let $x$ be a solution to \eqref{eq:WDTAP_LP} and let $L'\subseteq L$. Let $v\in V\setminus \{r\}$. We say that $a_v$ is \emph{$\gamma$-up-light} with respect to $x$ and $L'$ if 
\[x(L^\downarrow_v)\le \gamma\cdot x(\{\ell\in L\colon \covr(\ell)\cap A_{up}\cap A^{vis}_v(L')\ne \emptyset\}\setminus L^\downarrow_v).\]
We say that $a_v$ is \emph{$\gamma$-down-light} with respect to $x$ and $L'$ if 
\[x(L^\uparrow_v)\le \gamma\cdot x(\{\ell\in L\colon \covr(\ell)\cap A_{down}\cap A^{vis}_v(L')\ne \emptyset\}\setminus L^\uparrow_v).\]
\end{definition}
Note that $L'$ is only used to specify visibility, however, we evaluate $x$ on all of $L$.

Next, we define the type of splitting operation that we will perform when encountering a light arc.

\begin{definition}
Let $v\in V$ and let $L'\subseteq L$. The splitting $\sigma_{v,L'}$ is defined as follows. For $\ell=(u,w)\in L'$ with $v\in \mathrm{in}(P_\ell)$, we define $\sigma_{v,L'}(\ell)=\{(u,v),(v,w)\}$. For every other link $\ell$, we define $\sigma_{v,L'}(\ell)=\{\ell\}$.
\end{definition}
Recall that $\mathrm{in}(P_\ell)$ denotes the set of inner vertices of the path $P_\ell$. The splitting $\sigma_{v,L'}$ splits every link $\ell\in L'$ with $v\in \mathrm{in}(P_\ell)$ into two shadows; one starting and one ending in $v$.

The next proposition allows us to bound the cost increase incurred by successive splitting operations.
\begin{proposition}\label{lemma:cost_increase_splitting_new}
Let $x$ be a solution to \eqref{eq:WDTAP_LP}, let $v\in V$ and let $L'\subseteq L$. Let $x'\coloneqq \mathrm{split}(x,\sigma_{v,L'})$. Then
\[c(x')\le c(x)+\Delta\cdot x(L').\]
\end{proposition}
\begin{proof}
This follows from \cref{lemma:cost_increase_concatenate_splittings} by observing that $|\sigma_{v,L'}(\ell)|\le 2$ for $\ell\in L'$ and $|\sigma_{v,L'}(\ell)|=1$ for $\ell\notin L'$.
\end{proof}
The link sets that we will choose as $L'$ will be of the form $L^\downarrow_v$ and $L^\uparrow_v$, respectively. The following lemma tells us that splitting cannot increase the total $x$-value on these subsets. (It can, however, decrease it to zero if splits are performed at $v$.)
\begin{proposition}\label{lemma:splitting_decreases_L_up_down}
Let $x$ be a solution to \eqref{eq:WDTAP_LP}, let $\sigma$ be a splitting of $x$ and let $x'\coloneqq \mathrm{split}(x,\sigma)$. Let $v\in V\setminus\{r\}$.
\begin{itemize}
	\item  We have $x'(L^\downarrow_v)\le x(L^\downarrow_v)$ and $x'(L^\uparrow_v)\le x(L^\uparrow_v)$.
	\item If $\sigma=\sigma_{w,L'}$ and $w\ne v$ or $L'\cap L^\downarrow_v=\emptyset$, then $x'(L^\downarrow_v)= x(L^\downarrow_v)$.
	\item If $\sigma=\sigma_{w,L'}$ and $w\ne v$ or $L'\cap L^\uparrow_v=\emptyset$, then $x'(L^\uparrow_v)= x(L^\uparrow_v)$.
\end{itemize}

\end{proposition}
\begin{proof}
We only prove the statements for $L^\downarrow_v$, the proof for $L^\uparrow_v$ is analogous. Let $\ell\in L$. We make the following two observations:
\begin{itemize}
	\item If there is $\ell'\in\sigma(\ell)\cap L^\downarrow_v$, then $\ell\in L^\downarrow_v$ because $\ell'$ is a shadow of $\ell$.
	\item For $\ell\in L^\downarrow_v$, we have $|\sigma(\ell)\cap L^\downarrow_v|\le 1$ because $a_v\in\cov(\ell')$ for every $\ell'\in \sigma(\ell)\cap L^\downarrow_v$, but the sets $(\cov(\ell'))_{\ell'\in \sigma(\ell)}$ are pairwise disjoint. (Recall that $\cov(\ell')$ is the arc set of $P_{\ell'}$).
\end{itemize}
This yields
\[
	x'(L^\downarrow_v)=\sum_{\ell'\in L^\downarrow_v}\sum_{\substack{\ell\in L\colon\\ \ell'\in \sigma(\ell)}} x_\ell=\sum_{\ell'\in L^\downarrow_v}\sum_{\substack{\ell\in L^\downarrow_v\colon\\ \ell'\in \sigma(\ell)}} x_\ell=\sum_{\ell\in L^\downarrow_v} |\sigma(\ell)\cap L^\downarrow_v|\cdot x_\ell\le x(L^\downarrow_v),
\]
proving the first statement (for $L^\downarrow_v$). We note that if $\sigma=\sigma_{w,L'}$ and $w\ne v$ or $L'\cap L^\downarrow_v=\emptyset$, then for every $\ell\in L^\downarrow_v$, there is exactly one $\ell'\in\sigma(\ell)$ with $\ell'\in L^\downarrow_v$ and we get equality above.
\end{proof}
\cref{algorithm:light_link_splitting} shows the splitting procedure that we employ in order to prove \cref{theorem:first_phase_splitting}. We traverse the vertices in $V\setminus\{r\}$ from the leaves towards the root. If $a_v$ is $\gamma$-up-light, we split all links pointing into $T_v$ at $v$, if $a_v$ is $\gamma$-down-light, we split links pointing out of $T_v$. Throughout the algorithm, we keep track of the current (split) LP solution $x^*$, the splitting $\sigma^*$ with $x^*=\mathrm{split}(x,\sigma^*)$ and the sets $W_{up}$ and $W_{down}$ of vertices $v$ at which links pointing into and out of $T_v$ have been split, respectively. We write $\sigma_{\rm{id}}$ to denote the initial \emph{identity splitting} given by $\sigma_{\rm{id}}(\ell)=\{\ell\}$ for every $\ell\in L$. Note that $x=\mathrm{split}(x,\sigma_{\rm{id}})$.
\begin{algorithm}
\begin{algorithmic}[1]
	\Require{solution $x$ to \eqref{eq:WDTAP_LP}}
	\Ensure{splitting $\sigma^*$, $x^*=\mathrm{split}(x,\sigma^*)$, vertex set $W^*$}
	
	\State $\sigma^*\gets \sigma_{\rm{id}}$, $x^*\gets x$, $W_{up}\gets\emptyset$, , $W_{down}\gets\emptyset$
	\For{$v\in V\setminus \{r\}$ in order of non-increasing distance to $r$}
	\If{$a_v$ is $\gamma$-up-light (with respect to $x$ and $\mathrm{supp}(x^*)$)}
	\State $\sigma^*\gets \sigma_{v,L^\downarrow_v}\circ\sigma^*$, $x^*\gets \mathrm{split}(x^*,\sigma_{v,L^\downarrow_v})$
	\State $W_{up}\gets W_{up}\cup \{v\}$
	\EndIf
	\If{$a_v$ is $\gamma$-down-light (with respect to $x$ and $\mathrm{supp}(x^*)$)}
	\State $\sigma^*\gets \sigma_{v,L^\uparrow_v}\circ\sigma^*$, $x^*\gets \mathrm{split}(x^*,\sigma_{v,L^\uparrow_v})$
	\State $W_{down}\gets W_{down}\cup \{v\}$
	\EndIf
	\EndFor
	\State \Return{$\sigma^*$, $x^*$, $W^*\coloneqq W_{up}\cup W_{down}\cup \{r\}$}
\end{algorithmic}
\caption{Light link splitting.\label{algorithm:light_link_splitting}}	
\end{algorithm}
Note that \cref{algorithm:light_link_splitting} runs in polynomial time. We will show that the output $(\sigma^*,x^*,W^*)$ of \cref{algorithm:light_link_splitting} meets the requirements of \cref{theorem:first_phase_splitting}. Our first goal is to establish \cref{theorem:first_phase_splitting}~{\it\ref{item:first_phase_cost_bound}}.

To this end, for $v\in V\setminus \{r\}$, let $x^{*,v}$ denote value of $x^*$ at the beginning of the iteration of the for-loop where $v$ is considered. Let $W_{up}$ and $W_{down}$ denote the values of the respective sets when the algorithm terminates.
\begin{proposition}\label{lemma:cost_bound_first_phase_1} We have
\begin{align*}c(x^*) \le&\quad c(x)\\&+\Delta\cdot\sum_{v\in W_{up}}	\gamma\cdot x(\{\ell\in L\colon \covr(\ell)\cap A_{up}\cap A^{vis}_v(\mathrm{supp}(x^{*,v}))\ne \emptyset\}\setminus L^\downarrow_v)\\&+\Delta\cdot \sum_{v\in W_{down}}	\gamma\cdot x(\{\ell\in L\colon \covr(\ell)\cap A_{down}\cap A^{vis}_v(\mathrm{supp}(x^{*,v}))\ne \emptyset\}\setminus L^\uparrow_v).\end{align*}
\end{proposition}
\begin{proof}
By \cref{lemma:cost_increase_splitting_new,lemma:splitting_decreases_L_up_down}, we have
\[c(x^*) \le c(x)+\Delta\cdot \left(\sum_{v\in W_{up}} x^{*,v}(L^\downarrow_v)+\sum_{v\in W_{down}} x^{*,v}(L^\uparrow_v)\right).\]
By \cref{lemma:splitting_decreases_L_up_down}, we know that $x^{*,v}(L^\downarrow_v)\le x(L^\downarrow_v)$ and $x^{*,v}(L^\uparrow_v)\le x(L^\uparrow_v)$ for all $v\in V\setminus\{r\}$. The desired statement, hence, follows from \cref{def:up_down_light}.
\end{proof}
To derive a good bound on the cost increase from \cref{lemma:cost_bound_first_phase_1}, we need to make sure that a link $\ell$ does not appear in too many of the sets $\{\ell\in L\colon \covr(\ell)\cap A_{up}\cap A^{vis}_v(\mathrm{supp}(x^{*,v}))\ne \emptyset\}\setminus L^\downarrow_v$ and $\{\ell\in L\colon \covr(\ell)\cap A_{down}\cap A^{vis}_v(\mathrm{supp}(x^{*,v}))\ne \emptyset\}\setminus L^\uparrow_v$, respectively. The next lemma takes care of this.
\begin{proposition}\label{lemma:splitting_up_charge}
Let $\ell\in L$. 
\begin{itemize}
	\item There is at most one vertex $v\in W_{up}$ such that $\covr(\ell)\cap A_{up}\cap A^{vis}_v(\mathrm{supp}(x^{*,v}))\ne \emptyset$ and $\ell\notin L^\downarrow_v$.
	\item There is at most one vertex $v\in W_{down}$ such that $\covr(\ell)\cap A_{down}\cap A^{vis}_v(\mathrm{supp}(x^{*,v}))\ne \emptyset$ and $\ell\notin L^\uparrow_v$.
\end{itemize}
\end{proposition}
\begin{proof}
We only prove the statement for $v \in W_{up}$, as the other one can be derived analogously.
If there is no such vertex $v\in W_{up}$, there is nothing to show. Next, assume that there is at least one such vertex and let $v_0$ be the first one considered by the algorithm.
Let $\ell=(y,z)$. As $\ell$ covers an arc in $A_{up}\cap A_v$, we have $z\in U_v\setminus\{v\}$. As $\ell\notin L^\downarrow_v$, $y\in U_v$. Hence, $\covr(\ell)\subseteq \cov(\ell)\subseteq A_v$. 
As $v_0\in W_{up}$, we know that every link in $L^\downarrow_v$ is split at $v_0$ and after this, we have $x^*(L^\downarrow_{v_0})=0$. By \cref{lemma:splitting_decreases_L_up_down}, for every vertex $v_1$ considered after $v_0$, we also have $x^{*,v_1}(L^\downarrow)=0$, i.e., $\mathrm{supp}(x^{*,v_1})\cap L^\downarrow_v=\emptyset$. But this tells us that no vertex $v_1$ considered after $v_0$ can see any arc in $A_{up}\cap A_v$. Indeed, if $v_1$ is considered after $v_0$, then $v_1\notin U_{v_0}$ because every vertex in $U_{v_0}\setminus \{v_0\}$ has a larger distance to $r$ than $v_0$. If there were an arc $a\in A_{up}\cap A_v$ visible to $v_1$ (w.r.t. $\mathrm{supp}(x^{*,v_1})$), then there were a link $\ell'=(y',z')\in\mathrm{supp}(x^{*,v_1})$ such that $a\in\covr(\ell')$ and $v_1\in\mathrm{in}(\overline{P}_{\ell'})$. In particular, $z'\in U_{v_0}\setminus \{v_0\}$ and $y'\notin U_{v_0}$ (as $\mathrm{in}(\overline{P}_{\ell'})\subseteq U_{v_0}$ otherwise). So $\ell'\in  \mathrm{supp}(x^{*,v_1})\cap L^\downarrow_v$, a contradiction. Hence, $A_v\cap A_{up}\cap A^{vis}_{v_1}(\mathrm{supp}(x^{*,v_1}))=\emptyset$ for every $v_1$ that is considered after $v$.
As $\covr(\ell)\cap A_{up}\subseteq A_{up}\cap A_v$, this concludes the proof.
\end{proof}
We are now ready to prove \cref{theorem:first_phase_splitting}~{\it\ref{item:first_phase_cost_bound}}.
\begin{lemma}
$c(x^*)\le (1+\epsilon)\cdot c(x)$.
\end{lemma}
\begin{proof}
	We use \cref{lemma:cost_bound_first_phase_1} and that the sets $\{\ell\in L\colon \covr(\ell)\cap A_{up}\cap A^{vis}_v(\mathrm{supp}(x^{*,v}))\ne \emptyset\}\setminus L^\downarrow_v$ for $v\in W_{up}$ and $\{\ell\in L\colon \covr(\ell)\cap A_{down}\cap A^{vis}_v(\mathrm{supp}(x^{*,v}))\ne \emptyset\}\setminus L^\uparrow_v$ for $v\in W_{down}$ are pairwise disjoint by \cref{lemma:splitting_up_charge}. Hence, \cref{lemma:cost_bound_first_phase_1} yields
	\[\sum_{\ell\in L} c(x^*)\le c(x)+2\cdot\Delta\cdot\gamma\cdot x(L)\le c(x)+2\cdot \Delta\cdot \gamma\cdot c(x)\le (1+\epsilon)\cdot c(x),\]
	where we used $c(\ell)\ge 1$ for all $\ell\in L$ for the second and \eqref{eq:constants_1} for the third inequality.
\end{proof}
Next, we establish \cref{theorem:first_phase_splitting}~{\it\ref{item:first_phase_W}}.
\begin{lemma}\label{lemma:up_down_independent_first_phase}
For $v\in W_{up}$, we have $x^*(L^\downarrow_v)=0$ and for $v\in W_{down}$, we have $x^*(L^\uparrow_v)=0$. In particular, every vertex in $W^*$ is up- or down-independent with respect to $\mathrm{supp}(x^*)$.
\end{lemma}
\begin{proof}
$r$ is both up- and down-independent since $A=A_r$. For $v\in W_{up}$, we have $x^*(L^\downarrow_v)=0$ immediately after splitting all links in $L^\downarrow_v$ at $v$. By \cref{lemma:splitting_decreases_L_up_down}, this property is preserved until the end, so $\mathrm{supp}(x^*)\cap L^\downarrow_v=\emptyset$. By \cref{lemma:up-down-independent}, $v$ is up-independent. Analogously, we can establish that every vertex in $W_{down}$ is down-independent.
\end{proof}
We are left with proving \cref{theorem:first_phase_splitting}~{\it\ref{item:first_phase_viwidth}}. Let $L'=\mathrm{supp}(\mathrm{split}(x^*,\sigma))$ be obtained from $\mathrm{supp}(x^*)$ by splitting all $W^*$-cross-links at their apex.
\begin{lemma}\label{lemma:viwup_W}
With respect to the link set $L'$, we have $\viw(r)=0$, $\viwu(v)=0$  for $v\in W_{up}$, and $\viwd(v)=0$  for $v\in W_{down}$.
\end{lemma}
\begin{proof}
Every link $\ell$ with $r\in\mathrm{in}(\overline{P}_\ell)$  is an $r$-cross-link. As $r\in W^*$, there is no link in $\ell\in L'$ with $r\in\mathrm{in}(\overline{P}_\ell)$. Hence, no arc is visible from $r$ and $\viw(r)=0$.

Next, $v\in W_{up}$ and let $a\in A_v\cap A_{up}$. We need to show that $a$ is not visible for $v$. Let $\ell=(y,z)\in L'$ be a link with $a\in\covr(\ell)$. Then $z\in U_v\setminus \{v\}$. As we have observed in the proof of \cref{lemma:up_down_independent_first_phase}, $\mathrm{supp}(x^*)\cap L^\downarrow_v=\emptyset$. As $L'=\mathrm{supp}(\mathrm{split}(x^*,\sigma))$, we also have $L'\cap L^\downarrow_v=\emptyset$ by \cref{lemma:splitting_decreases_L_up_down}. Hence, $y\in U_v$, so $\mathrm{apex}(\ell)\in U_v$. As $L'$ contains no $v$-cross-links, $v\notin \mathrm{in}(\overline{P}_\ell)$.
The statement for $v\in W_{down}$ can be derived analogously.
\end{proof}
The following lemma concludes the proof of \cref{theorem:first_phase_splitting}~{\it\ref{item:first_phase_viwidth}}.
\begin{lemma}
Let $v\in V\setminus \{r\}$.
\begin{itemize}
	\item If $\viwu(v)>k$ with respect to $L'$, then $x^*(L^\downarrow_v)> \zeta_2$, and $a_v$ is a $\zeta_2$-heavy down-arc.
	\item If $\viwd(v)>k$ with respect to $L'$, then $x^*(L^\uparrow_v)> \zeta_2$, and $a_v$ is a $\zeta_2$-heavy up-arc.
\end{itemize}
\end{lemma}
\begin{proof} We only prove the first statement, the second one can be derived analogously. Let $v\in V\setminus \{r\}$ with $\viwu(v)>k$. By \cref{lemma:viwup_W}, $v\notin W_{up}$. We begin by showing the following claim.
	\begin{claim}
		$x(L^\downarrow_v)>\zeta_2$.
	\end{claim}
	\begin{proof}[Proof of claim]
	
	Assume towards a contradiction that $x(L^\downarrow_v)\le \zeta_2$. 

As $v\notin W_{up}$, we know that
\[x(L^\downarrow_v)> \gamma\cdot x(\{\ell\in L\colon \covr(\ell)\cap A_{up}\cap A^{vis}_v(\mathrm{supp}(x^{*,v}))\ne \emptyset\}\setminus L^\downarrow_v).\]
As $\viwu(v)>k$, let $a_1,\dots,a_{k+1}\in A_v$ be ancestor-free up-arcs that are visible for $v$ with respect to $L'$. As $L'$ arises from $\mathrm{supp}(x^{*,v})$ by splitting links, $a_1,\dots, a_{k+1}$ are also visible for $v$ with respect to $\mathrm{supp}(x^{*,v})$. This implies
\[x(L^\downarrow_v)> \gamma\cdot x(\{\ell\in L\colon \covr(\ell)\cap \{a_1,\dots,a_{k+1}\}\ne \emptyset\}\setminus L^\downarrow_v),\]
which yields
\begin{align*}x(\{\ell\in L\colon \covr(\ell)\cap \{a_1,\dots,a_{k+1}\}\ne \emptyset\})&\le x(\{\ell\in L\colon \covr(\ell)\cap \{a_1,\dots,a_{k+1}\}\ne \emptyset\}\setminus L^\downarrow_v)+x(L^\downarrow_v)\\
	&< (1+\gamma^{-1})\cdot x(L^\downarrow_v)\le (1+\gamma^{-1})\cdot\zeta_2.\end{align*}
On the other hand, as the arcs $a_1,\dots,a_{k+1}$ are ancestor-free, there is no link $\ell$ covering two of them. Using that $x$ is a solution to \eqref{eq:WDTAP_LP}, this implies
\[k+1\le \sum_{i=1}^{k+1} x(\{\ell\colon a_i\in \covr(\ell)\})=x(\{\ell\in L\colon \covr(\ell)\cap \{a_1,\dots,a_{k+1}\}\ne \emptyset\})\le (1+\gamma^{-1})\cdot\zeta_2\stackrel{\eqref{eq:constants_2}}{<} k+1,\] a contradiction.
\end{proof}
As $v\notin W_{up}$, the only splitting at $v$ that we might perform in the course of \cref{algorithm:light_link_splitting} is $\sigma_{v,L^\uparrow_v}$ (in case $v\in W_{down}$). Using that $L^\downarrow_v\cap L^\uparrow_v=\emptyset$, by \cref{lemma:splitting_decreases_L_up_down}, we can infer that $x^*(L^\downarrow_v)=x(L^\downarrow_v)>\zeta_2$.
Finally, we observe that $a_v$ must be a down-arc. Indeed, if $a_v$ were an up-arc, then every link in $L^\downarrow_v$ would cover $a_v$, implying that $a_v$ would be $\zeta_2$-covered for $x$ (and $x^*$). However, this contradicts \eqref{eq:no_covered_arcs} and \eqref{eq:constants_3}.  Hence, $a_v$ is a $\zeta_2$-heavy down-arc. 
\end{proof}
This concludes the proof of \cref{theorem:first_phase_splitting}.

\section{Components and cores: handling heavy coverage in the wrong direction\label{sec:components_and_cores}}
For this section, we again fix a rooted WDTAP instance $(T,L,c,r)$ with cost ratio at most $\Delta$ and a solution $x$ to \eqref{eq:WDTAP_LP} satisfying \eqref{eq:no_covered_arcs}. Moreover, let $\sigma^*$, $x^*$ and $W^*$ be as given by \cref{theorem:first_phase_splitting}. The goal of this section is to prove \cref{theorem:second_phase_splitting}, which allows us to establish strong structural properties with respect to the $\zeta_2$-heavy arcs (for $x^*$). To state \cref{theorem:second_phase_splitting}, we require the following definitions.
\begin{definition}
An \emph{up-component} (\emph{down-component}) is a (weakly) connected component of the digraph $(V,A_{up})$ ($(V,A_{down})$). We denote the collection of up- and down-components by $\mathcal{C}_{up}$ and $\mathcal{C}_{down}$, respectively, and we let $\mathcal{C}\coloneqq \mathcal{C}_{up}\cup\mathcal{C}_{down}$. We say that $C$ is a \emph{component} if $C$ is an up- or a down-component, i.e., $C\in\mathcal{C}$.
For a component $C$, we let the \emph{root $r_C$ of $C$} be the vertex of $C$ closest to the root $r$ of $T$. 
\end{definition}
\begin{definition}
Let $C=(V',A')$ be a component. We call an arc $a'\in A'$ a \emph{base arc} if $a'$ is $\zeta_2$-heavy (w.r.t. $x^*$) and moreover, no arc of $C$ below $a'$ has this property. We denote the set of base arcs of $C$ by $B_C$.
\end{definition}
\begin{proposition}
	For every component $C$, $B_C$ is ancestor-free.\hfill\textsquare
\end{proposition}
\begin{corollary}\label{cor:base_arcs_disjoint_coverages}
Let $C\in\mathcal{C}$ and let $b,b'\in B_C$ with $b\ne b'$. Then $\{\ell\in L\colon b\in\covw(\ell)\}\cap\{\ell\in L\colon b'\in\covw(\ell)\}=\emptyset$.
\end{corollary}
\begin{proof}
Let $\ell\in L$. All arcs in $\covw(\ell)\cap A_{up}$, as well as all arcs in $\covw(\ell)\cap A_{down}$, share a pairwise ancestral relationship. Then fact that $B_C$ is an ancestor-free set are up-arcs, if $C\in\mathcal{C}_{up}$, and an ancestor-free set of down-arcs, if $C\in\mathcal{C}_{down}$, concludes the proof. 
\end{proof}
\begin{definition}
Let $C$ be a component. The \emph{core $\mathring{C}$ of $C$} consists of the union of the paths connecting the (lower vertices of) the base arcs to $r_C$, if $B_C\ne\emptyset$, and is empty otherwise.  
Let $\mathring{\mathcal{C}}_{up}\coloneqq \{\mathring{C}\colon C\in\mathcal{C}_{up},B_C\ne \emptyset\}$, $\mathring{\mathcal{C}}_{down}\coloneqq \{\mathring{C}\colon C\in\mathcal{C}_{down},B_C\ne \emptyset\}$ and
$\mathring{\mathcal{C}}\coloneqq \mathring{\mathcal{C}}_{up}\cup  \mathring{\mathcal{C}}_{down}$ be the collection of all non-empty cores.

For a core $\mathring{C}\in \mathring{\mathcal{C}}$, we call $B_{\mathring{C}}\coloneqq B_C$ the \emph{set of base arcs of $\mathring{C}$}.
\end{definition}
\begin{proposition}\label{lemma:heavy_arcs_in_core}
Let $v\in V\setminus\{r\}$ such that $a_v$ is $\zeta_2$-heavy, and let $C\in\mathcal{C}$ be the component containing $a_v$. Then $a_v$ is contained in $\mathring{C}$.
\end{proposition}
\begin{proof}
This is clear if $a_v\in B_C$. Otherwise, there exists a base arc $b\in B_C$ such that $a_v$ lies on the path from the bottom vertex of $b$ to $r_C$. Hence, $a_v$ is contained in $\mathring{C}$.
\end{proof}
We are now ready to state the main theorem of this section.
\begin{theorem}\label{theorem:second_phase_splitting}
We can, in polynomial time, compute a splitting $\sigma^{**}$ of $L$ and $x^{**}\coloneqq \mathrm{split}(x^*,\sigma^{**})=\mathrm{split}(x,\sigma^{**}\circ \sigma^*)$ such that the following properties hold:
\begin{enumerate}[label=(\thetheorem.\arabic*)]
	\item \label{property:cost_bound} $c(x^{**}) \le (1+\epsilon)^2\cdot c(x)$
	\item\label{property:not_right_and_wrong_direction} Let $C\in\mathring{\mathcal{C}}$. There is no $\ell\in\mathrm{supp}(x^{**})$ with $\covr(\ell)\cap A(C)\ne \emptyset$ and $\covw(\ell)\cap A(C)\ne \emptyset$.
	\item \label{property:no_entering_at_rC} Let $C\in\mathring{\mathcal{C}}_{up}$. Then $\mathrm{supp}(x^{**})\cap L^\uparrow_{r_C}=\emptyset$ and for every $\ell\in \mathrm{supp}(x^{**})\cap L^\downarrow_{r_C}$, $\cov(\ell)\cap A(C)=\emptyset$. \\
	Let $C\in\mathring{\mathcal{C}}_{down}$. Then $\mathrm{supp}(x^{**})\cap L^\downarrow_{r_C}=\emptyset$ and for every $\ell\in \mathrm{supp}(x^{**})\cap L^\uparrow_{r_C}$, $\cov(\ell)\cap A(C)=\emptyset$.
	\item \label{property:no_covering_two_cores} Let $C\in\mathring{\mathcal{C}}$ and let $\ell\in\mathrm{supp}(x^{**})$. If $\covr(\ell)$ contains an arc of $C$ incident to $r_C$, then $r_C$ is an endpoint of $\ell$.
\end{enumerate}
\end{theorem}
Before we move on to proving this theorem, we state an application of it.
\begin{corollary}\label{cor:apex_in_core}
Let $\sigma$ be a splitting of $L$ and let $\ell\in\mathrm{supp}(\mathrm{split}(x^{**},\sigma))$. Let $C\in\mathring{\mathcal{C}}$ such that $\cov(\ell)\cap A(C)\ne \emptyset$. Then $\mathrm{apex}(\ell)\in V(C)$.
\end{corollary}
\begin{proof}
We only consider the case where $C\in\mathring{\mathcal{C}}_{up}$, the case $C\in\mathring{\mathcal{C}}_{down}$ can be handled analogously. By \ref{property:no_entering_at_rC} and \cref{lemma:splitting_decreases_L_up_down}, we know that $\mathrm{supp}(\mathrm{split}(x^{**},\sigma))\cap L^\uparrow_{r_C}=\emptyset$, so $\ell\notin L^\uparrow_{r_C}$. We further show that $\ell\notin L^\downarrow_{r_C}$. As $\ell\in \mathrm{supp}(\mathrm{split}(x^{**},\sigma))$, there is $\ell'\in \mathrm{supp}(x^{**})$ with $\ell\in\sigma(\ell')$. Then $\ell$ is a shadow of $\ell'$, implying that also $\cov(\ell')\cap A(C)\ne \emptyset$. By \ref{property:no_entering_at_rC}, $\ell'$ is neither contained in $L^\uparrow_{r_C}$ nor in $L^\downarrow_{r_C}$, hence, both endpoints of $\ell'$ must be contained in $U_{r_C}$ (because also $\cov(\ell')\cap A_{r_C}\supseteq \cov(\ell')\cap A(C)\ne \emptyset$). As $\ell$ is a shadow of $\ell'$, both endpoints of $\ell$ are contained in $U_{r_C}$ as well. As there is $a\in A(C)\cap \cov(\ell)$, $\mathrm{apex}(\ell)$ must be an ancestor of $\mathrm{apex}(a)$, but a descendant of $r_C$. Hence, $\mathrm{apex}(\ell)\in V(C)$.
\end{proof}
The rest of this section is dedicated to the proof of \cref{theorem:second_phase_splitting}.
We first conduct a structural analysis of links in $\mathrm{supp}(x^*)$ violating properties \ref{property:not_right_and_wrong_direction}, \ref{property:no_entering_at_rC} and \ref{property:no_covering_two_cores}.
Then, we describe a link splitting procedure (\cref{alg:core_splitting}) designed to make sure that the support of the resulting solution $x^{**}$ does not contain any of the ``problematic links''.
Finally, we explain how to charge the cost of the splitting against the total costs of the links ``heavily covering base arcs in the wrong direction''.
\subsection{Structural analysis of problematic links}
For every $C\in \mathring{\mathcal{C}}$, fix an ordering $B_C=\{b^C_1,\dots,b^C_{t_C}\}$ of its base arcs. 
We introduce a decomposition of the cores into pairwise arc-disjoint paths, that will allow us to characterize the structure of the ``problematic links'' and guide our charging procedure when splitting them.
\begin{definition}[trunk decomposition]
	Let $C\in \mathring{\mathcal{C}}$. The \emph{trunk decomposition of $C$} is the decomposition of $C$ into arc-disjoint paths $P^C_1,\dots,P^C_{t_C}$ defined as follows:
	\begin{itemize}
		\item $P^C_1$ is the path connecting the bottom vertex of $b^C_1$ to $r_C$.
		\item For $i=2,\dots,t_C$, let $P'_i$ be the path connecting the bottom vertex of $b^C_i$ to $r_C$ and let $P^C_i$ be the prefix of $P'_i$ ending at the first vertex in $V(P'_i)\cap \bigcup_{j=1}^{i-1} V(P^C_j)$.
	\end{itemize}
	The paths $(P^C_i)_{i=1}^{t_C}$ are called the \emph{trunks} of the trunk decomposition.
\end{definition}
Note that we can compute $\mathring{\mathcal{C}}_{up}$, $\mathring{\mathcal{C}}_{down}$ and a trunk decomposition of every core, in polynomial time: following the definition, each step can be done in linear time and the total number of base arcs is bounded by the total number of arcs.

\begin{figure}[h]
\centering
\begin{minipage}[t]{0.4\linewidth}
    \centering
\begin{tikzpicture}[mynode/.style={draw, fill, circle, minimum size = 2mm, inner sep = 0pt}, myarc/.style={thick, arrows = {-Stealth[scale=1.5]}},mylink/.style={thick, arrows = {-Stealth[scale=1.5]}, black!70!white, dashed}, covered/.style={line width = 2mm, draw opacity = 0.4}, myboldarc/.style={line width = 1.2mm, draw opacity = 0.75, arrows = {-Stealth[scale=0.5]}}]

\node           at ( 0.5, 4) {$r_C$};
\node[mynode] (v1) at (0,4){};
\node[mynode] (v2) at (0,3){};
\node[mynode] (v3) at (1,3){};
\node[mynode] (v4) at (-1,2){};
\node[mynode] (v5) at (0,2){};
\node[mynode] (v6) at (1,2){};
\node[mynode] (v7) at (2,2){};
\node[mynode] (v8) at (-1,1){};
\node[mynode] (v9) at (1,1){};
\node[mynode] (v10) at (-2,0){};
\node[mynode] (v11) at (2,0){};

\draw[myboldarc] (v2) to (v1); 
\draw[myarc] (v3) to (v1); 
\draw[myarc] (v4) to (v2);
\draw[myarc] (v5) to (v2);
\draw[myarc] (v6) to (v3);
\draw[myboldarc] (v7) to node[midway, right=5pt] {$\color{red!70!black} b^C_3$} (v3); 
\draw[myboldarc] (v8) to node[midway, left=5pt] {$\color{blue} b^C_1$} (v5); 
\draw[myarc] (v9) to (v5);
\draw[myarc] (v10) to (v8);
\draw[myboldarc] (v11) to node[midway, right=5pt] {$\color{green!70!black} b^C_2$} (v9); 
\end{tikzpicture}
\label{fig:trunks}

\end{minipage}
\begin{minipage}[t]{0.4\textwidth}
    
\centering
\begin{tikzpicture}[mynode/.style={draw, fill, circle, minimum size = 2mm, inner sep = 0pt}, myarc/.style={thick, arrows = {-Stealth[scale=1.5]}},mylink/.style={thick, arrows = {-Stealth[scale=1.5]}, black!70!white, dashed}, covered/.style={line width = 2mm, draw opacity = 0.25}]

\node[anchor=west]           at ( 0.3, 4) {$r_C=v_1^C=v_3^C$};
\node[mynode] (v1) at (0,4){};
\node[mynode] (v2) at (0,3){};
\node[mynode] (v3) at (1,3){};
\node[mynode] (v5) at (0,2){};
\node at (0.6,2){$v_2^C$};
\node[mynode] (v7) at (2,2){};
\node[mynode] (v8) at (-1,1){};
\node[mynode] (v9) at (1,1){};
\node[mynode] (v11) at (2,0){};

\draw[myarc] (v2) to (v1); 
\draw[myarc] (v3) to (v1); 
\draw[myarc] (v5) to (v2);
\draw[myarc] (v7) to node[midway, right=5pt] {$\color{red!70!black} b^C_3$} (v3); 
\draw[covered,red] (v7) to (v3);
\draw[covered,red] (v3) to (v1);
\draw[myarc] (v8) to node[midway, left=5pt] {$\color{blue} b^C_1$} (v5); 
\draw[covered,blue] (v8) to (v5);
\draw[covered,blue] (v5) to (v2);
\draw[covered,blue] (v2) to (v1);
\draw[myarc] (v9) to (v5);
\draw[myarc] (v11) to node[midway, right=5pt] {$\color{green!70!black} b^C_2$} (v9); 
\draw[covered,green!70!black] (v11) to (v9);
\draw[covered,green!70!black] (v9) to (v5);

\end{tikzpicture}
\label{fig:trunks}
\end{minipage}

\caption{An up-component $C$ with root $r_C$ is shown on the left. The $\zeta_2$-heavy arcs are drawn in bold with $b_1^C$, $b_2^C$ and $b_3^C$ being the base arcs. The core of $C$ is shown on the right, together with its trunk decomposition, indicated by colors. Here, $P_1^C$ (blue) is the parent trunk of $P_2^C$ (green). The sibling arc of $P_2^C$ is $b_1^C$.}

\end{figure}
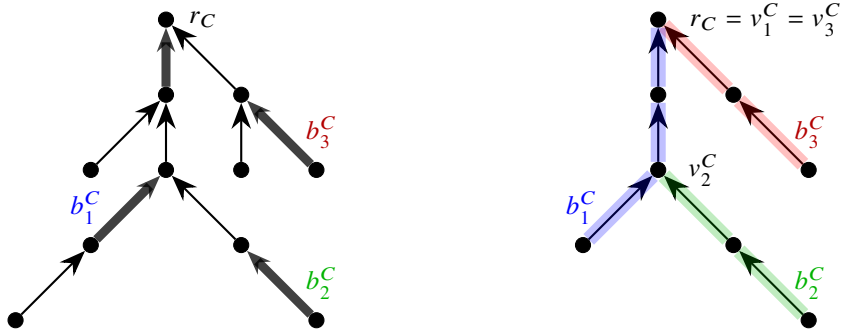

\begin{definition}
For a trunk $P_i^C$, we denote its top endpoint by $v_i^C$ and its top arc by $a_i^C$.
\end{definition}
\begin{definition}
Let $C\in \mathring{\mathcal{C}}$ and let $(P^C_i)_{i=1}^{t_C}$ be the trunks of the trunk decomposition of $C$. For $i\in\{1,\dots,t_C\}$ with $v^C_i\neq r_C$, we define the \emph{parent trunk of $P^C_i$} to be the trunk $P^C_j$ containing $a_{v^C_i}$. Note that as $B_C$ is ancestor-free, $v^C_i$ cannot be the bottom endpoint of $P^C_j$. We further call the arc of $P^C_j$ connecting to $v^C_i$ from below the \emph{sibling arc of $P^C_i$} and denote it by $s^C_i$.
\end{definition}
We are now ready to analyze the structure of links that violate the condition in \ref{property:not_right_and_wrong_direction}.
\begin{proposition}\label{lemma:cover_core_in_both_directions}
Let $C\in\mathring{\mathcal{C}}$ and let $\ell\in L$ with $\covr(\ell)\cap A(C)\ne\emptyset$ and $\covw(\ell)\cap A(C)\ne\emptyset$. Then:
\begin{enumerate}[label=(\thetheorem.\arabic*)]
	\item \label{property:neither_up_nor_down} $\ell$ is neither an up- nor a down-link, i.e., $\mathrm{apex}(\ell)\in\mathrm{in}(P_\ell)$.
	\item \label{property:cover_top_arc_or_sibling_arc} There is a trunk $P_i^C$ such that $\mathrm{apex}(\ell)=v_i^C$ and either $a_i^C\in\covr(\ell)$, or $v_i^C\ne r_C$ and $s_i^C\in\covr(\ell)$.
\end{enumerate}
\end{proposition}
\begin{proof}
The fact that $A(C)$ either only contains up-arcs or only down-arcs implies that for an up- or down-link $\ell$, $\covr(\ell)\cap A(C)=\emptyset$ or $\covw(\ell)\cap A(C)=\emptyset$. This establishes \ref{property:neither_up_nor_down}.
To simplify notation, we prove \ref{property:cover_top_arc_or_sibling_arc} only for the case where $C\in\mathring{\mathcal{C}}_{up}$; the case $C\in\mathring{\mathcal{C}}_{down}$ can be handled analogously. 
Let $\ell=(u,v)$, let $a\in \covr(\ell)\cap A(C)$ and let $a'\in \covw(\ell)\cap A(C)$. Then $a$ lies on the $\mathrm{apex}(\ell)$-$v$-path in $T$ and $a'$ lies on the $u$-$\mathrm{apex}(\ell)$-path in $T$. In particular, $\mathrm{apex}(\ell)$ is the lowest common ancestor of (the top vertices of) $a$ and $a'$, and, as $C$ is connected, $\mathrm{apex}(\ell)\in V(C)$. Let $a_0$ and $a'_0$ be the top arcs of the $v$-$\mathrm{apex}(\ell)$-path and the $u$-$\mathrm{apex}(\ell)$-path, respectively. Then $a_0\in \covr(\ell)$ and $a'_0\in\covw(\ell)$. Let $P_j^C$ be the trunk containing $a_0$. If $v_j^C=\mathrm{apex}(\ell)$, we let $P_i^C\coloneqq P^C_j$. Otherwise, $\mathrm{apex}(\ell)\ne r_C$ and $a_{\mathrm{apex}(\ell)}\in A(P_j^C)$. In this case, we let $P_i^C$ be the trunk containing $a'_0$. Then $v_i^C=\mathrm{apex}(\ell)$ and $a_0=s^C_i$ is the sibling arc of $P_i^C$.
\end{proof}
The next two propositions help us to understand the structure of links violating the conditions in \ref{property:no_entering_at_rC}.
\begin{proposition}\label{lemma:no_entering_rC_1}\
	\begin{itemize}
		\item Let $C\in\mathring{\mathcal{C}}_{up}$ and let $\ell\in L_{r_C}^\uparrow$. Then $r_C\ne r$, $r_C\in\mathrm{in}(P_\ell)$ and $a_{r_C}\in\covr(\ell)$.
		\item Let $C\in\mathring{\mathcal{C}}_{down}$ and let $\ell\in L_{r_C}^\downarrow$. Then $r_C\ne r$, $r_C\in\mathrm{in}(P_\ell)$ and $a_{r_C}\in\covr(\ell)$.
	\end{itemize}
\end{proposition}
\begin{proof}
We only prove the first statement; the second one follows analogously. Let $C\in\mathring{\mathcal{C}}_{up}$ and let $\ell=(u,v)\in L_{r_C}^\uparrow$. Then $u\in U_{r_C}\setminus \{r_C\}$ and $v\in V\setminus U_{r_C}$, so $r_C\in\mathrm{in}(P_\ell)$. Moreover, $\mathrm{apex}(\ell)$ is a strict ancestor of $r_C$ (implying $r_C\ne r$) and $a_{r_C}$ appears on the $u$-$\mathrm{apex}(\ell)$-path in $T$. As $C\in\mathring{\mathcal{C}}_{up}$, $a_{r_C}$ is a down-arc, so $a_{r_C}\in\covr(\ell)$.
\end{proof}
\begin{proposition}\label{lemma:no_entering_rC_2}
\
\begin{itemize}
	\item Let $C\in\mathring{\mathcal{C}}_{up}$ and let $\ell\in L_{r_C}^\downarrow$ with $\cov(\ell)\cap A(C)\ne \emptyset$. Then $r_C\in\mathrm{in}(P_\ell)$ and there is a trunk $P^C_i$ with $v^C_i=r_C$ and $a^C_i\in\covr(\ell)$.
	\item Let $C\in\mathring{\mathcal{C}}_{down}$ and let $\ell\in L_{r_C}^\uparrow$ with $\cov(\ell)\cap A(C)\ne \emptyset$. Then $r_C\in\mathrm{in}(P_\ell)$ and there is a trunk $P^C_i$ with $v^C_i=r_C$ and $a^C_i\in\covr(\ell)$.
\end{itemize}
\end{proposition}
\begin{proof}
Again, we only prove the first statement. Let $C\in\mathring{\mathcal{C}}_{up}$ and let $\ell=(u,v)\in L_{r_C}^\downarrow$ with $\cov(\ell)\cap A(C)\ne \emptyset$. Then $v\in U_{r_C}\setminus \{r_C\}$ and $u\notin U_{r_C}$. In particular, $r_C\in\mathrm{in}(P_\ell)$ and moreover, $\mathrm{apex}(\ell)$ is a strict ancestor of $r_C$. Let $a\in \cov(\ell)\cap A(C)$. Then $a$ lies on the $v$-$r_C$-path in $T$. Let $a'$ be the last arc of this path. As $C$ is connected, $a'\in A(C)$, so there is a trunk $P_i^C$ containing $a'$. As $a'$ is incident to $r_C$, $a'=a_i^C$ and $v^C_i=r_C$. Finally, as $C\in\mathring{\mathcal{C}}_{up}$, $a_i^C\in A_{up}$. As $a_i^C$ lies on the $v$-$\mathrm{apex}(\ell)$-path in $T$, $a_i^C\in\covr(\ell)$.
\end{proof}
\Cref{lemma:cover_core_in_both_directions,lemma:no_entering_rC_1,lemma:no_entering_rC_2} and \ref{property:no_covering_two_cores} motivate the splitting procedure presented in the following section.
\subsection{The splitting algorithm}
We obtain $\sigma^{**}$ and $x^{**}=\mathrm{split}(x^*,\sigma^{**})$ via \cref{alg:core_splitting}.
\begin{algorithm}
	\begin{algorithmic}[1]
		\Require{solution $x^*$ to \eqref{eq:WDTAP_LP}}
		\Ensure{splitting $\sigma^{**}$ of $L$, solution $x^{**}=\mathrm{split}(x^*,\sigma^{**})$ to \eqref{eq:WDTAP_LP}}
		
		\State $\sigma^{**}\gets \sigma_{\rm{id}}$, $x^{**}\gets x^{*}$
		\For{$C\in\mathring{\mathcal{C}}$}
		\State $L'\gets \{\ell\in L\colon a_{r_C}\in\covr(\ell)\}$
		 \State $\sigma^{**}\gets \sigma_{r_C,L'}\circ \sigma^{**}$, $x^{**}\gets \mathrm{split}(x^{**},\sigma_{r_C,L'})$ \label{line:split_coverage_arC}
		\For{$i\gets 1$ \textbf{to} $t_C$}
		\State $L'\gets \{\ell\in L\colon a_i^C\in\covr(\ell)\}$
		\State  $\sigma^{**}\gets \sigma_{v^C_i,L'}\circ \sigma^{**}$, $x^{**}\gets \mathrm{split}(x^{**},\sigma_{v^C_i,L'})$\label{line:split_top_arc}
		\If{$v^C_i\ne r_C$}
		\State $L'\gets \{\ell\in L\colon s^C_i\in\covr(\ell)\}$
		\State $\sigma^{**}\gets \sigma_{v^C_i,L'}\circ \sigma^{**}$, $x^{**}\gets \mathrm{split}(x^{**},\sigma_{v^C_i,L'})$\label{line:split_sibling_arc}
		\EndIf
		\EndFor
		\EndFor
		\State\Return {$\sigma^{**}$, $x^{**}$}
	\end{algorithmic}
	\caption{Core link splitting.}\label{alg:core_splitting}	
\end{algorithm}

Note that \cref{alg:core_splitting} runs in polynomial time.
The following technical claim is useful to further analyze \cref{alg:core_splitting}.
\begin{proposition}\label{lemma:later_splits_dont_destroy_earlier_splits}
Let $a\in A$ and let $u$ be an endpoint of $a$. Let $x$ be a solution to \eqref{eq:WDTAP_LP} such that for every $\ell\in\mathrm{supp}(x)$ with $a\in\covr(\ell)$, $u$ is an endpoint of $\ell$.

 Let $v\in V$, $L'\subseteq L$ and let $x'\coloneqq\mathrm{split}(x,\sigma_{v,L'})$. Then for every $\ell\in\mathrm{supp}(x')$ with $a\in\covr(\ell)$, $u$ is an endpoint of $\ell$.
\end{proposition}
\begin{proof}
Let $\ell\in\mathrm{supp}(x')$ with $a\in\covr(\ell)$. If $\ell\in \mathrm{supp}(x)$, the statement follows from our assumption on $x$. Otherwise, there is $\ell'\in L'\cap\mathrm{supp}(x)$ such that $\ell$ is a shadow of $\ell'$ with $a\in\covr(\ell)$. Then also $a\in\covr(\ell')$, so $u$ is an endpoint of $\ell'$. As $u$ is an endpoint of $\ell'$ and $a$, $\ell$ is a shadow of $\ell'$ and $a\in\cov(\ell)$, $u$ must be an endpoint of $\ell$ as well.
\end{proof}
\begin{proposition}\label{lemma:properties_x_double_star_1}
	The solution $x^{**}$ computed by \cref{alg:core_splitting} satisfies \ref{property:not_right_and_wrong_direction}.
\end{proposition}
\begin{proof}
Let $C\in\mathring{\mathcal{C}}$ and let $P_i^C$ be a trunk of $C$. We know that for every link $\ell\in\mathrm{supp}(x^{**})$ with $a_i^C\in\covr(\ell)$, $v^C_i$ is an endpoint of $\ell$ because this property holds immediately after line~\ref{line:split_top_arc} of \cref{alg:core_splitting} is executed (for $C$ and $i$) and it is preserved by later splits by \cref{lemma:later_splits_dont_destroy_earlier_splits}.

If $r^C_i\ne r_C$, we further know that for every link $\ell\in\mathrm{supp}(x^{**})$ with $s_i^C\in\covr(\ell)$, $v^C_i$ is an endpoint of $\ell$ because this property holds immediately after line~\ref{line:split_sibling_arc} of \cref{alg:core_splitting} is executed (for $C$ and $i$) and it is preserved by later splits by \cref{lemma:later_splits_dont_destroy_earlier_splits}.

 By \cref{lemma:cover_core_in_both_directions}, \ref{property:not_right_and_wrong_direction} is satisfied.
\end{proof}
\begin{proposition}\label{lemma:properties_x_double_star_2}
	The solution $x^{**}$ computed by \cref{alg:core_splitting} satisfies \ref{property:no_entering_at_rC}.
\end{proposition}
\begin{proof}
 We only prove the statement for $C\in\mathring{\mathcal{C}}_{up}$; the case  $C\in\mathring{\mathcal{C}}_{down}$ can be handled analogously.
 
 Let $C\in\mathring{\mathcal{C}}_{up}$ with $r_C\ne r$ (otherwise, $L^\uparrow_{r_C}=\emptyset$ and $L^\downarrow_{r_C}=\emptyset$ and there is nothing to show). We know that every link $\ell\in\mathrm{supp}(x^{**})$ with $a_{r_C}\in \covr(\ell)$ has $r_C$ as an endpoint because this property holds immediately after line~\ref{line:split_coverage_arC} of \cref{alg:core_splitting} is executed for $C$ and it is preserved by later splits by \cref{lemma:later_splits_dont_destroy_earlier_splits}. By \cref{lemma:no_entering_rC_1}, $\mathrm{supp}(x^{**})\cap L^\uparrow_{r_C}=\emptyset$.

Let $P^C_i$ be  a trunk of $C$ with $v^C_i=r_C$. We know that for every $\ell\in\mathrm{supp}(x^{**})$ with $a^C_i\in\covr(\ell)$, $v^C_i=r_C$ is an endpoint of $\ell$ because this property holds immediately after line~\ref{line:split_top_arc} of \cref{alg:core_splitting} is executed (for $C$ and $i$) and it is preserved by later splits by \cref{lemma:later_splits_dont_destroy_earlier_splits}. By \cref{lemma:no_entering_rC_2}, for every $\ell\in\mathrm{supp}(x^{**})\cap L^\downarrow_{r_C}$, $\cov(\ell)\cap A(C)=\emptyset$.
\end{proof}
\begin{proposition}\label{lemma:no_covering_two_cores}
The solution $x^{**}$ computed by \cref{alg:core_splitting} satisfies \ref{property:no_covering_two_cores}.	
\end{proposition}
\begin{proof}
Let $C\in\mathring{\mathcal{C}}$ and let $a\in A(C)$ be an arc incident to $r_C$. Then $a$ is the top arc of some trunk of $C$ ending in $r_C$, i.e., there is $i\in\{1,\dots,t_C\}$ such that $a=a_i^C$ and $v_i^C=r_C$. Now, every link $\ell\in\mathrm{supp}(x^{**})$ with $a\in \covr(\ell)$ has $r_C$ as an endpoint because this property holds immediately after line \ref{line:split_top_arc} of \cref{alg:core_splitting} is executed for $i$ and $C$, and it is preserved by later splits by \cref{lemma:later_splits_dont_destroy_earlier_splits}.
\end{proof}
To conclude the proof of \cref{theorem:second_phase_splitting}, it remains to establish \ref{property:cost_bound}, which is the goal of the following section.
\subsection{Bounding the cost of the splitting}
In order to bound the costs of the splitting operations performed in \cref{alg:core_splitting}, we first establish a lower bound on the total cost of $x^*$. \cref{lemma:heavy_x_value} allows us to relate the total $x^*$-value on links who have their apex in a core $C$ to the number $t_C$ of base arcs of $C$.
\cref{cor:lower_bound_costs_x_star} then gives a lower bound on the cost of $x^*$ in terms of the total number of base arcs in all cores.
For the proof of \cref{lemma:heavy_x_value}, we observe that since there are no $\zeta_1$-covered arcs with respect to $x$, \cref{lemma:splitting_same_coverage} implies that
\begin{equation}
\text{ there are no $\zeta_1$-covered arcs with respect to $x^*$.} \label{eq:no_covered_arcs_x_star}
\end{equation}
\begin{lemma}\label{lemma:heavy_x_value}
	Let $C\in\mathring{\mathcal{C}}$. Then $x^*(\{\ell\in L\colon\mathrm{apex}(\ell)\in V(C)\})\ge (1-\epsilon)\cdot\zeta_2\cdot t_C$.
\end{lemma}
\begin{proof}
	For $i\in\{1,\dots,t_C\}$, let $L_i\coloneqq \{\ell\in L\colon b^C_i\in\covw(\ell)\}$. By \cref{cor:base_arcs_disjoint_coverages}, the sets $(L_i)_{i=1}^{t_C}$ are pairwise disjoint. Moreover, as every base arc is $\zeta_2$-heavy, $x^*(L_i)\ge \zeta_2$ for every $i\in\{1,\dots,t_C\}$.
	\begin{claim}
		Let $i\in\{1,\dots,t_C\}$ and let $\ell\in L_i$. If $\mathrm{apex}(\ell)\notin V(C)$, then $r_C\ne r$ and $a_{r_C}\in\covr(\ell)$.
	\end{claim}
	\begin{proof}[Proof of claim]
		To simplify notation, we assume that $C\in\mathring{\mathcal{C}}_{up}$; the case $C\in\mathring{\mathcal{C}}_{down}$ can be handled analogously. Let $\ell=(u,v)$. As $b^C_i\in \covw(\ell)$, the bottom vertex $w$ of $b^C_i$ is an ancestor of $u$ and $\mathrm{apex}(\ell)$ is strict ancestor of $w$. If $\mathrm{apex}(\ell)$ appears on the $w$-$r_C$-path in $T$, then $\mathrm{apex}(\ell)\in V(C)$. Otherwise, $\mathrm{apex}(\ell)$ is a strict ancestor of $r_C$ and in particular, $r_C\ne r$. Moreover, $a_{r_C}$ is a down-arc that appears on the $u$-$\mathrm{apex}(\ell)$-path in $T$, implying $a_{r_C}\in\covr(\ell)$.
	\end{proof}
	By the claim, if $r=r_C$, then \[x^*(\{\ell\in L\colon\mathrm{apex}(\ell)\in V(C)\})\ge \sum_{i=1}^{t_C} x^*(L_i)\ge t_C\cdot\zeta_2.\]
	Next, assume that $r\ne r_C$. The claim yields
	\[x^*(\{\ell\in L\colon\mathrm{apex}(\ell)\in V(C)\})\ge \sum_{i=1}^{t_C} x^*(L_i)-x^*(\{\ell\in L\colon a_{r_C}\in\covr(\ell)\})\stackrel{\eqref{eq:no_covered_arcs_x_star}}{>} t_C\cdot \zeta_2-\zeta_1\stackrel{\eqref{eq:constants_3}}{>} (1-\epsilon)\cdot \zeta_2\cdot t_C.\]
\end{proof}
\begin{corollary}\label{cor:lower_bound_costs_x_star} We have
$c(x^*)\ge \frac{1}{2}\cdot(1-\epsilon)\cdot\zeta_2\cdot \sum_{C\in\mathring{\mathcal{C}}} t_C$.
\end{corollary}
\begin{proof}
This follows from \cref{lemma:heavy_x_value}, using that the link costs lie in $[1,\Delta]$ and that for every $\ell\in L$, there is at most one $C\in\mathring{\mathcal{C}}_{up}$ and at most one $C\in\mathring{\mathcal{C}}_{down}$ with $\mathrm{apex}(\ell)\in V(C)$ because the up-cores, as well as the down-cores, are pairwise vertex-disjoint.
\end{proof}
\begin{lemma}\label{lemma:cost_bound_x_double_star}
We have $c(x^{**}) \le (1+\epsilon)^2\cdot c(x)$, i.e., \ref{property:cost_bound} holds.
\end{lemma}
\begin{proof}
Using \Cref{lemma:splitting_same_coverage}, and \Cref{lemma:cost_increase_splitting_new}, we obtain

\begin{align*}
c(x^{**})&\le c(x^*)+\Delta\cdot \sum_{C\in\mathring{\mathcal{C}}}x^*(\{\ell\in L\colon a_{r_C}\in\covr(\ell)\})\\	
&\phantom{=}+\Delta\cdot\sum_{C\in\mathring{\mathcal{C}}}
\bigg[\sum_{i=1}^{t_C} x^*(\{\ell\in L\colon a_i^C\in\covr(\ell)\})+\sum_{i=1,v^C_i\ne r_C}^{t_C}x^*(\{\ell\in L\colon s_i^C\in\covr(\ell)\})\bigg]\\
&\stackrel{\eqref{eq:no_covered_arcs_x_star}}{\le} c(x^*)+\zeta_1\cdot\Delta\cdot |\mathring{\mathcal{C}}|+2\cdot \zeta_1\cdot \Delta\cdot \sum_{C\in\mathring{\mathcal{C}}} t_C\\
&\stackrel{(*)}{\le} (1+\epsilon)\cdot c(x)+3\cdot \zeta_1\cdot \Delta\cdot \sum_{C\in\mathring{\mathcal{C}}} t_C\stackrel{\eqref{eq:constants_4}}{\le}  (1+\epsilon)\cdot c(x)+\epsilon\cdot\frac{1}{2}\cdot(1-\epsilon)\cdot\zeta_2\cdot \sum_{C\in\mathring{\mathcal{C}}} t_C\\
&\le (1+\epsilon)\cdot c(x)+\epsilon\cdot c(x^*)\le (1+\epsilon)^2\cdot c(x),
\end{align*}
where the inequality marked $(*)$ follows from \cref{theorem:first_phase_splitting}~{\it\ref{item:first_phase_cost_bound}} and the fact that $t_C\ge 1$ for every core $C$, the second-to-last inequality follows from \cref{cor:lower_bound_costs_x_star}, and the last inequality follows again from \cref{theorem:first_phase_splitting}~{\it\ref{item:first_phase_cost_bound}}.
\end{proof}
Combining \Cref{lemma:properties_x_double_star_1,lemma:properties_x_double_star_2,lemma:no_covering_two_cores}, and \Cref{lemma:cost_bound_x_double_star} proves \Cref{theorem:second_phase_splitting}.
\section{Best of three solutions\label{sec:best_of_three}}
For this section, we again fix a rooted WDTAP instance $(T,L,c,r)$ with cost ratio at most $\Delta$ and a solution $x$ to \eqref{eq:WDTAP_LP} satisfying \eqref{eq:no_covered_arcs}. Moreover, let $\sigma^*$, $x^*$ and $W^*$ be as given by \cref{theorem:first_phase_splitting} and let $x^{**}$ and $\sigma^{**}$ be as given by \cref{theorem:second_phase_splitting}.
The goal of this section is to prove \cref{theorem:partial_separation_oracle} by constructing three different solutions arising from $x^{**}$ by splitting certain links. To this end, we classify and partition the links in the support of $x^{**}$.
We begin by proving the following technical claim, that will be helpful to define disjoint subsets of $\mathrm{supp}(x^{**})$.
\begin{proposition}\label{lemma:cover_core_apex}
Let $\ell\in \mathrm{supp}(x^{**})$ and let $C\in\mathring{\mathcal{C}}$ such that $\cov(\ell)\cap A(C)\ne \emptyset$. Then $\mathrm{apex}(\ell)\in V(C)$ and $\cov(\ell)\cap A(C)$ contains an arc incident to $\mathrm{apex}(\ell)$.
\end{proposition}
\begin{proof}
Let $a\in \cov(\ell)\cap A(C)$ and let $v$ be the bottom vertex of $a$. Then both $\mathrm{apex}(\ell)$ and $r_C$ appear on $P_{vr}$, the $v$-$r$-path in $T$. If $\mathrm{apex}(\ell)$ lies on the $v$-$r_C$-subpath of $P_{vr}$, then $\mathrm{apex}(\ell)\in V(C)$ and moreover, the last arc of the $v$-$\mathrm{apex}(\ell)$-subpath of $P_{vr}$ is contained in $\cov(\ell)\cap A(C)$. Otherwise, $\ell\in L^\uparrow_{r_C}\cup L^\downarrow_{r_C}$ and $\cov(\ell)\cap A(C)\ne \emptyset$, contradicting \ref{property:no_entering_at_rC}. 
\end{proof}
\begin{definition}
We define $\overrightarrow{L}\coloneqq \{\ell\in L\colon \exists C\in\mathring{\mathcal{C}}\colon \covr(\ell)\cap A(C)\ne\emptyset\}$ to be the set of links that ``cover part of a core in the right direction'' and $\overleftarrow{L}\coloneqq \{\ell\in L\colon \exists C\in\mathring{\mathcal{C}}\colon \covw(\ell)\cap A(C)\ne\emptyset\}$ to be the set of links that ``cover part of a core in the wrong direction''.
\end{definition}
It turns out that the support of $x^{**}$ does not contain any link in $\overrightarrow{L}\cap \overleftarrow{L}$. In fact, $\overrightarrow{L}\cap \overleftarrow{L}$ does not even contain a shadow of a link in $\mathrm{supp}(x^{**})$.
\begin{lemma}\label{lemma:no_cover_right_and_wrong}
$\overrightarrow{L}\cap \overleftarrow{L}\cap \{\ell\colon \text{ $\ell$ is a shadow of a link in $\mathrm{supp}(x^{**})$}\}=\emptyset$.
\end{lemma}
\begin{proof}
Assume towards a contradiction that $\ell'\in \overrightarrow{L}\cap \overleftarrow{L}$ is a shadow of $\ell\in\mathrm{supp}(x^{**})$. Then also $\ell\in\overrightarrow{L}\cap \overleftarrow{L}$ and there are $C\in\mathring{\mathcal{C}}$ with $\covr(\ell)\cap A(C)\ne \emptyset$ and $C'\in\mathring{\mathcal{C}}$ with $\covw(\ell)\cap A(C)\ne \emptyset$. By \ref{property:not_right_and_wrong_direction}, $\covw(\ell)\cap A(C)=\emptyset$ and $\covr(\ell)\cap A(C')=\emptyset$. Hence, $C\ne C'$ and, in particular, $A(C)\cap A(C')=\emptyset$. By \cref{lemma:cover_core_apex}, $\mathrm{apex}(\ell)\in V(C)\cap V(C')$ and one of the two incident arcs of $\mathrm{apex}(\ell)$ in $\cov(\ell)$ is contained in $A(C)$, call it $a$, and the other one is contained in $A(C')$, call it $a'$. Let $\ell=(u,v)$. Then either $a\in\covr((u,\mathrm{apex}(\ell)))$ and $a'\in\covw((\mathrm{apex}(\ell),v))$, implying that both $a$ and $a'$ are down-arcs, or $a'\in\covw((u,\mathrm{apex}(\ell)))$ and $a\in\covr((\mathrm{apex}(\ell),v))$, implying that both $a$ and $a'$ are up-arcs. In either case, $C$ and $C'$ are distinct cores for the same direction sharing a vertex (namely, $\mathrm{apex}(\ell)$), a contradiction.
\end{proof}
To construct the different solutions, we will split certain collections of links at their apex. To describe this operation formally, we introduce the following notation.
\begin{definition}
Let $L'\subseteq L$. The splitting $\sigma_{L'}$ of $L$, which splits every link in $L'$ (that is not already an up- or down-link) at its apex, is defined as follows.
If $\ell\in L\setminus L'$ or $\ell\in L'$ is an up-link or a down-link, we define $\sigma_{L'}(\ell)=\{\ell\}$.
If $\ell=(u,v)\in L'$ is neither an up- nor a down-link, we define $\sigma_{L'}(\ell)=\{(u,\mathrm{apex}(\ell)),(\mathrm{apex}(\ell),v)\}$.
\end{definition}

 Let $X\coloneqq \bigcup_{C\in\mathring{\mathcal{C}}} V(C)$ be the set of vertices of all cores. Let $L_{cross}$ consist of all $(W^*\cup X)$-cross-links that are neither contained in $\overrightarrow{L}$ nor in $\overleftarrow{L}$, i.e.,
 \[L_{cross}\coloneqq \{\ell\in L\setminus (\overrightarrow{L}\cup \overleftarrow{L})\colon \text{ $\ell$ is a $(W^*\cup X)$-cross-link}\}.\]  We further define $L_{rest}\coloneqq L\setminus (L_{cross}\cup \overrightarrow{L}\cup \overleftarrow{L})$.
 By \cref{lemma:no_cover_right_and_wrong}, we know that 
 \begin{equation}\mathrm{supp}(x^{**})=(L_{cross}\cap \mathrm{supp}(x^{**}))\dot{\cup}(\overrightarrow{L}\cap \mathrm{supp}(x^{**}))\dot{\cup}(\overleftarrow{L}\cap \mathrm{supp}(x^{**}))\dot{\cup}(L_{rest}\cap \mathrm{supp}(x^{**}))\label{eq:partition_supp}\end{equation} is a partition of $\mathrm{supp}(x^{**})$. We are now ready to construct the three solutions of interest.
 \begin{lemma}\label{lemma:x_1}
 Let $\sigma_1\coloneqq \sigma_{L\setminus L_{rest}}\circ\sigma^{**}\circ\sigma^{*}$, let $x_1\coloneqq \mathrm{split}(x^{**},\sigma_{L\setminus L_{rest}})=\mathrm{split}(x,\sigma_1)$ and 
let $L_1\coloneqq\mathrm{supp}(x_1)$. Then the visible width of $(T,L_1)$ is at most $k$. In particular, we can in polynomial time, compute a solution to $(T,L_1)$ of cost at most $c(x_1)$, or find violated visibly $k$-wide modification inequality for $(T,L,c,r)$.
 \end{lemma}
\begin{proof}
Recall that $L_1=\mathrm{supp}(x_1)$ and $x_1=\mathrm{split}(x^*,\sigma_{L\setminus L_{rest}}\circ \sigma^{**})$. By \cref{lemma:splitting_cannot_increase_width,theorem:first_phase_splitting}, we know that $r$, as well as every $v\in V\setminus \{r\}$ for which $a_v$ is not $\zeta_2$-heavy (with respect to $x^*$), has visible width at most $k$ with respect to $L_1$. Moreover, $\viwr(v)\le k$ for every $v\in V$. Hence, it suffices to show that for every $v\in V\setminus\{r\}$ for which $a_v$ is $\zeta_2$-heavy, $\viww(v)\le k$ (with respect to $L_1$). In fact, we will show that $\viww(v)=0$ by showing the following claim:
\begin{claim} No arc in $\overleftarrow{A}_v$ is visible for $v$ with respect to $L_1$. 
\end{claim}
\begin{proof}[Proof of claim]
Recall that $\overleftarrow{A}_v$ is the set of arcs in $A_v$ that have the opposite orientation of $a_v$, i.e., that are down-arcs, if $a_v$ is an up-arc, and up-arcs, if $a_v$ is a down-arc.
Towards a contradiction, let $a\in \overleftarrow{A}_v$ and assume that there is $\ell\in L_1$ with $a\in\covr(\ell)$ and $v\in\mathrm{in}(\overline{P}_\ell)$. By \cref{lemma:heavy_arcs_in_core}, there exists a core $C$ containing $a_v$. In particular, $v\in X$. We note that $L_1$ does not contain any $(W^*\cup X)$-cross-links because every $(W^*\cup X)$-cross-link is contained in $L\setminus L_{rest}$ and has, hence, been split at its apex when constructing $L_1$. In particular, $\ell$ cannot be a $v$-cross-link. As $v\in\mathrm{in}(\overline{P}_\ell)$, $\mathrm{apex}(\ell)$ has to be a strict ancestor of $v$ and exactly one endpoint of $\ell$, say $u$, is contained in $U_v$. Both $a$ and $a_v$ lie on the $u$-$\mathrm{apex}(\ell)$-path in $T$, however, as $a\in\overleftarrow{A}_v$, exactly one of them is an up-arc and exactly one of them is a down-arc. As $a\in\covr(\ell)$, this implies $a_v\in\covw(\ell)$. By \cref{cor:apex_in_core}, $\mathrm{apex}(\ell)\in V(C)\subseteq X$. As $L_1$ does not contain any $(W^*\cup X)$-cross-link, $\mathrm{apex}(\ell)$ is an endpoint of $\ell$.

 As all arcs on the $v$-$\mathrm{apex}(\ell)$-path are contained in $A(C)$ and oriented in the same way as $a_v$, none of them is contained in $\covr(\ell)$. Hence, $\overline{P}_\ell$ is a subpath of the $u$-$v$-subpath in $T$. Hence, $v\notin\mathrm{in}(\overline{P}_\ell)$, a contradiction.
 \end{proof}
Hence, $(T,L_1,c,r)$ has visible width at most $k$ and we can find an optimum solution in polynomial time by \cref{cor:solve_DTAP_bounded_viwidth}.  If the optimum solution has cost at most $c(x_1)=\sum_{\ell\in L}\sum_{\ell'\in\sigma_1(\ell)} c(\ell')\cdot x_\ell$ (by \cref{lemma:splitting_LP_solution}), then we have found the desired solution. Otherwise, the set of $\zeta_1$-covered arcs that we initially contracted to obtain \eqref{eq:no_covered_arcs}, together with the splitting $\sigma_1$ that we applied to get from our initial LP solution $x$ to $x_1$, yields a violated visibly $k$-wide modification inequality.
\end{proof}
\begin{lemma}\label{lemma:x_2}
Let $\sigma_2\coloneqq \sigma_{\overleftarrow{L}\cup L_{rest}}\circ \sigma^{**}\circ\sigma^{*}$, let $x_2\coloneqq \mathrm{split}(x^{**},\sigma_{\overleftarrow{L}\cup L_{rest}})=\mathrm{split}(x,\sigma_2)$ and let $L_2\coloneqq\mathrm{supp}(x_2)$.
Then $(T,L_2,c,r)$ is a willow. In particular, we can, in polynomial time, compute a solution of cost at most $c(x_2)$.
\end{lemma}
\begin{proof}
	By \cref{cor:apex_in_core}, applied with $\sigma=\sigma_{\rm{id}}$, we know that every link in $\overrightarrow{L}\cap\mathrm{supp}(x^{**})$ is an $X$-cross-link or an up- or down-link. Hence, every link in $\mathrm{supp}(x_2)$ is an up-link, a down-link, or a $(W^*\cup X)$-cross-link.
To establish that $(T,L_2,c,r)$ is a willow, it suffices to show that every vertex in $(W^*\cup X)$ is up- or down-independent with respect to $L_2$. For vertices in $W^*$, this follows from \cref{lemma:up_down_independent_first_phase}, and \cref{lemma:splitting_decreases_L_up_down,lemma:up-down-independent}. Recall that $X=\bigcup_{C\in\mathring{\mathcal{C}}} V(C)$. We show that for $C\in\mathring{\mathcal{C}}_{up}$, all vertices in $V(C)$ are down-independent. Analogously, one can show that for $C\in\mathring{\mathcal{C}}_{down}$, all vertices in $V(C)$ are up-independent.

Let $C\in\mathring{\mathcal{C}}_{up}$. By \ref{property:no_entering_at_rC}, $\mathrm{supp}(x^{**})\cap L_{r_C}^\uparrow=\emptyset$, so also $\mathrm{supp}(x_2)\cap L_{r_C}^\uparrow=\emptyset$ by \cref{lemma:splitting_decreases_L_up_down}. By \cref{lemma:up-down-independent}, $r_C$ is down-independent. Next, let $v\in V(C)\setminus \{r_C\}$. Then $a_v\in A(C)$. Assume towards a contradiction there were $\ell=(u,w)\in L_2$ with $\covr(\ell)\cap A_v\cap A_{down}\ne\emptyset$ and $\covr(\ell)\not\subseteq A_v$. The first property implies $u\in U_v\setminus\{v\}$, the second property tells us that $\mathrm{apex}(\ell)$ is a strict ancestor of $v$. In particular, $a_v\in\cov(\ell)$. By \cref{cor:apex_in_core}, $\mathrm{apex}(\ell)\in V(C)$. As $C\in\mathring{\mathcal{C}}_{up}$, $a_v$ is an up-arc, so $a_v\in\covw(\ell)$, implying $\ell\in \overleftarrow{L}$. But this implies that $\mathrm{apex}(\ell)=w$ is an endpoint of $\ell$ because all links in $\overleftarrow{L}$ were split at their apices. Hence, $\cov(\ell)\setminus A_v$ consists of the up-arcs on the $v$-$\mathrm{apex}(\ell)$-path, implying $\cov(\ell)\setminus A_v\subseteq \covw(\ell)$ and $\covr(\ell)\subseteq A_v$, contradicting our assumptions.  
\end{proof}
Before describing the splitting leading to our third solution, we make the following observation:
\begin{lemma}\label{lemma:unique_core}
Let $\ell\in\overrightarrow{L}\cap \mathrm{supp}(x^{**})$. Then there exists a unique core $C\in\mathring{\mathcal{C}}$ such that $\covr(\ell)\cap A(C)\ne \emptyset$.
\end{lemma}
\begin{proof}
By definition of $\overrightarrow{L}$, there exists at least one core with this property. Assume towards a contradiction that there are two distinct cores $C_1$ and $C_2$ such that $\covr(\ell)\cap A(C_1)\ne \emptyset$ and $\covr(\ell)\cap A(C_2)\ne \emptyset$. Let $w\coloneqq \mathrm{apex}(\ell)$. By \cref{cor:apex_in_core}, we know that $w\in V(C_1)\cap V(C_2)$. If $w=r$, then $w=r_{C_1}=r_{C_2}$. Otherwise, if $w\ne r$, $a_w$ is contained in at most one of the sets $A(C_1)$ or $A(C_2)$, assume w.l.o.g.\ that $a_w\notin A(C_1)$. As $w\in V(C_1)$, but $a_w\notin A(C_1)$, we must, again, have $w=r_{C_1}$.

Let $a\in \covr(\ell)\cap A(C_1)$ and let $x$ be the bottom vertex of $a$. As $a\in\covr(\ell)$ and $\mathrm{apex}(\ell)=w=r_{C_1}$, $\ell$ covers all arcs on the $x$-$r_{C_1}$-path in $C_1$, including the arc incident to $r_{C_1}$. By \ref{property:no_covering_two_cores}, $r_{C_1}=\mathrm{apex}(\ell)$ is an endpoint of $\ell$, so $\ell$ is an up-link or a down-link. As $\covr(\ell)\cap A(C_1)\ne \emptyset$ and $\covr(\ell)\cap A(C_2)\ne \emptyset$, $C_1$ and $C_2$ must both be down-cores, or both be up-cores, respectively. However, this contradicts $w\in V(C_1)\cap V(C_2)$ because two distinct down-cores/ up-cores are vertex-disjoint.
\end{proof}
For $\ell\in \overrightarrow{L}\cap \mathrm{supp}(x^{**})$, let $C_\ell$ be the unique core with $\covr(\ell)\cap A(C)\ne \emptyset$. We define a splitting $\sigma'_3$ of $L$ as follows:
\begin{itemize}
	\item Let $\ell=(u,v)\in \overrightarrow{L}\cap \mathrm{supp}(x^{**})$ with $C_{\ell}\in \mathring{\mathcal{C}}_{up}$. Let $w$ be the lowest vertex from $V(C_\ell)$ on the $\mathrm{apex}(\ell)$-$v$-path in $T$ and define $\sigma'_3(\ell)\coloneqq \{(u,\mathrm{apex}(\ell)),(\mathrm{apex}(\ell),w),(w,v)\}$.
	\item Let $\ell=(u,v)\in \overrightarrow{L}\cap \mathrm{supp}(x^{**})$ with $C_{\ell}\in \mathring{\mathcal{C}}_{down}$. Let $w$ be the lowest vertex from $V(C_\ell)$ on the $u$-$\mathrm{apex}(\ell)$-path in $T$ and define $\sigma'_3(\ell)\coloneqq \{(u,w),(w,\mathrm{apex}(\ell)),(\mathrm{apex}(\ell),v)\}$.
	\item Let $\ell=(u,v)\in L_{cross}\cup L_{rest}$. Define $\sigma'_3(\ell)=\{(u,\mathrm{apex}(\ell)),(\mathrm{apex}(\ell),v)\}$.
	\item For every other link $\ell$, define $\sigma'_3(\ell)=\{\ell\}$.
\end{itemize}
\begin{lemma}\label{lemma:x_3}
Let $\sigma_3\coloneqq \sigma'_3\circ \sigma^{**}\circ \sigma^{*}$ and let $x_3\coloneqq \mathrm{split}(x^{**},\sigma'_3)=\mathrm{split}(x,\sigma_3)$.
Let $L_3\coloneqq\mathrm{supp}(x_3)$ and let $M\in\{0,1\}^{A\times L_3}$ denote the arc-link-coverage matrix of $(T,L_3,c,r)$.
Then $M$ is TU. In particular, we can, in polynomial time, compute a solution to $(T,L_3,c,r)$ of cost at most $c(x_3)$.
\end{lemma}
Before we prove \cref{lemma:x_3}, we first make the following observations:
\begin{proposition}\label{prop:no_up_or_down_link}
Let $\ell\in L_3$ such that $\ell$ is neither an up- nor a down-link. Then $\ell\in \overleftarrow{L}\cap \mathrm{supp}(x^{**})$.
\end{proposition}
\begin{proof}
As $\ell\in L_3$, there is $\ell'\in\mathrm{supp}(x^{**})$ with $\ell\in\sigma'_3(\ell')$. For every link $\ell''\in \mathrm{supp}(x^{**})\cap (\overrightarrow{L}\cup L_{cross}\cup L_{rest})$, $\sigma'_3(\ell'')$ consists of up- and down-links only. Hence, $\ell'\in \overleftarrow{L}\cap \mathrm{supp}(x^{**})$. As $\sigma'_3(\ell')=\{\ell'\}$, we have $\ell=\ell'$.
\end{proof}
\begin{proposition}\label{prop:only_cover_core}
Let $\ell\in L_3\cap \overrightarrow{L}$. Then there exists a unique core $C\in\mathring{\mathcal{C}}$ such that $\covr(\ell)\cap A(C)\ne \emptyset$. Moreover, $\covr(\ell)\subseteq A(C)$ and $\ell$ is an up- or down-link.
\end{proposition}
\begin{proof}
Let $\ell\in L_3\cap \overrightarrow{L}$ and let $C$ be a core with $\covr(\ell)\cap A(C)\ne \emptyset$. Let $\ell'\in\mathrm{supp}(x^{**})$ such that $\ell\in\sigma'_3(\ell')$. Then $\ell$ is a shadow of $\ell'$, so $\covr(\ell)\subseteq \covr(\ell')$. In particular, $\covr(\ell')\cap A(C)\ne \emptyset$, $\ell'\in \overrightarrow{L}$ and by \cref{lemma:unique_core}, $C=C_{\ell'}$ is the unique core $C'$ such that $\covr(\ell')\cap A(C')\ne \emptyset$. Hence, $C$ is also the unique core $C'$ such that $\covr(\ell)\cap A(C')\ne \emptyset$. Let $w$ be as in the definition of $\sigma'_3(\ell')$. The second part of the statement follows from the facts that $\mathrm{apex}(\ell')\in V(C)$ by \cref{cor:apex_in_core} and that we must have $\ell=(\mathrm{apex}(\ell'),w)$, if $C$ is an up-core, and $\ell=(w,\mathrm{apex}(\ell'))$, if $C$ is a down-core, because this is the only link in $\sigma'_3(\ell')$ covering part of $A(C)$.
\end{proof}
We further introduce the following notation.
\begin{definition}
	We call a matrix $M\in\mathbb{R}^{I\times J}$ a \emph{block diagonal matrix with blocks $M[I_s,J_s]$, $s=1,\dots,t$} if $I=\dot{\bigcup}_{s=1}^t I_s$ is a partition of $I$, $J=\dot{\bigcup}_{s=1}^t J_s$ is a partition of $J$ and for $i\in I_{s_1}$ and $j\in J_{s_2}$, $M_{ij}\ne 0$ implies $s_1=s_2$, i.e., non-zero entries can only occur within one block.
\end{definition}
This definition may differ from notions used in the literature in that we do not require the blocks to be square or have the same size.
\begin{proof}[Proof of \cref{lemma:x_3}]
	We first establish the following claim:
	\begin{claim}
		$M$ is a block diagonal matrix with blocks $M[A(C),\{\ell\in L_3\colon\covr(\ell)\cap A(C)\ne \emptyset\}]$ for $C\in\mathring{\mathcal{C}}$ and $M[A\setminus\bigcup_{C\in\mathring{\mathcal{C}}}A(C),L_3\setminus \overrightarrow{L}]$.
	\end{claim}
\begin{proof}[Proof of claim]
By definition, the sets $A(C)_{C\in\mathring{\mathcal{C}}}$ are pairwise disjoint. This shows that the arcs sets indexing the rows of the blocks form a partition of $A$.
 
By \cref{prop:only_cover_core}, the link sets indexing the columns of the blocks form a partition of $L_3$.

Finally, we verify the block structure of $M$. First, let $C\in\mathring{\mathcal{C}}$, let $a\in A(C)$ and let $\ell'\in L_3$ such that $M_{a,\ell'}\ne 0$. Then $a\in\covr(\ell')$, so $\ell'\in\{\ell\in L_3\colon \covr(\ell)\cap A(C)\ne \emptyset\}$.

Next, let $a\in A\setminus\bigcup_{C\in\mathring{\mathcal{C}}}A(C)$. By \cref{prop:only_cover_core}, we have $a\notin\covr(\ell)$, and, hence, $M_{a,\ell}=0$ for every $\ell\in L_3\cap \overrightarrow{L}$. Thus, if $M_{a,\ell}\ne 0$ for some $\ell\in L_3$, then $\ell\in L_3\setminus \overrightarrow{L}$. \end{proof}
To show that $M$ is TU, it suffices to establish total unimodularity of each of the blocks separately. For the blocks of the form $M[A(C),\{\ell\in L_3\colon\covr(\ell)\cap A(C)\ne \emptyset\}]$ with $C\in\mathring{\mathcal{C}}$, this follows from \cref{theorem:unimodularity} and the fact that $\{\ell\in L_3\colon\covr(\ell)\cap A(C)\ne \emptyset\}$ only consists of up- or down-links by \cref{prop:only_cover_core}. 

The block $M[A\setminus\bigcup_{C\in\mathring{\mathcal{C}}}A(C),L_3\setminus \overrightarrow{L}]$ corresponds to the instance $(T',L',c',r')$ obtained from the tuple $(T,L_3\setminus \overrightarrow{L},c,r)$ by contracting all arcs in $\bigcup_{C\in\mathring{\mathcal{C}}}A(C)$. Each (super-)vertex $v\in V(T')$ corresponds to a set $Y_v\subseteq V(T)$ of original vertices. Let $R\coloneqq \{v\in V(T')\colon \exists C\in\mathring{\mathcal{C}}\colon r_C\in Y_v\}$.
\begin{claim}
Let $\ell'\in L'$ such that $\ell'$ is neither an up- nor a down-link. Then $\ell'$ is an $R$-cross-link.
\end{claim}
\begin{proof}[Proof of claim]
By \cref{prop:no_up_or_down_link}, every such link $\ell'\in L'$ corresponds to a link $\ell\in \overleftarrow{L}\cap \mathrm{supp}(x^{**})$. By definition of $\overleftarrow{L}$, let $C\in\mathring{\mathcal{C}}$ such that $\covw(\ell)\cap A(C)\ne \emptyset$. By \cref{cor:apex_in_core}, $\mathrm{apex}(\ell)\in V(C)$ and all of $V(C)$, including $r_C$, is contracted into the same super-vertex, which becomes $\mathrm{apex}(\ell')$.
\end{proof} 
We show that every vertex in $R$ is up- or down-independent, establishing that $(T',L',c',r')$ is a willow and concluding the proof by \cref{theorem:unimodularity}.

 Let $v\in R$. $Y_v$ is the vertex set of a connected subgraph of $T$, consisting of a collection of cores. In particular, there is a unique vertex in $Y_v$ that is closest to the root $r$ of $T$, and it is the root $r_C$ of a core $C$. We may assume without loss of generality that $C\in\mathring{\mathcal{C}}_{up}$; the case where $C\in \mathring{\mathcal{C}}_{down}$ can be handled analogously. We show that no link in $L'$ points out of $T'_v$, establishing that $v$ is down-independent by \cref{lemma:up-down-independent}. 
Assume towards a contradiction there were $\ell'=(u',w')\in L'$ pointing out of $T'_v$, i.e., $u'\in V(T'_v)\setminus \{v\}$ and $w'\in V(T')\setminus V(T'_v)$. Let $\ell'$ correspond to the link $\ell=(u,w)\in L_3\setminus \overrightarrow{L}$ with $u\in Y_{u'}$ and $w\in Y_{w'}$. As $Y_{u'}$, $Y_v$ and $Y_{w'}$ are vertex sets of connected, vertex-disjoint subgraphs of the tree $T$ and $r_C$ is the vertex of $Y_v$ closest to the root of $T$, $Y_{u'}\subseteq V(T_{r_C})\setminus \{r_C\}$ and $Y_{w'}\subseteq V(T)\setminus V(T_{r_C})$. Hence, $\ell$ points out of $T_{r_C}$, contradicting \ref{property:no_entering_at_rC}.
\end{proof}
We are now ready to finally prove \cref{theorem:partial_separation_oracle}, which we restate for convenience.
\theoremoracle*
\begin{proof}
Let $\epsilon \coloneqq \frac{\min\{1,\bar{\epsilon}\}}{7}$. Define the constants $\gamma$, $\zeta_1$, $\zeta_2$ and $k$ as in \Cref{sec:proof_weak_dream} and let $k(\bar{\epsilon},\Delta)\coloneqq k$. Let $(T,L,c,r)$ and $x$ arise from $(\bar{T},\bar{L},\bar{c},\bar{r})$ and $\bar{x}$ by contracting the set $\bar{A}$ of $\zeta_1$-covered arcs. Note that $x$ is a solution to \eqref{eq:WDTAP_LP} of cost $c(x)\le \bar{c}(\bar{x})$ satisfying \eqref{eq:no_covered_arcs}. We apply \cref{theorem:first_phase_splitting} and \cref{theorem:second_phase_splitting} to obtain splittings $\sigma^*$ and $\sigma^{**}$ and solutions $x^*$ and $x^{**}$ to \eqref{eq:WDTAP_LP} for $(T,L,c,r)$. We define $L_{cross}$, $\overrightarrow{L}$, $\overleftarrow{L}$ and $L_{rest}$ as in the beginning of this section.
We apply \cref{lemma:x_1} to, in polynomial time, either compute a solution $S_1$ to $(T,L,c,r)$ of cost at most $c(x_1)$ or a violated visibly $k$-wide modification inequality for $(T,L,c,r)$. In the latter case, the corresponding splitting, together with $\bar{A}$, gives rise to a violated visibly $k$-wide modification inequality for $(\bar{T},\bar{L},\bar{c},\bar{r})$ and $\bar{x}$. Hence, we may assume in the following that we have 
 found a solution $S_1$ to $(T,L,c,r)$ of cost $c(S_1)\le c(x_1)$. We further apply \cref{lemma:x_2,lemma:x_3} to, in polynomial time, compute solutions $S_2$ and $S_3$ to $(T,L,c,r)$ of cost $c(S_2)\le c(x_2)$ and $c(S_3)\le c(x_3)$, respectively. 
 
 For $L'\in \{L_{cross},\overleftarrow{L},\overrightarrow{L},L_{rest}\}$, we define $C^{**}(L')\coloneqq \sum_{\ell\in L} c(\ell)\cdot x^{**}(\ell)$. By \eqref{eq:partition_supp}, we know that
 \[c(x^{**})=C^{**}(L_{cross})+C^{**}(\overleftarrow{L})+C^{**}(\overrightarrow{L})+C^{**}(L_{rest}).\]
By \cref{lemma:cost_increase_concatenate_splittings} and because $|\sigma_{L\setminus L_{rest}}(\ell)|=2$ for $\ell\in L\setminus L_{rest}$ and $|\sigma_{L\setminus L_{rest}}(\ell)|=1$ for $\ell\in L_{rest}$, we have
\[c(x_1)\le c(x^{**})+C^{**}(L_{cross})+C^{**}(\overleftarrow{L})+C^{**}(\overrightarrow{L}).\]
By \cref{lemma:cost_increase_concatenate_splittings} and because $|\sigma_{\overleftarrow{L}\cup L_{rest}}(\ell)|=2$ for $\ell\in \overleftarrow{L}\cup L_{rest}$ and $|\sigma_{\overleftarrow{L}\cup L_{rest}}(\ell)|=1$ for $\ell\notin \overleftarrow{L}\cup L_{rest}$, we have
\[c(x_2)\le c(x^{**})+C^{**}(\overleftarrow{L})+C^{**}(L_{rest}).\]
By \cref{lemma:cost_increase_concatenate_splittings} and because $|\sigma'_3(\ell)|=3$ for $\ell\in \overrightarrow{L}\cap\mathrm{supp}(x^{**})$, $|\sigma'_3(\ell)|=2$ for $\ell\in L_{cross}\cup L_{rest}$ and $|\sigma'_3(\ell)|=1$ for every other link $\ell$, we have
\[c(x_3)\le c(x^{**})+C^{**}(L_{cross})+2\cdot C^{**}(\overrightarrow{L})+C^{**}(L_{rest}).\]
Let $S^*$ be the best one among the three solutions $S_1$, $S_2$ and $S_3$.
Then
\begin{align*}
c(S^*)\le&\frac{1}{4}\cdot c(x_1)+\frac{1}{2}\cdot c(x_2)+\frac{1}{4}\cdot c(x_3)\le c(x^{**})+\frac{1}{2}\cdot C^{**}(L_{cross})+\frac{3}{4}\cdot C^{**}(\overleftarrow{L})+\frac{3}{4}\cdot C^{**}(\overrightarrow{L})+\frac{3}{4}\cdot C^{**}(L_{rest})\\
&\le \frac{7}{4}\cdot c(x^{**})\le \frac{7}{4}\cdot (1+\epsilon)^2\cdot c(x)\le \frac{7}{4}\cdot (1+\epsilon)^2\cdot \bar{c}(\bar{x}),
\end{align*}
where the second-to-last inequality follows from \ref{property:cost_bound}. By \cref{lemma:solution_for_zeta_1_covered_arcs}, we can extend $S^*$ to a solution $S$ for $(\bar{T},\bar{L},\bar{c})$ of cost at most \[\left(\frac{7}{4}\cdot (1+\epsilon)^2+\epsilon\right)\cdot \bar{c}(\bar{x})\le\frac{7}{4}\cdot (1+3\cdot\epsilon+\epsilon^2)\cdot \bar{c}(\bar{x})\le \frac{7}{4}\cdot (1+4\cdot\epsilon)\cdot \bar{c}(\bar{x})=\left(\frac{7}{4}+\bar{\epsilon}\right)\cdot \bar{c}(\bar{x}).\]
Observing that $\frac{7}{4}=1.75$ concludes the proof.
\end{proof}

\section*{Acknowledgments}
This work originated from a collaboration that included Siyue Liu and R. Ravi, whose early contributions are gratefully acknowledged. We also thank the anonymous reviewers who helped improve the presentation of our paper. This work was supported in part by EPSRC grant EP/X030989/1.

\section*{Data Availability Statement}
No data are associated with this article. Data sharing is
not applicable to this article.

\bibliographystyle{siamplain}
\bibliography{sample}


\appendix
\section{Appendix}\label{sec:appendix}

\subsection{Multi 2-TAP reduces to DTAP}\label{subsec:multi-2-tap}
In this section we prove that the weighted multi 2-TAP problem reduces to WDTAP. In weighted multi 2-TAP, we are given an undirected tree $T = (V,E)$, and a set of links $L \subseteq {V \choose 2}$ with positive costs. For every $e\in E$, denote by $S_e \subseteq V$ a shore of the fundamental tree cut induced by $e$. A multiset (repetition allowed) of undirected links $F \subseteq L$ is a \textit{$k$-covering} of $T$ if every fundamental cut is covered at least $k$ times by $F$, that is $|\delta_F(S_e)|\geq k$ for all $e \in E$. The cost of a multi-set of links is the sum of the costs of links in that set, weighted by multiplicity.  

\begin{problem}[Weighted Multi 2-TAP]
    Given an undirected tree $T = (V,E)$ and a collection of links $L \subseteq {V \choose 2}$ with positive costs, find the cheapest 2-covering of $T$.
\end{problem}

The reduction is easiest to explain by introducing an intermediate problem, which we call bi-directed tree cover, and which is equivalent to WDTAP.  

\begin{problem}[Bi-directed Tree Cover]
    Given an undirected tree $T = (V,E)$ and a collection of directed links $L \subseteq V \times V$ with positive costs, choose a cheapest set of links $F$ so that $|\delta_{F}^+(S_e)| \geq 1$, and $|\delta_F^-(S_e)| \geq 1$ for all tree edges $e \in E$.
\end{problem}

In other words, the cut induced by each tree edge must be crossed in both directed by the solution $F$. Bi-directed tree cover is easily seen to be equivalent to WDTAP: on the one hand, it can be reduced to WDTAP by subdividing every tree edge and orienting them in opposite directions. On the other hand, WDTAP can be reduced to bi-directed tree cover by adding a zero-cost directed link $\ell_a$ parallel to each tree arc $a \in A$ in the same direction as $a$, and making the tree undirected.

We now show that weighted multi $2$-TAP can be reduced to bi-directed tree cover. We replace every link $\ell \in L$ by two directed links $\ell^+$ and $\ell^-$ in opposite directions, each having the same cost as $\ell$. Clearly, every bi-directed tree cover solution is a feasible solution to the weighted multi $2$-TAP instance with the same cost. The following proposition shows that every weighted multi $2$-TAP solution can be oriented into a bi-directed tree cover solution of the same cost. 
\begin{proposition}
    Given a 2-covering $F \subseteq L$ of a tree $T = (V,E)$, there is an orientation $\vec{F}$ such that $|\delta^+_{\vec{F}}(S_e)|,|\delta^-_{\vec{F}}(S_e)|\geq 1$ for every fundamental cut shore $S_e$.
\end{proposition}
\begin{proof}
    Let $\mathcal{S}:=\{(S_e) \cup (V \setminus S_e) \mid e\in E\}$ be all the shores of fundamental cuts. Let $\vec{H}$ be an arbitrary orientation of the $2$-covering $F$. We seek an integral solution to the following submodular flow polyhedron:
    \[
    |\delta_{\vec{H}}^+(S)|-x(\delta_{\vec{H}}^+(S))+x(\delta_{\vec{H}}^-(S))\geq 1, \forall S\in \mathcal{S}.
    \]
    This is indeed a submodular flow because $\mathcal{S}$ is cross-free and thus trivially a crossing family. We can take $x_\ell = \frac{1}{2}$ for every $\ell \in \vec{H}$ to be a fractional feasible solution, and thus by the integrality of the submodular flow polyhedron, there is an integral feasible solution $x$. Flipping $\ell$ if and only if $x_\ell=1$, yields the desired orientation $\vec{F}$.
\end{proof}

\subsection{Hardness of DTAP}\label{subsec:hardness}
In this section, we prove the following result on the hardness of DTAP.
\begin{proposition}
    DTAP is NP-hard and APX-hard, even in the unweighted setting.
\end{proposition}
\begin{proof}
    We prove NP-hardness using the same reduction as the one for CSTA in \cite{frederickson1981approximation}. We reduce $3$-dimensional matching ($3$DM) to unweighted DTAP.

Let $M\subseteq W\times X\times Y$ be an instance of $3$DM with $|M|=p$, $W=\{w_i\mid i=1,...,q\}$, $X=\{x_i\mid i=1,...,q\}$, $Y=\{y_i\mid i=1,...,q\}$. We define an instance of DTAP as follows. Let $V=\{r\}\cup\{w_i,x_i,y_i\mid i=1,...,q\}\cup \{a_{ijk},a'_{ijk}\mid (w_i,x_j,y_k)\in M\}$.

$A_T=\{(r,x_i),(r,w_i),(y_i,r)\mid i=1,...,q\}\cup \{(a_{ijk},w_i),(w_i,a'_{ijk})\mid (w_i,x_j,y_k)\in M\}$.

$A_L=\{(x_j,a_{ijk}),(a'_{ijk},a_{ijk}),(a'_{ijk},y_k)\mid (w_i,x_j,y_k)\in M\}$.

We claim that there exists a $3$DM of size $q$ if and only if the minimum size of a DTAP solution is $p+q$. Indeed, notice that there are $2p+2q$ leaves in $T$, where each leaf needs at least one link to cover it. Thus, the minimum size of a DTAP solution is at least $p+q$. On the one hand, if there is a $3$DM $M'$ of size $q$, we obtain a DTAP solution $L':=\{(x_j,a_{ijk}),(a'_{ijk},y_k)\mid (w_i,x_j,y_k)\in M'\}\cup \{(a'_{ijk},a_{ijk})\mid (w_i,x_j,y_k)\in M\setminus M'\}$ whose size is $2q+(p-q)=p+q$. On the other hand, if there is a DTAP solution $L'$ of size $p+q$, by the previous argument, $L'$ forms a perfect matching on the leaves. Let $M'$ be the edges $(w_i,x_j,y_k)\in M$ such that $a_{ijk}$ is matched to $x_j$. Clearly, $a_{ijk}$ is matched to $x_j$ if and only if $a'_{ijk}$ is matched to $y_k$. Thus, $|M'|=q$ and it intersects every node in $X$ or $Y$ exactly once. Since for every $i$, $(r,w_i)$ is covered by the links from $a'_{ijk}$ to $y_k$, $M'$ intersects every node in $W$, and thus intersects every node in $W$ exactly once. Therefore, $M'$ is indeed a $3$DM of size $q$. 
\end{proof}

Following the methods in~\cite{DBLP:journals/siamcomp/KortsarzKL04}, the above proof can be extended to show APX-hardness via a reduction from bounded degree 3DM.

\subsection{Lower Bound on the Integrality Gap of DTAP}\label{subsection:integralitygap}
We show the following lower bound on the integrality gap of the natural set covering relaxation for DTAP.

\begin{proposition}
    The integrality gap of the set covering formulation for DTAP given in (\ref{eq:WDTAP_LP}) is at least $\frac{6}{5}$.
\end{proposition}
\begin{proof}

Consider the following unweighted DTAP instance whose constraint matrix corresponds to a $5$-cycle:

\begin{figure}[h]
\centering
\begin{tikzpicture}[mynode/.style={draw, fill, circle, minimum size = 2mm, inner sep = 0pt}, myarc/.style={thick, arrows = {-Stealth[scale=1.5]}},mylink/.style={thick, arrows = {-Stealth[scale=1.5]}, black!70!white, dashed}, covered/.style={line width = 2mm, draw opacity = 0.5}]

\node[mynode] (v1) at (-1,1){};
\node[mynode] (v2) at (-1,0){};
\node[mynode] (v3) at (-1,-1){};
\node[mynode] (v6) at (1,1){};
\node[mynode] (v7) at (1,0){};
\node[mynode] (v8) at (1,-1){};

\draw[myarc] (v2) to (v1);\label{7}
\draw[myarc] (v3) to (v2);\label{6}
\draw[myarc] (v2) to (v7);\label{2}
\draw[myarc] (v7) to (v6);\label{3}
\draw[myarc] (v8) to (v7);\label{4}

\draw[mylink] (v1) to [bend right = 30] (v3);
\draw[mylink] (v6) to (v2);
\draw[mylink] (v7) to (v3);
\draw[mylink] (v1) to (v8);
\draw[mylink] (v6) to [bend left = 30] (v8); 

\end{tikzpicture}
\caption{Choosing $x_\ell = \frac{1}{2}$ for all $\ell \in L$ yields a solution of cost $\frac{5}{2}$, while the smallest integral solution has cost 3.}
\label{fig:DTAPgap}
\end{figure}
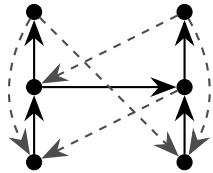

\end{proof}

\end{document}